\documentclass[11pt]{article}

\usepackage{amsmath}
\usepackage{amsthm}
\usepackage{amsfonts}

\usepackage[
backend=biber,
maxbibnames=99,
style=alphabetic,url=false,doi=false
]{biblatex}
\usepackage{subcaption}

\usepackage[a4paper, total={6.5in, 9.5in}]{geometry}

\usepackage{braket}

\usepackage{quantikz}
\usetikzlibrary{external}
\usepackage[noblocks]{authblk}

\usepackage{graphicx}
\usepackage{hyperref}

\usepackage[ruled,vlined]{algorithm2e}
\usepackage[nopatch]{microtype}
\usepackage{stmaryrd}

\usepackage[ruled,vlined]{algorithm2e}
\usepackage[nopatch]{microtype}
\usepackage{booktabs}
\usepackage{stfloats} 
\usepackage{caption}
\renewcommand{\mod}{\mathbin{\;\mathrm{mod}\;}}

\usepackage[normalem]{ulem}

\usepackage[mode=build]{standalone}

\newtheorem{theorem}{Theorem}
\numberwithin{theorem}{section}

\newtheorem{definition}[theorem]{Definition}

\newtheorem{remark}[theorem]{Remark}

\newtheorem{lemma}[theorem]{Lemma}
\newtheorem*{lemma*}{Lemma}

\newtheorem{proposition}[theorem]{Proposition}

\newtheorem{corollary}[theorem]{Corollary}

\newcommand{\maj}{\textrm{maj}} 
\newcommand{\MAJ}{\textrm{MAJ}} 
\newcommand{\UMA}{\textrm{UMA}} 

\title{Measurement-based uncomputation of \\ quantum circuits for modular arithmetic}

\author[1,2]{Alessandro Luongo}
\author[3]{Antonio Michele Miti}
\author[4]{Varun Narasimhachar}
\author[5,*]{Adithya Sireesh}

\affil[1]{\small{Centre for Quantum Technologies, National University of Singapore, Singapore}}
\affil[2]{Inveriant Pte. Ltd., Singapore}
\affil[3]{Department of Mathematics, Sapienza Università di Roma, Rome, Italy} 
\affil[4]{A*STAR Quantum Innovation Centre (Q.InC), Institute of High Performance Computing (IHPC), \:\: \:\: Agency for Science, Technology and Research (A*STAR), Singapore.}
\affil[5]{School of Informatics, University of Edinburgh, Scotland, United Kingdom}
\affil[*]{Corresponding author: \url{asireesh@ed.ac.uk}}

\begin{document}

\maketitle
 

\begin{abstract}
Measurement-based uncomputation (MBU) is a technique used to perform probabilistic uncomputation of quantum circuits. We formalize this technique  for the case of single-qubit registers, and we show applications to modular arithmetic. First, we present formal statements for several variations of quantum circuits performing non-modular addition: controlled addition, addition by a constant, and controlled addition by a constant. We do the same for subtraction and comparison circuits. This addresses gaps in the current literature, where some of these variants were previously unexplored. Then, we shift our attention to modular arithmetic, where again we present formal statements for modular addition, controlled modular addition, modular addition by a constant, and controlled modular addition by a constant, using different kinds of plain adders and combinations thereof. We introduce and prove a ``MBU lemma'' in the context of single-qubit registers, which we apply to all aforementioned modular arithmetic circuits. Using MBU, we reduce the Toffoli count and depth by $10\%$ to $15\%$ for modular adders based on the architecture of \cite{vedral1996quantum}, and by almost $25\%$ for modular adders based on the architecture of \cite{beauregard2002circuit}. Our results have the potential to improve other circuits for modular arithmetic, such as modular multiplication and modular exponentiation, and can find applications in quantum cryptanalysis.
\end{abstract}

\section{Introduction}
Consider a quantum circuit $U_f$ computing a Boolean function
$ f: \{0,1\}^n \to \{0,1\}^m$ between bit strings. 
When $f$ is invertible (meaning that $f$ is injective and $n=m$) we say that  $U_f$ computes $f$ \emph{in-place} if it performs the mapping
\begin{equation} U_f: \ket{x}_n \mapsto \ket{f(x)}_m
\end{equation}
for any bit string $x\in \{0,1\}^n$, where we use the subscripts to denote the number of qubits in a quantum register. An in-place implementation is not always possible. In the most general case, one needs to compute $f$ \emph{out-of-place}, meaning that the circuit $U_f$ realises the following mapping
\begin{equation}
    U_f: \ket{x}_n\ket{0}_m \mapsto \ket{x}_n \ket{f(x)}_m~.
\end{equation}

A realistic implementation of the above gate will, in principle, involve a number $k$ of auxiliary qubits which get mapped in a \emph{garbage} state $\ket{g(x)}$ as follows:
\begin{equation}
    \widetilde{U}_f:\ket{x}_n\ket{0}_m\ket{0}_k \mapsto \ket{x}_n \ket{f(x)}_m\ket{g(x)}_k.
\end{equation}

The auxiliary qubits are entangled with the other registers and must be uncomputed, i.e. restored to their original value, in order to get a reversible out-of-place implementation of $f$. Such uncomputation can be achieved in several ways. 
A first solution, due to \cite{Bennett1973}, simply prescribes copying the output $\ket{f(x)}_m$ into an auxiliary $m$-qubit register via a set of $m$ CNOT gates and then apply $\widetilde{U}_f^\dagger$. Therefore, one can implement $U_f$ as in figure~\ref{fig:bennet-uncomputation}. 
\begin{figure}[h!]
\centering
    \includegraphics[width=0.4\textwidth]{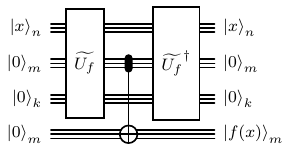}
    \caption{Bennet's garbage uncomputation.}\label{fig:bennet-uncomputation}
\end{figure}
This strategy doubles the resource cost of the reversible computation of $f$, also requiring $m$ more ancillary qubits. Alternatively, one could look for an out-of-place implementation of the garbage function $U_g$, seeing $g$ as a function of $x$ and $f(x)$.
When this unitary is available, one can implement $U_f$ as in figure~\ref{fig:garbage-uncomputation}.
 \begin{figure}[h!]
 \centering
     \includegraphics[width=0.4\textwidth]{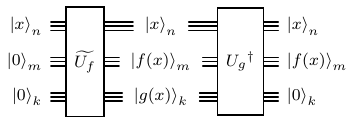}
     \caption{Garbage gate uncomputation.}\label{fig:garbage-uncomputation}
 \end{figure}

In principle, invoking the garbage circuit $U_g$ can be quite resource intensive. Therefore, it is necessary to develop techniques to optimize the synthesis of quantum circuits and their uncomputation.

\paragraph{Arithmetic Circuits.} Effective implementations of binary arithmetic are fundamental to many quantum algorithms, such as quantum algorithms for factoring~\cite{shor1994algorithms, gidney2021factor} and discrete logarithm~\cite{DBLP:journals/iacr/Ekera18}, quantum walks~\cite{childs2002quantum}, and for algorithms that construct oracles for Grover's search and quantum phase estimation, e.g.~\cite{berry2019qubitization,doriguello2022quantum,aggarwal2017quantum}. Arithmetic operations like (modular) addition, subtraction, and comparison (by a constant) have been extensively studied, resulting in numerous implementation strategies, each with its own tradeoffs. Many circuits for arithmetic are based on the plain addition circuit defined by the operation $\ket{x}_n\ket{y}_n \xmapsto{\mathsf{Q_{ADD}}} \ket{x}_n\ket{x+y}_{n+1}$ for two $n$-qubit quantum registers. Other primitives can be derived by composing these circuits for addition. However, there are often more efficient ways to build arithmetic primitives. Broadly, these quantum circuits are designed for both quantum-quantum arithmetic functions, i.e., circuits of the form $\ket{x}\ket{y} \xmapsto{\mathsf{U}} \ket{x}\ket{f(x,y)}$, and quantum-classical arithmetic functions, i.e., circuits of the form $\ket{x} \xmapsto{\mathsf{U}(a)} \ket{x}\ket{f(x,a)}$.

In the computation of our function $f$, intermediate values may need to be appropriately uncomputed. For example, any \emph{ripple carry} addition algorithm, analogously to the \emph{schoolbook} addition method, requires computing the carry bitwise starting from the least significant bits of the two addends.

However, this process leaves carry bits in memory, which need to be uncomputed from the rest of the registers to ensure the correctness of the computation.
Efficient quantum arithmetic circuits are pivotal for advancing quantum computing, underpinning numerous key algorithms. This also involves finding cheaper ways of performing uncomputation on qubits storing garbage/intermediate computations.

\subsection{Main results}
In this work, we improve quantum circuits for modular arithmetic using measurement-based uncomputation, specifically focusing on quantum modular adders. In section~\ref{sec:adders}, we begin by formalizing previously known circuits for addition and defining in a theorem the resources required for their implementation. For all the adders we consider~\cite{cuccaro2004new, vedral1996quantum, draper2000addition, gidney2018halving}, we also cover the possible cases for (controlled) addition, subtraction and comparison (by a constant). While these constructions can be easily derived, we could not find records of many of them in previous literature, e.g., (controlled) comparator by a constant, controlled modular addition, etc. We offer a modular and composable framework to construct these arithmetic circuits from other known arithmetic primitives. We also show a construction for a controlled adder using just a single ancilla by modifying the CDKPM adder (theorem~\ref{thm:cuccaro-controlled-adder-1-extra-ancilla}). In section~\ref{sec:modular_adders}, we consider various tradeoffs for combining the results presented in the previous section to form a quantum circuit that performs a (controlled) modular addition (by a constant). In particular, we show how two different adder architectures (the CDKPM adder~\cite{cuccaro2004new} and Gidney adder~\cite{gidney2018halving}) can be combined, leading to better space-time tradeoffs compared to using either adder individually within the modular adder architecture (see theorem~\ref{thm: modular adder gidney + CDKPM}). These tradeoffs are particularly important when running quantum algorithms in early error-corrected settings. In section~\ref{sec:MBU}, we discuss the proof of the measurement-based uncomputation lemma, which we present below. This technique was already known in the literature~\cite{gidney2018halving, kornerup2021tight, gidney2019approximate}, but to the best of our knowledge, it has not been formalized as a plug-and-play result. While MBU can be generalized to uncomputation of multi-qubit registers, for modular arithmetic we can use the following lemma.

\begin{lemma*}[Measurement-based uncomputation lemma]
    Let $g:\,\{0,1\}^n\to\{0,1\}$ be a function, and let $U_g$ be a self-adjoint $(n+1)$-qubit unitary that maps
$    U_g\ket{x}_n\ket{b}_1=\ket{x}_n\ket{b\oplus g(x)}_1$.
There is an algorithm performing the mapping 
$\sum_x\alpha_x\ket{x}_n\ket{g(x)}_1 \mapsto \sum_x\alpha_x\ket{x}_n\ket{0}_1 $
        using $U_g$, $2$ $\mathsf{H}$ gates, and $1$ $\mathsf{NOT}$ gate with $1/2$ probability.
    It also involves one single-qubit computational basis measurement and a single $\mathsf{H}$ gate.
\end{lemma*}
Applying the MBU lemma leads to gate costs expressed \emph{``in expectation''}. The expected value is calculated based on a Bernoulli distribution with a probability \( p = \frac{1}{2} \). This is because, for any quantum state $\sum_x\alpha_x\ket{x}_n\ket{g(x)}_1$, if $g(x) \in \{0,1\}$, when the last register is measured in the \( X \) basis, there is an equal probability of obtaining a \( \ket{0} \) or \( \ket{1} \) state. The remainder of section~\ref{sec:MBU} is dedicated to deriving formal statements for modular addition circuits optimized by MBU. While table~\ref{tab:results-main} presents our results for modular addition (section~\ref{sec:MBUmodularaddition}), we also share resource counts for controlled modular addition (section~\ref{sec:MBUcontrolledmodular}), modular addition by a constant (section~\ref{sec:MBUmodularbyaconstant}), and controlled modular addition by a constant (section~\ref{sec:MBUcontrolledmodularbyaconstant}). For example, for the controlled modular adder derived from~\cite{cuccaro2004new}, MBU can reduce the $\mathsf{Tof}$ (Toffoli) count by $n+\frac{1}{2}$ gates in expectation (compare proposition~\ref{prp:modular adder with CDKPM} with proposition~\ref{prp:cdkpm-controlled-modular-adder}).
The minimum number of $\mathsf{Tof}$ gates is achieved by proposition~\ref{prp:gidney-controlled-modular-adder} using $2n+2$ ancilla qubits. 
For modular addition by a constant, MBU reduces the $\mathsf{Tof}$ count by $n$ gates in expectation (compare proposition~\ref{prp:takahashi-cheap-const-mod-adder} with theorem~\ref{thm:mbu-takahashi-cheap-const-mod-adder}), leading to a $16.7\%$ improvement on the original circuit. We obtain similar improvements for controlled modular addition by a constant. We also show a way to perform a double-sided comparison (theorem~\ref{thm:two-sided-comp}), i.e. to see if a given register's value is contained within two other register's values. This circuit is a prime candidate for MBU, leading to a nearly $25\%$ reduction in the cost of the circuit.

\begin{table}[htbp]
\centering
\label{tab:mytable}
\begin{adjustbox}{width=\textwidth}
\begin{tabular}{c||c|c|c|c|c|c|c}
    &\textbf{Logical Qubits}&\multicolumn{2}{c|}{\textbf{Toffoli}}&\multicolumn{2}{c|}{\textbf{CNOT,CZ}}&\multicolumn{2}{c}{\textbf{X}}\\\hline
    &&w/o MBU&with MBU&w/o MBU&with MBU&w/o MBU&with MBU\\\hline
    (5 adder) VBE~\cite{vedral1996quantum}&$4n+2$&$20n+10$&$16n+8$&$20n + 2|p|+22$&$16n+2|p|+18$&$|p|+2$&$|p|+2.5$\\\hline
    (4 adder) VBE&$4n+2$&$16n+4$&$14n+4$&$20n +2|p|+18$&$17n+2|p|+15.5$&$2|p|+1$&$2|p|+1.5$\\\hline
    CDKPM~\cite{cuccaro2004new}&$3n+2$&$8n$&$7n$&$16n+2|p|+4$&$14n+2|p|+3.5$&$2|p|+1$&$2|p|+ 1.5$\\\hline
    Gidney~\cite{gidney2018halving}&$4n+2$&$4n$&$3.5n$&$26n+2|p|+4$&$22.75n+2|p|+3.5$&$2|p|+1$&$2|p|+1.5$\\\hline
    CDKPM+Gidney&$3n+2$&$6n$&$5.5n$&$21n+2|p|+4$&$17.75n+2|p|+3.5$&$2|p|+1$&$2|p|+1.5$\\\hline
    &\textbf{Logical Qubits}&\multicolumn{2}{c|}{$\mathsf{QFT}_{n+1}$}&\multicolumn{2}{c|}{$\mathsf{PCQFT}_{n+1}$}&\multicolumn{2}{c}{\textbf{-}}\\\hline
    &&w/o MBU&with MBU&w/o MBU&with MBU&-&-\\\hline
    Draper~\cite{draper2000addition}&$2n+2$&$10$&$8$&$1$&$1$&&\\\hline
    Draper (Expect)&$2n+2$&$8$&$6$&$1$&$1$&&\\\hline
\end{tabular}
\end{adjustbox}
\caption{Cost of modular addition using different plain adder circuits in the VBE architecture of modular addition. We denote by $|p|$ the Hamming weight of the bit string of the number $p$ expressed in binary expansion. $\mathsf{PCQFT}$ stands for ``partially-classical'' quantum Fourier transform, where each controlled rotation in the $\mathsf{QFT}$ is replaced by a classically controlled rotation. Every row in this table corresponds to a theorem in our manuscript. Results with MBU are considered ``in expectation.''. Here, Draper (Expec) represents the repeated use of the Draper addition circuit. In expectation, the cost of the first $\mathsf{QFT}$ and last $\mathsf{IQFT}$ become negligible compared to the rest of the circuit.
} \label{tab:results-main}
\end{table}

Our work has the potential to improve quantum circuits for modular multiplication, modular exponentiation, and quantum circuits for cryptographic attacks~\cite{wang2024optimal,vedral1996quantum,wang2024comprehensive}. We leave this research for future work. The interested reader can access
\footnote{\url{https://github.com/AdithyaSireesh/mbu-arithmetic}} the code for symbolic computation of some of the statements presented in this work.

\subsection{Related works}
This section reviews recent progress in quantum circuits for modular arithmetic, with a focus on optimizing size and depth. Then, we turn our attention to previous applications of measurement-based uncomputation.
 Early proposals for quantum addition~\cite{cuccaro2004new, vedral1996quantum, Draper2004}, some of which are employed in fault-tolerant algorithms with state-of-the-art resource estimates~\cite{gidney2021factor, litinski2023compute, Gouzien2023}, were inspired by classical circuits for addition, with an emphasis on minimizing the $\mathsf{T}$ count and/or the $\mathsf{Tof}$ (Toffoli) count, and the number of ancilla qubits. Subsequent approaches have explored more quantum-specific methods of performing arithmetic~\cite{zalka2006shor, draper2000addition, beauregard2002circuit,yuan2023improved}. A nice quantum-specific way to perform a modular addition into a quantum register $\ket{x}_n$ was developed by Zalka~\cite{zalka2006shor}. It involves only a single non-modular addition (with modulus $p$) as long as $\ket{x}_n$ is prepared in a \emph{coset state}: $\sum_{c=0}^{2^k-1}\ket{x+c\cdot p}_{n+k}$. The overhead here is that the size of the padding (i.e. the value of $k$ of extra qubits required) for the coset state grows logarithmically in the number of additions that are to be performed (See ~\cite{gidney2019approximate} for more details). All these arithmetic circuits listed previously (and many others) serve as building blocks for more complex operations such as modular addition, multiplication, and exponentiation, to name a few. Another survey on quantum adders can be found in~\cite{takahashi2009quantum}. The interested reader is encouraged to read~\cite{wang2024comprehensive} for a comprehensive literature review.

\paragraph{Measurement-Based Uncomputation (MBU).} Measurement-based uncomputation (MBU) was developed to efficiently address the challenge of uncomputing boolean functions. MBU is a technique used to perform uncomputation in a probabilistic manner~\cite{gidney2018halving}. This technique involves measuring the qubits to be uncomputed on the $X$ basis. Depending on the outcome of the measurement, either the uncomputation circuit is executed, or it is determined that the uncomputation has already been completed, thereby reducing the overall quantum resource requirements.

MBU has been employed in several cases~\cite{gidney2018halving, kornerup2021tight, gidney2019approximate, gidney2019windowed} to improve the resource counts of uncomputation. In particular, Gidney applied this strategy to addition circuits in ~\cite{gidney2018halving}. Starting from a Ripple Carry Addition (RCA) algorithm (see section~\ref{sec:adders} for details), they manage to reduce the $\mathsf{T/Tof}$ gate cost by replacing every $\mathsf{Tof}$ gate with a temporary logical-$\mathsf{AND}$ gate (equivalent to a $\mathsf{Tof}$ gate). While a $\mathsf{Tof}$ gate performs the mapping $\ket{x}_1\ket{y}_1\ket{t}_1 \xmapsto{\mathsf{Tof}} \ket{x}_1\ket{y}_1\ket{x \cdot y \oplus t}_1$, the logical-$\mathsf{AND}$ gate uses an ancilla qubit to execute the $\mathsf{Tof}$ gate as follows:

\[
\ket{x}_1\ket{y}_1\ket{t}_1\ket{0}_1 \xmapsto{\mathsf{Tof}} \ket{x}_1\ket{y}_1\ket{t}_1\ket{x \cdot y}_1 \xmapsto{\mathsf{CNOT}} \ket{x}_1\ket{y}_1\ket{x \cdot y \oplus t}_1\ket{x \cdot y}_1
~.
\]

Since the $\mathsf{Tof}$ gates come in pairs (the $\mathsf{Tof}$ gates used in the RCA computation must also be uncomputed), MBU is used to avoid performing a $\mathsf{Tof}$ gate for the uncomputation, replacing it with an $X$ basis measurement followed by a classically controlled $\mathsf{C\text{-}Z}$ gate. In \cite{gidney2019approximate}, MBU is used to construct the coset state more effectively. In the QROM lookup circuits introduced in~\cite{babbush2018encoding, gidney2019windowed}, MBU is used to construct a highly efficient unlookup (uncomputation of lookup) operation. For an address register $\ket{a}_n$, a lookup table performs the operation
\begin{equation}
\ket{a}_n \xmapsto{\mathsf{Lookup}} \ket{a}_n \ket{\ell_a},
\end{equation}
where $\ell_a$ is the lookup value associated with address $a$. An address register of size $n$ qubits corresponds to a lookup table of $L=2^n$ elements, making the cost of performing a lookup $L=2^n$ Toffolis. After $\ell_a$ has been used in the computation, it must be uncomputed. Naively performing the uncomputation would cost the same as the lookup. However, by combining unary encoding with MBU, it is possible to reduce the cost of the uncomputation to just $\sqrt{L}$~\cite{babbush2018encoding, gidney2019windowed} gates.

Other works study the more general problem of reversibly computing a chain of functions $x_0 \xmapsto{f_1} x_1=f_1(x_0) \xmapsto{f_2} \cdots x_{m-1}=f_{m-1}(x_{m-2}) \xmapsto{f_{m}} x_m=f_m(x_{m-1})$ through the use of pebbling games~\cite{bennett1989time}. Pebbling games offer a way to perform these computations reversibly and find effective ways of studying space-time tradeoffs of various strategies to uncompute intermediate results. In the case of chains of unitaries $\{\mathsf{U}_{f_i}\}_{i=1}^m$, spooky pebble games were introduced by Gidney~\cite{Gidney2019SpookyPebbles}, which allowed $X$ basis quantum measurements as an additional operation in pebbling games. The measurements may result in a phase applied to the states in the superposition, which is then tracked and corrected at a later time in the computation.

\subsection{Notation and convention}\label{subsec:Conventions}

We assume basic knowledge of quantum computing; we refer the interested reader to~\cite{nielsen}.
As in~\cite{cuccaro2004new},  we denote the quantum registers by capital letters $A, B, C$, etc. Each register is encoded with several qubits; we denote by $A_i$ the $i$-th qubit of register $A$.
Each qubit can be understood as a memory location; if not differently stated, we will assume that in qubit $A_i$ is initially stored the value $a_i$. We represent the state of a $n$ qubit register $A$ as $\ket{a}_n$.
In principle, applying a quantum gate to a given register affects the values stored in its qubit; we denote by $A_i \leftarrow x$ the storing of value $x \in \{0, 1\}$ in the $i$-th component of register $A$ as the result of a certain quantum circuit. Where necessary, we put the number of qubits in a register as a subscript following the ket representing the register's state, e.g.\ $\ket{x}_n$ for an $n$-qubit register.
For an $n$ bit string $\mathbf{x}=x_{n-1}x_{n-2}\dots x_{0}\in\{0,1\}^n$ we shall often equate it with the decimal number $x$ encoded in binary by $\mathbf{x}$. We denote with $a|_n$ the $n$  bit string encoding the number $a \in \{-2^n,\dots,0,1,\dots,2^n-1\}$ in the $2$'s complement notation (see remark~\ref{rem:2scomplement} in the appendix). For a number $a \in \mathbb{N}$, $|a|$ denotes the Hamming weight of the binary representation of $a$. The interested reader is referred to~\cite{von1993first} or ~\cite{harris2010digital} and appendix~\ref{Sec:binaryarithmetics} for further details on bit strings operations.

\paragraph{Quantum computing.}
We recall the definition of the following single-qubit gates.
\begin{center}
 $\mathsf{X}=\begin{bmatrix}
0 & 1 \\
1 & 0 
\end{bmatrix}$,~ $\mathsf{Z}=\begin{bmatrix}
1 & 0 \\
0 & -1 
\end{bmatrix}$,~ $\mathsf{S}=\begin{bmatrix}
1 & 0 \\
0 & \mathrm{i} 
\end{bmatrix}$,~ $\mathsf{T}=\begin{bmatrix}
1 & 0 \\
0 & e^{\mathrm{i}\pi/4}
\end{bmatrix}$~.
\end{center}
We denote by $\mathsf{C}_i\text{-}\mathsf{R}_j(\theta_k)$ the rotation by $\theta_k=2\pi/2^k$ of qubit $j$ controlled by qubit $i$ (See figure~\ref{fig:controlled-rotation-matrix}). We can eventually use register names indexed by its qubits: e.g. $\mathsf{C}_{X_j}\text{-}\mathsf{R}_{Y_i}(\theta_k)$ performs a controlled rotation between the $j$-th qubit in the register $X$ and targets the $i$-th qubits in the register $Y$. We define the function $\mathbf{1}[\mathsf{condition}]$ to be $1$ if $\mathsf{condition}$ is true, otherwise is $0$. In the case of a $\mathsf{QFT}$ based adder by a constant, we use $\mathsf{PCQFT_{a,n}}$ to denote a classically controlled set of rotation gates that need to be applied in order to perform the required addition in the phase. 
\begin{figure}[!ht]
\centering
\includegraphics{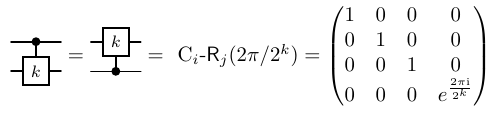}
\caption{Quantum Controlled Rotation Gate. We show both the circuit representation and the corresponding unitary matrix of the controlled rotation. The $i$ and $j$ correspond to the two qubits involved in the controlled rotation.}
\label{fig:controlled-rotation-matrix}
\end{figure}
\\
Let $y \in \{0,1\}^n$. 
The state obtained after applying the quantum Fourier transform (see e.g. ~\cite{nielsen}) $\mathsf{QFT}$ to the state $\ket{y}$ is given by
\begin{equation*}\label{fact:QFT on binary}
    \ket{\phi(y)}_n := \mathsf{QFT}_n \ket{y}_n = \frac{1}{2^{n/2}} \sum_x e^{2 \pi \mathrm{i} xy /2^n} \ket{x}_n.    
\end{equation*}
We denote by $\mathsf{QFT}_{n}$ the quantum Fourier transform acting on a $n$ qubits register and by $\mathsf{IQFT}_{n}$ its inverse.

Let $\mathsf{Q}$ be a quantum circuit composed of $m$ gates $\mathsf{Q}_j$ from a gateset $\mathcal{G} = \{\mathcal{G}_1, \mathcal{G}_2, \ldots, \mathcal{G}_K\}$. Define the gate count $c_i := |\{j \in [m] \mid \mathsf{Q}_j = \mathcal{G}_i\}|$ for $i \in [K]$. We represent the gate counts of $\mathsf{Q}$ as 
\[
 \mathsf{count}(\mathsf{Q}) = \{(\mathcal{G}_i, c_i) \mid i \in [K]\}.
\]

\begin{remark}[Circuit size of QFT~\cite{nielsen}]\label{remark:qftcost} A $\mathsf{QFT}_m$ has a circuit size of 
\[
 \mathsf{count}(\mathsf{QFT}_m) = \{(\mathsf{H}, m)\} \cup \{(\mathsf{C\text{-}R}(\theta_i), m+1-i) \mid i \in [1,m+1]\} . 
\]
\end{remark}

\paragraph{Boolean arithmetic.}
In this paragraph, we will review some basic definitions of bit string arithmetic. Denote by $\mathbf{x}= x_{n-1}\dots x_0 \in \{0,1\}^n$ a $n$-bit string i.e. a string of $n$ binary numbers. Further facts about the bit string representation of integers are reported in appendix~\ref{Sec:binaryarithmetics}. We consider the following four operations on the set of $n$-bits string.
    \begin{definition}[Bit string addition]\label{def:stringAddition}
        We call \emph{bit string addition} the binary operator 
        $$+:~\{0,1\}^n\times\{0,1\}^n\to \{0,1\}^{n+1}~, (\mathbf{x} , \mathbf{y}) \mapsto \mathbf{s} $$
        defined bit-wise by
        $s_i = x_i \oplus y_i \oplus c_i, $
        where the bit $c_i$, called $i$-th carry, is defined recursively by
        \begin{align}
            c_0 & = 0 \notag \\
            c_i & = \maj(x_{i-1}, y_{i-1}, c_{i-1}) = x_{i-1}y_{i-1} \oplus x_{i-1}c_{i-1} \oplus y_{i-1}c_{i-1} ~. \label{eq:majority}
        \end{align}
        The function $\maj$ is a logic gate called \emph{majority}, which gives $1$ whenever at least two of the three entries are different than $0$.
    \end{definition}

    \begin{definition}[Bit string $1$'s complement]\label{def:string1complement}
        We call \emph{$1$'s complement} the operator defined bit-wise by $ (\overline{\mathbf{x}})_i = x_i \oplus 1$ as
        $$\overline{~\cdot~}:~\{0,1\}^n\to \{0,1\}^{n}~, \mathbf{x} \mapsto \overline{\mathbf{x}}.$$
    \end{definition}

    \begin{definition}[Bit string $2$'s complement]\label{def:string2complement}
        We call \emph{$2$'s complement} the operator defined as $ \overline{\overline{\mathbf{x}}}=\overline{\mathbf{x}} + \mathbf{1}$, where $\mathbf{1}$ denotes the $n$ bits string $0\,\dots0\,1$ as
        $$\overline{\overline{~\cdot~}}:~\{0,1\}^n\to \{0,1\}^{n}~, \mathbf{x} \mapsto \overline{\overline{\mathbf{x}}}.$$
    \end{definition}

    \begin{definition}[Bit string subtraction]\label{def:stringSubtraction}
        We call \emph{bit string subtraction} the binary operator 
        $$-:~\{0,1\}^n\times\{0,1\}^n\to \{0,1\}^{n+1}~, (\mathbf{x} , \mathbf{y}) \mapsto \overline{\left(\overline{\mathbf{x}}+\mathbf{y}\right)} ~.$$
        The string $\mathbf{d}=\mathbf{x}-\mathbf{y}$ is given bit-wise by
        $ d_i = x_i \oplus y_i \oplus b_i$,         where the bit $b_i$, called $i$-th borrow, is defined recursively by
        \begin{align}
            b_0 & = 0 \notag \\
            b_{i+1} & = \maj(x_i\oplus 1, y_i, b_i) ~. \label{eq:borrow}
        \end{align}
    \end{definition}

\textbf{Naming convention for theorems and propositions.} As a guideline, we use ``proposition'' to refer to statements already present in the literature; in cases where we are unaware of a formal proof in the literature, we provide one. We use ``theorem'' to denote facts first stated and proven in this manuscript (to the best of our knowledge). We use the convention ``operation name - subroutine variant used (if any) - additional descriptors (if any)'' to name theorems and propositions. For example, ``Comparator - CDKPM - using half a subtractor'' represents a quantum comparator operation based on a subroutine variant by authors CDKPM, using a partial subtraction circuit.


\section{Quantum adders}\label{sec:adders}
%
\begin{definition}[Quantum addition]\label{def:abstractadder}
We define a \emph{quantum adder} as any unitary gate implementing the $n$-bit strings addition (see definition~\ref{def:stringAddition}) according to the following circuit
\begin{equation*}
    \centering
        \includegraphics[width=0.4\textwidth]{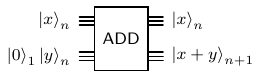}
\end{equation*}
\end{definition}
The initial extra qubit, prepared in $\ket{0}$, in the second register is needed to take into account the possible overflow. Note that by the theory of bit string additions (see appendix~\ref{Sec:binaryarithmetics} and section~\ref{subsec:Conventions}), such a circuit implements both the addition of unsigned integers, interpreting the bit string $x$ with the corresponding decimal number, and the addition of signed integers, interpreting the two entries in the 2's complement convention. Several quantum implementations of bit string addition have been discussed in the literature.
Due to their significant role as the fundamental components of quantum arithmetic, we recollect four of them for the sake of completeness.


\begin{table}[!htbp]
\centering
\small
\begin{tabular}{cccccc}
\toprule
\textbf{Procedure} & \textbf{Toff. Count} & \textbf{Ancillas} & \textbf{CNOT} & \textbf{Ref.} \\
\midrule
VBE plain adder~\cite{vedral1996quantum} & $4n$ & $n$ & $4n+4$ & prop.~\ref{prp:vedraladder} \\
CDKPM plain adder~\cite{cuccaro2004new} & $2n$ & $1$ & $4n+1$ & prop.~\ref{prp: CDKPM plain adder} \\
Gidney plain adder~\cite{gidney2018halving} & $n$ & $n$ & $6n-1$ & prop.~\ref{prp:gidneyAdder} \\
\toprule
 & $\mathsf{QFT}_{n+1}$ & \textbf{Ancillas} &  & \textbf{Ref.} \\
\midrule
Draper~\cite{draper2000addition} & $3$ & $0$ & $-$ & prop.~\ref{prp:draper-orig-qft-adder} \\
\midrule
\bottomrule
\end{tabular}
\caption{Gate count and logical qubits count for plain adders.}
\end{table}

\paragraph{VBE implementation}
One of the first implementations of a quantum adder had been proposed by Vedral, Barenko and Ekert (VBE) in \cite{vedral1996quantum}.
Such a circuit is built out of the auxiliary unitary gates $\mathsf{CARRY}$ and $\mathsf{SUM}$ (see figure~\ref{fig:carrySumGates})
computing the carry and the sum on a single bit.
\begin{figure}[h!]
    \centering
    \includegraphics[width=0.75\textwidth]{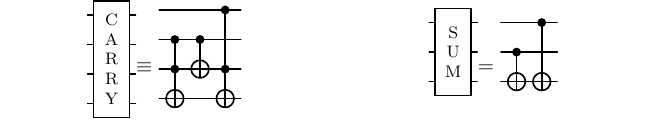}
    \caption{CARRY and SUM gate decomposition for the VBE adder \cite{vedral1996quantum}. The CARRY gate performs the mapping $\ket{c_k, x_k, y_k, c_{k+1}} \mapsto \ket{c_k, x_k, y_k \oplus x_k, c_{k+1} \oplus \maj(x_k, y_k, c_k)}$ where $\maj$ is defined in equation~\eqref{eq:majority} of section~\ref{subsec:Conventions}. The SUM gate performs the mapping $\ket{c_k, x_k, y_k} \mapsto \ket{c_k, x_k, y_k \oplus c_k \oplus x_k}$ }
    \label{fig:carrySumGates}
\end{figure}
The VBE implementation is subsumed by the following proposition.
\begin{proposition}[VBE plain adder~\cite{vedral1996quantum}]\label{prp:vedraladder}
There is a circuit implementing an $n$-bit quantum addition as per definition~\ref{def:abstractadder} using $n$ ancilla qubits and $4n$ $\mathsf{Tof}$ gates.
\end{proposition}
\begin{proof}
The sought circuit is given by figure~\ref{fig:vedralPlainAdder} together with the definition of the gates $\mathsf{CARRY}$ and $\mathsf{SUM}$ in figure~\ref{fig:carrySumGates}.
\begin{figure}[!ht]
    \centering
    \includegraphics[width=0.75\textwidth]{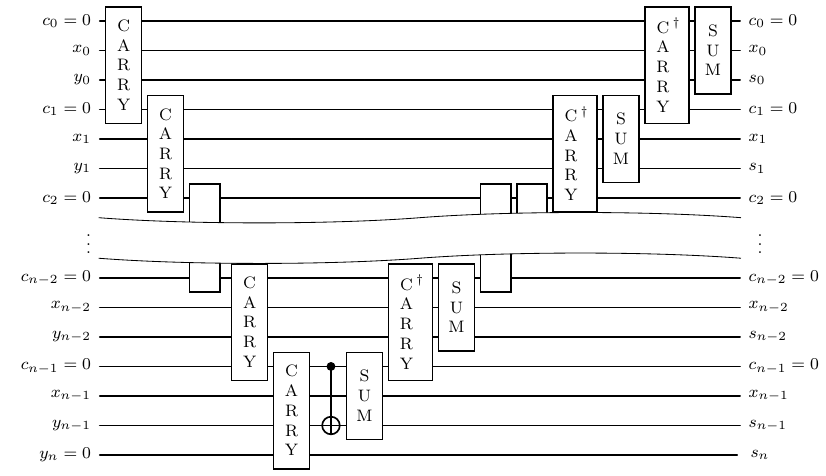}
    \caption{Adder circuit based on Vedral et al. \cite{vedral1996quantum}}
    \label{fig:vedralPlainAdder}
\end{figure}
\begin{enumerate}
    \item The circuit is initiated in the state 
    \begin{align*}
    \ket{c}_n\ket{x}_n\ket{y}_{n+1}=\ket{c_0, x_0, y_0, \dots c_i, x_i, y_i, c_{i+1}, \dots c_{n-1}, x_{n-1}, y_{n-1}, y_{n}}
    \end{align*}
    where $c_i = 0$ for any $0 \leq i \leq n $.

    \item The first part of the circuit consists of a chain of $n$ carry gates. The $i$-th gate computes the carry $c_{i+1}$ through the recursive equation~\eqref{eq:majority} storing its value in register $C_{i+1}$.
    At last, it computes the most significant bit of the result $x+y$, expressed by the carry $c_{n}=y_{n}$, storing it in the qubit $Y_{n}$, added to take into account the possible overflow.
    More precisely, the state after the chain of the $n$ different CARRY gates is:
    \begin{align*}
        C_0 & \leftarrow 0, \nonumber \\
        \vdots \nonumber \\ 
        C_i & \leftarrow  \maj(x_{i-1}, y_{i-1}, c_{i-1}) \nonumber \\
        X_i &\leftarrow x_i \nonumber \\ 
        Y_i & \leftarrow x_i \oplus y_i \nonumber \\
        \vdots \nonumber \\
        Y_{n} & \leftarrow y_{n} \oplus \maj(x_{n-1}, y_{n-1}, c_{n-1}) ~.
    \end{align*}

\item The central $\mathsf{CNOT}$ computes the second most significant bit of the sum, namely
    \begin{align*}
        Y_{n-1} & \leftarrow    s_{n-1}=x_{n-1}\oplus y_{n-1} \oplus c_{n-1} ~.
    \end{align*}
\item 
    The last part of the circuit consists of a chain of $n-1$ $\mathsf{CARRY}^\dagger$ and $n$ $\mathsf{SUM}$ gates.
    Inverses of the $\mathsf{CARRY}$ gate appear to restore every qubit of the temporary register $C$ to its initial state $|0\rangle$. In particular, $C_i\leftarrow 0$ only after $B_i\leftarrow a_i\oplus b_i \oplus c_i$.
\end{enumerate}
\end{proof}

\paragraph{CDKPM implementation }

A subsequent advancement over the VBE implementation was proposed by Cuccaro, Draper, Kutin, and Petrie Moulton (CDKPM) in \cite{cuccaro2004new}.
Their strategy greatly improves the use of ancillary qubits, requiring only one ancilla in contrast to the linearly many qubits required by the VBE adder.  
Moreover, this circuit enjoys a smaller depth, uses fewer gates and requires less space.
The CDKPM implementation is based on the two auxiliary gates $\mathsf{MAJ}$ (majority) and $\mathsf{UMA}$(unmajority and add).
The $\mathsf{MAJ}$ gate is given in figure~\ref{fig:MajGate}; it computes the majority function of three bits in place (See equation~\eqref{eq:majority} in section~\ref{subsec:Conventions}), while the $\mathsf{UMA}$, given in figure~\ref{fig:UmaGate} gate, restores the original values and writes the sum to the output as shown in figure~\ref{fig:MajUmaGate}.

\begin{figure}[h!]
\centering
    \includegraphics[width=0.4\textwidth]{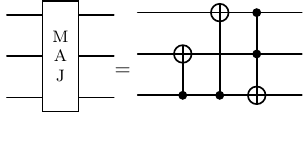}
\caption{Majority gate ($\mathsf{MAJ}$) as described in Cuccaro et al. \cite{cuccaro2004new}
        The $\mathsf{MAJ}$ gate performs the mapping $\ket{c_k, y_k, x_k} \mapsto \ket{c_k\oplus x_k, y_k\oplus x_k, \maj(x_k, y_k, c_k)}$ where the function $\maj$ is defined in equation~\eqref{eq:majority}.
}\label{fig:MajGate}
\end{figure}

\begin{figure}[h!]
    \centering
    \includegraphics[width=\textwidth]{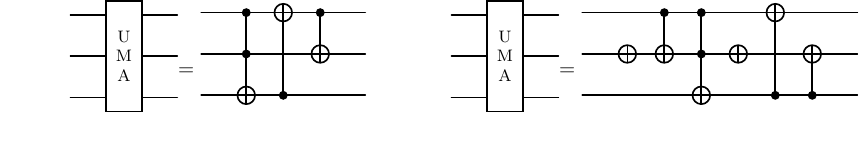}
    \caption{
        Unmajority gate (2 $\mathsf{CNOT}$ version and 3 $\mathsf{CNOT}$ version) as described in Cuccaro et al.~\cite{cuccaro2004new}. 
        The $\mathsf{UMA}$ gate performs the mapping
        $\ket{x_i, y_i, c_i} \mapsto \ket{a_i\oplus c_i \oplus x_i y_i, x_i\oplus y_i \oplus c_i \oplus x_i y_i, c_i \oplus x_i y_i}$.
        The operator is not self-adjoint, $(\mathsf{UMA})^\dagger$ performs the mapping $\ket{c_k, y_k, x_k} \mapsto \ket{c_k\oplus x_k, c_k\oplus y_k, \maj(x_k, y_k\oplus 1, c_k)}$.
    }\label{fig:UmaGate}
\end{figure}

\begin{figure}[h!]
    \centering
    \includegraphics[width=0.75\textwidth]{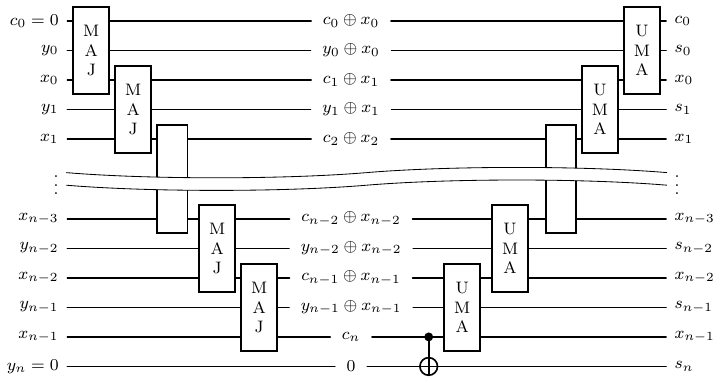}
    \caption{CDKPM ripple-carry adder of \cite{cuccaro2004new}. The corresponding significant bits of $a=a_0a_1\cdots a_{n-1}$ and $b=b_0b_1\cdots b_{n-1}$ are intertwined with each other. The register $b$ is extended by a single qubit to account for the last carry bit in the addition}\label{fig:cuccaro-ripple-carry-adder}
\end{figure}

\begin{figure}[h!]
\centering
    \includegraphics[width=0.8\textwidth]{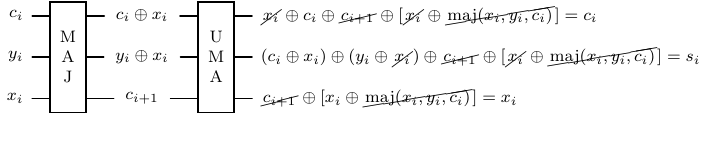}
    \caption{The combination of $\MAJ$ (majority) and $\UMA$ (unmajority) gates. Observe that $[(c_i\oplus x_i)(y_i\oplus x_i)]= x_i \oplus \maj(x_i,y_i,c_i)$ where the majority function $\maj$ is given in  equation~\eqref{eq:majority} of section~\ref{subsec:Conventions}.}\label{fig:MajUmaGate}
\end{figure}

\begin{proposition}[CDKPM plain adder~\cite{cuccaro2004new}]\label{prp: CDKPM plain adder}
There is a circuit implementing an $n$-bit quantum addition as per definition~\ref{def:abstractadder} using $1$ ancilla qubit and $2n$ $\mathsf{Tof}$ gates.
\end{proposition}
\begin{proof}
    The sought circuit is given in figure~\ref{fig:cuccaro-ripple-carry-adder}. The first qubit $C_0$ is an ancilla prepared in the state $c_0=0$, and the last qubit $Y_{n}$ is the extra qubit meant to take into account the possible sum overflow.
    The circuit consists of a sequence of $n$ $\mathsf{MAJ}$ gates, a central $\mathsf{CNOT}$ copying the most significant carry in the register $Y_{n}$ and a final sequence of $\mathsf{UMA}$ gates.
    The combination of $\mathsf{MAJ}$ and $\mathsf{UMA}$ gates is shown in figure~\ref{fig:MajUmaGate}.
\end{proof}

\paragraph{Gidney implementation }
A variant of the CDKPM adder was proposed by Gidney~\cite{gidney2018halving} (shown in figure~\ref{fig:gidney-plain-adder}). The implementation involves replacing $\mathsf{Tof}$ gates in the $\mathsf{MAJ}$ blocks by a temporary logical-$\mathsf{AND}$ stored in an ancilla qubit. This increases the number of ancillas by $n$: one for each of the $n$ $\mathsf{MAJ}$ blocks. In the $\mathsf{UMA}$ blocks, the uncomputation of the $\mathsf{Tof}$ gate is performed by a measurement and classically controlled $\mathsf{C}$-$\mathsf{Z}$ gate (using MBU). We show a block of Gidney's original contruction below.
\begin{figure}[!ht]
    \centering
    \includegraphics{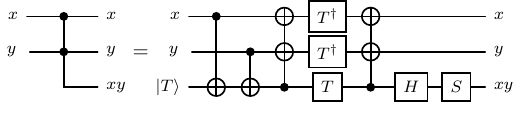}
    \caption{
        The temporary logical-$\mathsf{AND}$ construction of Gidney~\cite{gidney2018halving}. Here, the decomposition is using $\mathsf{T}$ gates, however, we consider just consider each temporary logical-$\mathsf{AND}$ gate implemented using a $\mathsf{Tof}$ gate in our paper.
    }
    \label{fig:gidney-temp-logical-and}
\end{figure}
\begin{figure}[!ht]
    \centering
    \includegraphics{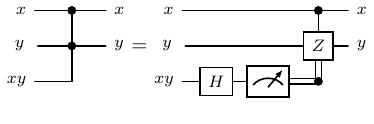}
    \caption{
        Uncomputation of the temp logical-$\mathsf{AND}$ by Gidney~\cite{gidney2018halving}. Here the circuit uses no $\mathsf{T}$ gates. It instead includes an $\mathsf{H}$ gate, a computation basis measurement and a classically controlled $\mathsf{C\text{-}Z}$ gate.
        }
    \label{fig:gidney-temp-logical-and-uncomputation}
\end{figure}

\begin{figure}[!ht]
    \centering
    \includegraphics[width=0.5\textwidth]{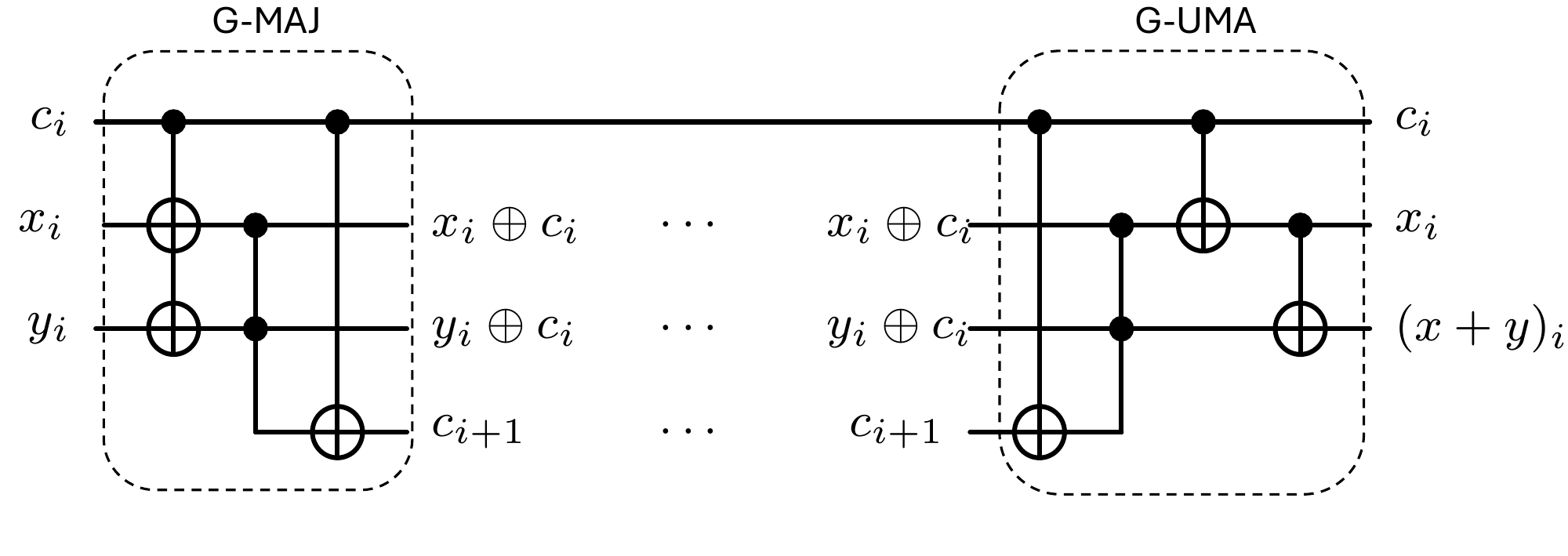}
    \caption{Adder block in Gidney's adder~\cite{gidney2018halving}}
    \label{fig:gidney-adder-block}
\end{figure}

\begin{figure}[!ht]
    \centering
    \includegraphics[width=0.75\textwidth]{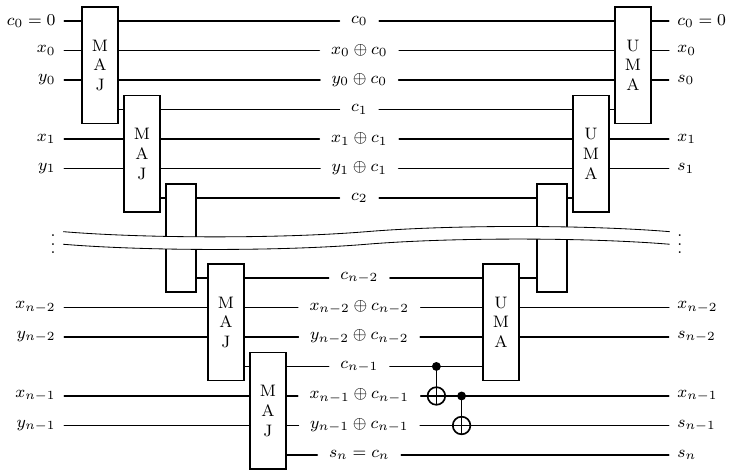}
    \caption{Gidney's logical-$\mathsf{AND}$ adder. The qubit $C_0$ starts in the state $\ket{0}$ and doesn't change value throughout the computation, therefore it can be omitted to reduce the  ancilla count by $1$ ~\cite{gidney2018halving}.}
    \label{fig:gidney-plain-adder}
\end{figure}

\begin{proposition}[Gidney adder~\cite{gidney2018halving}]\label{prp:gidneyAdder}
There is a circuit implementing an $n$-bit quantum adder as per definition~\ref{def:abstractadder}
using $n$ ancilla qubits and $n$ $\mathsf{Tof}$ gates.
\end{proposition}
\begin{proof}
An adder block in the temp logical-$\mathsf{AND}$ adder is constructed as shown in figure~\ref{fig:gidney-adder-block}.
The full circuit is shown in figure~\ref{fig:gidney-plain-adder}, it uses $n$ $\mathsf{MAJ}$ blocks, $(n-1)$ $\mathsf{UMA}$ blocks, and $2$ additional $\mathsf{CNOT}$s. We require $1$ less $\mathsf{UMA}$ block with respect to the CDKPM implementation because the final $\mathsf{MAJ}$ block computes the most significant carry bit $c_n$ which also happens to the be most significant bit of the addition $s_n=(x+y)_n$. We only use the two additional $\mathsf{CNOT}$'s for restoring the values of $x_{n-1}$ and computing the second most significant bit of the sum i.e. $s_{n-1}=(x+y)_{n-1}$.
\end{proof}

\paragraph{Draper implementation}
While the implementations described in the previous sections use essentially classical algorithms for addition and convert them to reversible quantum circuits, Draper~\cite{draper2000addition} came up with a quantum-specific way of addition using $\mathsf{QFT}$ gates.
\begin{figure}[!ht]
    \centering
    \includegraphics[width=0.8\textwidth]{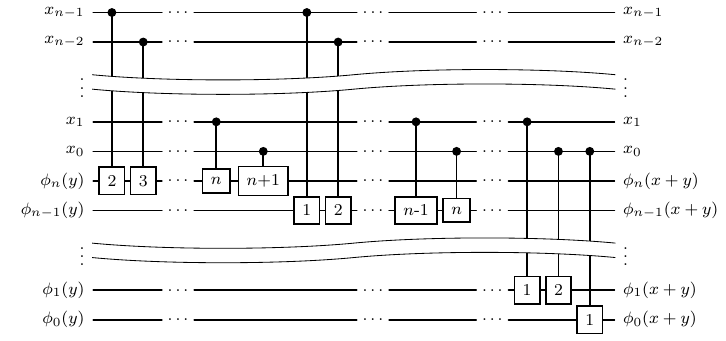}
    \caption{Draper's plain adder~\cite{draper2000addition}. It performs the operation $\ket{x}_n\ket{\phi(y)}_{n+1} \xmapsto{\mathsf{\Phi_{ADD}}} \ket{x}_n\ket{\phi(x + y)}_{n+1}$. See figure~\ref{fig:controlled-rotation-matrix} for the definition of the gates in the figure. Initially, $y_n=\ket{0}$  to account for the overflow of the addition
    }
    \label{fig:draper-adder}
\end{figure}

\begin{proposition}[Draper's plain adder~\cite{draper2000addition}]\label{prp:draper-orig-qft-adder}
Let $x\in \{0,1\}^{n}, y\in \{0,1\}^{n}$ be two bit strings stored in register $X$ and $Y$ of size $n$ and $n+1$ respectively (the most significant bit of register $Y$ is set to 0). Also, let $\ket{\phi(y)}_{n+1}:= \mathsf{QFT}_{n+1}\ket{y}_{n+1}$. There is a circuit $\mathsf{\Phi_{ADD}}$ performing the operation 
$$\ket{x}_n\ket{\phi(y)}_{n+1} \xmapsto{\mathsf{\Phi_{ADD}}} \ket{x}_n\ket{\phi(x+y)}_{n+1}$$
using zero ancilla qubits, and $\mathsf{count}(\mathsf{\Phi_{ADD}}) = \{(\mathsf{C\text{-}R}(\theta_1), n)\} \cup \{(\mathsf{C\text{-}R}(\theta_i), n+2-i) \mid i \in [2,n+1]\}~.$
\end{proposition}
\begin{proof}
    We denote as $\phi_i(y)$ the $i$-th qubit of $\phi(y)$. For $\phi_i(y)$, the following operations are performed in Draper's circuit (see figure~\ref{fig:draper-adder}). Consider $\widetilde{U}_i=\mathsf{C}_{X_0}\text{-}\mathsf{R}_{Y_i}(\theta_1)\mathsf{C}_{X_1}\text{-}\mathsf{R}_{Y_i}(\theta_2)\cdots\mathsf{C}_{X_0}\text{-}\mathsf{R}_{Y_i}(\theta_{i+1})$, the unitary acting on $\ket{\phi_i(y)}_1$. The transformation on $\ket{\phi_i(y)}_1$ is as follows:

\begin{align*}
\ket{x}_n\ket{\phi_i(y)}_1 &\xmapsto{\widetilde{U}_i} \mathsf{C}_{X_0}\text{-}\mathsf{R}_{Y_i}(\theta_1)\mathsf{C}_{X_1}\text{-}\mathsf{R}_{Y_i}(\theta_2)\cdots\mathsf{C}_{X_0}\text{-}\mathsf{R}_{Y_i}(\theta_{i+1})\ket{x}_n\left[\ket{0} + \exp\left(\frac{2\pi \mathrm{i} y}{2^{i+1}}\right)\ket{1}\right] \\
&=\ket{x}_n\left[\ket{0} + \left(\prod_{j=0}^i \exp\left(\frac{2 \pi \mathrm{i} x_j}{2^{i-j+1}}\right) \right)\exp\left(\frac{2\pi \mathrm{i} y}{2^{i+1}}\right)\ket{1}\right] \\
&=\ket{x}_n\left[\ket{0} + \exp\left(2 \pi \mathrm{i} \sum_{j=0}^i \frac{x_j}{2^{i-j+1}}\right)\exp\left(\frac{2\pi \mathrm{i} y}{2^{i+1}}\right)\ket{1}\right] \\
&=\ket{x}_n\left[\ket{0} + \exp\left(\frac{2 \pi \mathrm{i}}{2^{i+1}}\sum_{j=0}^i 2^j x_j\right)\exp\left(\frac{2\pi \mathrm{i} y}{2^{i+1}}\right)\ket{1}\right] \\
&=\ket{x}_n\left[\ket{0} + \exp\left(\frac{2 \pi \mathrm{i}}{2^{i+1}}\sum_{j=0}^{n-1} 2^j x_j\right)\exp\left(\frac{2\pi \mathrm{i} y}{2^{i+1}}\right)\ket{1}\right] \\
&=\ket{x}_n\left[\ket{0} + \exp\left(\frac{2 \pi \mathrm{i} x}{2^{i+1}}\right)\exp\left(\frac{2\pi \mathrm{i} y}{2^{i+1}}\right)\ket{1}\right] \\
&=\ket{x}_n\left[\ket{0} + \exp\left(\frac{2\pi \mathrm{i} (x+y)}{2^{i+1}}\right)\ket{1}\right] \\
&=\ket{x}_n\ket{\phi_i(x+y)}_1
~.
\end{align*}
We see that the $i^{th}$ qubit of the target register $\ket{\phi(y)}_{n+1}$, that is $\ket{\phi_{i}(y)}_1$, now holds the $i^{th}$ qubit of the Fourier transformed sum i.e. $\ket{\phi_{i}(x+y)}_1$. 
Hence, generalising the above result on the entire state of the target register, we conclude that the Draper adder performs the operation 
$$\ket{x}_{n}\ket{\phi(y)}_{n+1}=\ket{x}_{n}\bigotimes_{j=0}^{n} \ket{\phi_j(y)}_{1} \mapsto \ket{x}_{n}\bigotimes_{j=0}^{n} \ket{\phi_j(x+y)}_1 = \ket{x}_n\ket{\phi(x+y)}_{n+1}~.$$

The $\mathsf{\Phi_{ADD}}$ circuit on $n$ qubits has the same structure as a $\mathsf{QFT}$ circuit, excluding a $\mathsf{C\text{-}R}(\theta_1)$ gate and all the $\mathsf{H}$ gates. Therefore, the gate count of the $\mathsf{\Phi_{ADD}}$ circuit is:
\begin{equation*}~\label{eq:phi-add-cost}
    \mathsf{count}(\mathsf{\Phi_{ADD}}) = \{(\mathsf{C\text{-}R}(\theta_1), n)\} \cup \{(\mathsf{C\text{-}R}(\theta_i), n+2-i) \mid i \in [2,n+1]\}~.
\end{equation*}
\end{proof}

\begin{remark}
Resource-wise, $\mathsf{\Phi_{ADD}}$ is bounded above by the cost of a $\mathsf{QFT}_{n+1}$ gate.
\end{remark}
In Draper's original circuit, he considers the register $\ket{\phi(y)}_{n+1}=\mathsf{QFT}_{n+1}\ket{y}$ to be prepared beforehand, and the final output to be in the state $\ket{\phi(x+y)}_{n+1}=\mathsf{QFT}_{n+1}\ket{x+y}_{n+1}$ after the addition. However, to stay consistent with definition~\ref{def:abstractadder}, we must take into account the cost of the $\mathsf{QFT}_{n+1}$ to prepare $\ket{\phi(y)}_{n+1}$, and the cost of the $\mathsf{QFT}_{n+1}$ to return back to the computation basis at the end of the addition. Hence, we derive the following corollary.
\begin{corollary}[Draper's QFT-adder]\label{cor:Draper-qft-adder}
Let $x\in \{0,1\}^{n}, y\in \{0,1\}^{n}$ be two bit strings stored in register $X$ and $Y$ of size $n$ and $n+1$ respectively (the most significant bit of register $Y$ is set to 0). There is a circuit implementing a quantum adder as per definition~\ref{def:abstractadder} using $n$ ancilla qubits. The circuit gate cost is bounded above by the cost of $3\mathsf{QFT}_{n+1}$ circuits.
\end{corollary}

\subsection{Controlled addition}
\begin{definition}[Controlled addition]\label{def:abstractcontrolledadder}
    We define a \emph{controlled quantum adder} as any unitary gate implementing the $n$ strings addition controlled on a qubit $\ket{c}$ according to the following circuit
\begin{equation*}
        \includegraphics[width=0.4\textwidth]{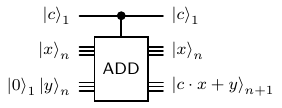}
\end{equation*}
\end{definition}

\begin{table}[!htbp]
\centering
\small
\begin{tabular}{cccccc}
\toprule
\textbf{Procedure} & \textbf{Tof Count} & \textbf{Ancillas} & \textbf{CNOT} & \textbf{Ref.} \\
\midrule
CDKPM & $3n$ & $1$ & $4n+1$ & thm.~\ref{thm:cuccaro-controlled-adder-1-extra-ancilla} \\
Gidney & $2n$ & $n+1$ & $7n-1$ & prop.~\ref{prp:gidney-controlled-adder-1-extra-ancilla} \\
\toprule
 & \textbf{Tof Count} & \textbf{Ancillas} & $\mathsf{QFT}_{n+1}$ & \textbf{Ref.} \\
\midrule
Draper & $n$ & $1$ & $3$ & thm.~\ref{thm:controlled-draper-single-ancilla} \\
\midrule
\bottomrule
\end{tabular}
\caption{Controlled addition by a constant}\label{table:controlled addition}
\end{table}

The theorem below gives a general recipe to perform a controlled addition given any adder.
\begin{theorem}[Controlled adder - with $n$ extra ancillas and $2n$ extra $\mathsf{Tof}$ gate]\label{thm:controlled-adder-n-extra-ancillas}
 Let $\mathsf{Q_{ADD}}$ be a quantum circuit that performs an $n$-bit quantum addition as per definition~\ref{def:abstractadder} using $s$ ancillas, and $r$ $\mathsf{Tof}$ gates.
Then, there is a circuit implementing a controlled quantum adder as per definition~\ref{def:abstractcontrolledadder} using $n+s$ ancilla qubits and $r+2n$ $\mathsf{Tof}$ gates.
\end{theorem}
\begin{proof}
    We first initialize a circuit in the state 
    \begin{align*}
    \ket{c}_1\ket{x}_{n}\ket{y}_{n}\ket{0}_{s+n+1}&~.\\ \intertext{Now, we use a sequence of $n$ $\mathsf{Tof}$ gates to load a controlled version of the addend $x$ i.e.}
    \ket{c}_1\ket{x}_{n}\ket{y}_{n}\ket{0}_{s+n+1}&\xmapsto{\mathsf{Load}}\ket{c}_1\ket{x}_{n}\ket{y}_{n}\ket{c\cdot x}_{n}\ket{0}_{s+1}.\\ \intertext{We can now perform the required addition using $c\cdot x$ as our addend}
    \ket{c}_1\ket{x}_{n}\ket{y}_{n}\ket{c\cdot x}_{n}\ket{s+1} &\xmapsto{\mathsf{Q_{ADD}}} \ket{c}_1\ket{x}_{n}\ket{y + c\cdot x}_{n+1}\ket{c\cdot x}_{n}\ket{0}_{s}~.\\ \intertext{We finally unload $c \cdot x$ using $n$ $\mathsf{Tof}$ gates}
    \ket{c}_1\ket{x}_{n}\ket{y + c\cdot x}_{n+1}\ket{c\cdot x}_{n}\ket{0}_{s+1}&\xmapsto{\mathsf{Unload}}\ket{c}_1\ket{x}_{n}\ket{y + c\cdot x}_{n}\ket{0}_{n+s}
    ~.
    \end{align*}
\end{proof}
The loading of $c\cdot x$ is performed in a register initially in the $\ket{0}_n$ state. This can be executed with just a temporary logical-$\mathsf{AND}$, and therefore can be uncomputed using measurements as in the case of the Gidney plain adder (proposition~\ref{prp:gidneyAdder}). This leads to the following corollary.

\begin{corollary}[Controlled adder - $n$ extra ancillas and $n$ extra $\mathsf{Tof}$ gates]\label{cor:controlled-addition-half-tofolli}
Let $\mathsf{Q_{ADD}}$ be a quantum circuit that performs an $n$-bit quantum addition as per definition~\ref{def:abstractadder} using $s$ ancillas, and $r$ $\mathsf{Tof}$ gates. Then, there is a circuit implementing a controlled quantum adder as per definition~\ref{def:abstractcontrolledadder} using $n+s$ ancilla qubits and $r+n$ $\mathsf{Tof}$ gates.
\end{corollary}
Although theorem~\ref{thm:controlled-adder-n-extra-ancillas} and corollary~\ref{cor:controlled-addition-half-tofolli} can be used with any general adder, Gindey~\cite{gidney2018halving} adapted his temporary logical-$\mathsf{AND}$ construction to give an adder that uses only an additional $n$ $\mathsf{Tof}$ gates, and a single extra ancilla. We present his construction below.
\begin{proposition}[Controlled adder - Gidney -  with 1 extra ancilla~\cite{gidney2018halving}]\label{prp:gidney-controlled-adder-1-extra-ancilla}
There is a circuit implementing an $n$-bit controlled quantum addition as per definition~\ref{def:abstractcontrolledadder} using $n+1$ ancilla qubits and $2n$ $\mathsf{Tof}$ gates. 
\end{proposition}
\begin{proof}
Figure~\ref{fig:ctrl-adder-building-block} is taken from Gidney's construction of a controlled adder building block. Here, we can isolate what happens in the case of a single $\mathsf{MAJ}-\mathsf{UMA}$ pair, and show what is the resultant state. Apart from the $\mathsf{Tof}$ gates of the $n$ $\mathsf{MAJ}$ blocks, we also require a $\mathsf{Tof}$ gate per $\mathsf{UMA}$ block to perform the controlled addition, leading to $2n$ $\mathsf{Tof}$ gates in total to execute a controlled adder in this case.
\end{proof}
\begin{figure}[!ht]
    \centering
    \includegraphics[width=0.75\textwidth]{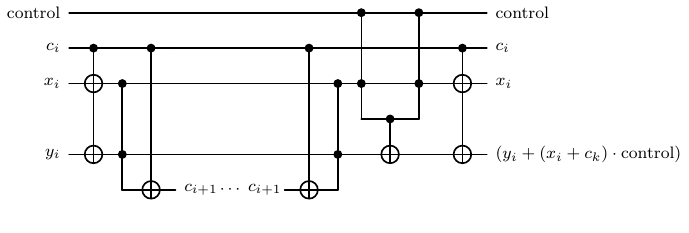}
    \caption{The controlled adder building block by Gidney~\cite{gidney2018halving}
    }
    \label{fig:ctrl-adder-building-block}
\end{figure}
We now show a version of the CDKPM adder that can perform controlled addition with a single ancilla qubit.
\begin{theorem}[Controlled adder - CDKPM - with $1$ ancilla]\label{thm:cuccaro-controlled-adder-1-extra-ancilla}
There is a circuit implementing an $n$-bit controlled quantum addition as per definition~\ref{def:abstractcontrolledadder} using $1$ ancilla qubit and $3n$ $\mathsf{Tof}$ gates.
\end{theorem}
\begin{proof}
    Controlled addition can be achieved by slightly modifying the CDKPM circuit (see figure~\ref{fig:cuccaro-ripple-carry-adder}) substituting the $\UMA$ gates with their controlled version given by figure~\ref{fig:cUmaGate}.
    \begin{figure}[h!]
        \centering
        \includegraphics[width=0.4\textwidth]{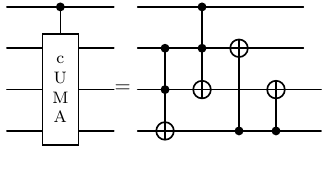}
        \caption{
            Controlled Unmajority gate.
        }\label{fig:cUmaGate}
    \end{figure}
    Overall, the consecutive application of $\mathsf{MAJ}$ and $\mathsf{C\text{-}UMA}$ gates acts as figure~\ref{fig:MajUmaGateControl}, thus determining the sought implementation of a controlled addition.
    \begin{figure}[h!]
        \centering
        \includegraphics[width=0.6\textwidth]{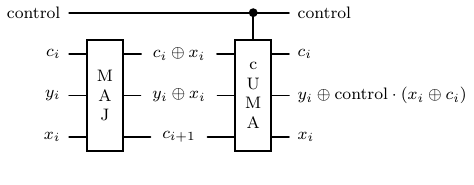}
        \caption{The combination of $\MAJ$ (majority) and $c-\UMA$ (controlled unmajority) gates.}\label{fig:MajUmaGateControl}
    \end{figure}
\end{proof}

While the Draper adder (as defined in corollary~\ref{cor:Draper-qft-adder} employing three $\mathsf{QFT}$'s) fits into the general framework of corollary~\ref{cor:controlled-addition-half-tofolli}, 
an alternate construction can be derived using only a single ancilla. The symmetries in the circuit can be exploited to avoid controls on every gate in the controlled Draper. Only the central $\mathsf{\Phi_{ADD}}$ gate would need to be controlled, while the the quantum Fourier transforms at the beginning and the end of the circuit remain uncontrolled.

\begin{theorem}
[Controlled adder - Draper]\label{thm:controlled-draper-simplified}
Let $x\in \{0,1\}^{n}, y\in \{0,1\}^{n}$ be two bit strings, and $c\in \{0,1\}$ be a control. Also,
let $\mathsf{Q_{QFT\text{-}ADD}}$ be the Draper adder circuit described in corollary~\ref{cor:Draper-qft-adder}, performing the operation 
$$\ket{x}_n\ket{y}_n\xmapsto{\mathsf{Q_{QFT\text{-}ADD}}}\ket{x}_n\ket{x+y}_{n+1}.$$
We can perform a controlled Draper adder $\mathsf{C\text{-}QFT_{ADD}}$ using a $\mathsf{QFT}_{n+1}$, a single controlled $\mathsf{\mathsf{\Phi_{ADD}}}$ (performing $\ket{c}_1\ket{x}_n\ket{\phi(y)}_{n+1}\xmapsto{\mathsf{C\text{-}\Phi_{ADD}}}\ket{c}_1\ket{x}_n\ket{\phi(y+c\cdot x)}_{n+1}$), and an $\mathsf{IQFT}_{n+1}$
\end{theorem}
\begin{proof}
Start with a quantum Fourier transform
\begin{align*}
    \ket{0}_{s+1}\ket{c}_1\ket{x}_n\ket{y}_n &\xmapsto{\mathsf{QFT}_{n+1}}\ket{0}_{s}\ket{c}_1\ket{x}_n\ket{\phi(y)}_{n+1}~.\\ \intertext{Then, one performs a $\mathsf{C\text{-}\Phi_{ADD}}$}
    \ket{0}_{s}\ket{c}_1\ket{x}_n\ket{\phi(y)}_{n+1}&\xmapsto{\mathsf{C\text{-}\Phi_{ADD}}}\ket{0}_{s}\ket{c}_1\ket{x}_n\ket{\phi(c\cdot x+y)}_{n+1}~,\\ \intertext{followed by a $\mathsf{IQFT}_{n+1}$ gate to return the required state back to the computation basis}
    \ket{0}_{s}\ket{c}_1\ket{x}_n\ket{\phi(c\cdot x+y)}_{n+1}&\xmapsto{\mathsf{IQFT}_{n+1}}\ket{0}_{s}\ket{c}_1\ket{x}_n\ket{c\cdot x+y}_{n+1}~.
\end{align*}
\end{proof}
Below, we show a way to execute the $\mathsf{C\text{-}\Phi_{ADD}}$ (as introduced in theorem~\ref{thm:controlled-draper-simplified}) using a single ancilla and $n$ extra $\mathsf{Tof}$ gates.
\begin{theorem}[Controlled adder - Draper - with 1 ancilla]~\label{thm:controlled-draper-single-ancilla}
Let $x\in \{0,1\}^{n}, y\in \{0,1\}^{n}$ be two bit
strings, and $c\in \{0,1\}$ be a control. Also,
let $$\mathsf{Q_{QFT\text{-}ADD}}=(\mathbf{I_n}\otimes\mathsf{IQFT_{n+1}})\mathsf{\Phi_{ADD}}(\mathbf{I_n}\otimes\mathsf{QFT_{n+1}})$$  be the circuit described in corollary~\ref{cor:Draper-qft-adder} (Draper's adder), performing a quantum addition as per definition~\ref{def:abstractadder}.
We can perform a controlled Draper adder $\mathsf{C\text{-}Q_{QFT\text{-}ADD}}$ using a circuit of cost bounded by $2\mathsf{QFT}_{n+1}$, and a $\mathsf{C\text{-}\Phi_{ADD}}$ gate. The circuit requires only a single ancilla.
\end{theorem}
\begin{proof}
Following up from the above theorem~\ref{thm:controlled-draper-simplified}, we find a construction for $\mathsf{C\text{-}\Phi_{ADD}}$, requiring only a single ancilla. We first notice that all gates in $\mathsf{\Phi_{ADD}}$ commute. Therefore, the order in which they are performed is irrelevant. Secondly, the main difference between $\mathsf{\Phi_{ADD}}$ and $\mathsf{C\text{-}\Phi_{ADD}}$ is that any controlled rotation $\mathsf{C\text{-}R}(\theta_i)$ (for any $i \in [n]$), now becomes a double controlled rotation $\mathsf{CC\text{-}R}(\theta_i)$ gate:
\begin{align*}
    &\ket{c_1}_1\ket{c_2}_1\ket{t}_1 \xmapsto{\mathsf{CC\text{-}R}(\theta_i)} \exp(2\pi \mathrm{i}(c_1 \cdot c_2 \cdot t)/2^i) \ket{c_1}_1 \ket{c_2}_1 \ket{t}_1
\end{align*}

We can decompose any $\mathsf{CC\text{-}R}$ gate it into a temporary logical-$\mathsf{AND}$, a single controlled rotation, followed by an uncomputation of the logical-$\mathsf{AND}$:
\begin{align*}
    \left[
    \begin{aligned}
        &\ket{c_1}_1 \ket{c_2}_1 \ket{0}_1 \ket{t}_1 \\
        &\xmapsto{\mathsf{Tof}} \ket{c_1}_1 \ket{c_2}_1 \ket{c_1 \cdot c_2}_1 \ket{t}_1 \\
        &\xmapsto{\mathsf{C\text{-}R}(\theta_i)} \exp(2\pi \mathrm{i}(c_1 \cdot c_2 \cdot t)/2^i) \ket{c_1}_1 \ket{c_2}_1 \ket{c_1 \cdot c_2}_1 \ket{t}_1 \\
        &\xmapsto{\mathsf{Uncompute\ Tof}} \exp(2\pi \mathrm{i}(c_1 \cdot c_2 \cdot t)/2^i) \ket{c_1}_1 \ket{c_2}_1 \ket{0}_1 \ket{t}_1
    \end{aligned}
    \right]
\end{align*}

Using the above decomposition, we can convert every $\mathsf{CC\text{-}R}$ into a temporary logical-$\mathsf{AND}$ + $\mathsf{C\text{-}R}$ (and finally an uncomputation of the temporary logical-$\mathsf{AND}$).
\begin{align*}
\ket{0}_1\ket{c}_1\ket{x}_n\ket{\phi(y)}_{n+1}&\xmapsto{\mathsf{Tof}}  \ket{c\cdot x_i}_1\ket{c}_1\ket{x}_n\ket{\phi(y)}_{n+1}~.
\end{align*}
We can now perform the required rotation gates that initially had $x_i$ as the control, by using $c\cdot x_i$ as our control instead. Compared to $\mathsf{\Phi_{ADD}}$, we require an extra $n(n+3)/2$ $\mathsf{Tof}$ for the temp logical-$\mathsf{AND}$'s, $n(n+3)/2$ $\mathsf{H}$ gates and $n(n+3)/4\text{ }\mathsf{C\text{-}R}(\theta_1)$ gates in expectation (cost of uncomputation of logical $\mathsf{AND}$s) in total for $\mathsf{C\text{-}\Phi_{ADD}}$. We only require a single ancilla as the ancilla is reset after every gate. The $\mathsf{Tof}$ cost can be made cheaper by noticing that we can rearrange the order of the controlled rotation gates in $\mathsf{C\text{-}\Phi_{ADD}}$ and group all gates that have a common control $x_i$.  This construction requires only an additional $n$ $\mathsf{Tof}$ gates (cost of logical $\mathsf{AND}$s), $n$ $\mathsf{H}$ gates and $n/2 \text{ }\mathsf{C\text{-}R}(\theta_1)$ gates in expectation (cost of uncomputation of the logical $\mathsf{AND}$s).
\end{proof}
\subsection{Addition by a constant}

\begin{definition}[Addition by a constant]\label{def:adderwithconstant}
    We define a \emph{quantum adder by a constant} as any unitary gate implementing the $n$ strings addition by a classically known constant $a\in \{0,1\}^n$ according to the following circuit
    \begin{equation*}
        \includegraphics[width=0.4\textwidth]{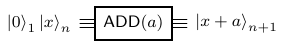}
    \end{equation*}
\end{definition}

Such a gate can be obtained from any given quantum adder implementation $\mathsf{Q_{ADD}}$ in several ways. In this section, we present a way to perform addition by a constant using any circuit as per definition~\ref{def:abstractadder}. We also dive into more specific implementations. 

\begin{table}[!htbp]
\centering
\small
\begin{tabular}{cccccc}
\toprule
\textbf{Procedure} & \textbf{Tof Count} & \textbf{Ancillas} & \textbf{CNOT} & \textbf{Ref.} \\
\midrule
CDKPM& $2n$ & $n+1$ & $4n+1$ & prop.~\ref{prp: addition by a constant - any} \\
Gidney& $n$ & $2n$ & $6n-1$ & prop.~\ref{prp: addition by a constant - any} \\
\toprule
 & $\mathsf{QFT}_{n+1}$ & \textbf{Ancillas} & $\mathsf{\Phi_{ADD}}(a)$  & \textbf{Ref.} \\
\midrule
Draper & $2$ & $0$ & $1$ & prop.~\ref{prp: addition by a constant Draper} \\
\midrule
\bottomrule
\end{tabular}
\caption{Additon by a constant}\label{table:addition by a constant}
\end{table}

\begin{proposition}
[Adder by a constant~\cite{cuccaro2004new}]\label{prp: addition by a constant - any}
Let $\mathsf{Q_{ADD}}$ be a quantum circuit that performs an $n$-bit quantum addition as per definition~\ref{def:abstractadder} using $s$ ancillas, and $r$ $\mathsf{Tof}$ gates. Then, there is a circuit implementing $$\ket{0}_1\ket{x}_n\mapsto \ket{x + a}_{n+1}$$ 
as per definition~\ref{def:abstractadder} using $n+s$ ancilla qubits and $r$ $\mathsf{Tof}$ gates
\end{proposition}
\begin{proof}
    This circuit can be obtained first by performing the mapping
$$\ket{0}_{s+n+1}\ket{x}_n\xmapsto{\mathsf{Load}} \ket{a}_{n}\ket{0}_{s+1}\ket{x}_n,$$ 
using $|a|$ NOT gates (one on each qubits $i$ such that $a_i\neq 0$) and treating the first $n+1$ qubits input as an ancillary channel. 
Then, we use the adder circuit $\mathsf{Q_{ADD}}$, to perform the required addition
$$\ket{a}_{n}\ket{0}_{s+1}\ket{x}_n \xmapsto{\mathsf{Q_{ADD}}} \ket{a}_{n}\ket{0}_{s}\ket{x+a}_{n+1} ~.$$
Finally, we uncompute the ancilla register storing $a$ by applying the previous sequence of not gates again.
$$\ket{a}_{n}\ket{0}_{s}\ket{x+a}_{n+1} \xmapsto{\mathsf{Unload}}\ket{0}_{s+n}\ket{x+a}_{n+1} ~.$$
    The corresponding circuit is depicted in figure~\ref{fig:abstract-constant-adder}, keeping implicit the $s$ ancillary qubits required by $\mathsf{Q_{ADD}}$.
\begin{figure}[!ht]
    \centering
        \includegraphics[width=0.7\textwidth]{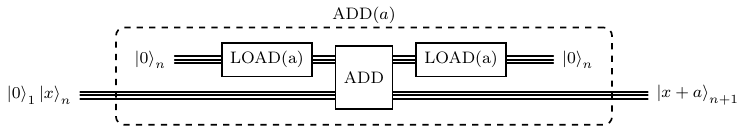}
    \caption{Addition by a constant value $a$. The unitary gate $\mathsf{LOAD(a)}$ is a self-adjoint implementation of the mapping which load and unload the value $a$.}
    \label{fig:abstract-constant-adder}    
\end{figure}

\end{proof}

While proposition~\ref{prp: addition by a constant - any} is applicable to any adder that can perform the mapping $\mathsf{Q_{ADD}}\ket{x}_n\ket{y}_n\mapsto\ket{x}_n\ket{x+y}_{n+1}$, in the case of Draper's $\mathsf{QFT}$ adder, there is a method to add constants without requiring ancillas. This method was used by Beauregard~\cite{beauregard2002circuit}.
\begin{proposition}[Adder by a constant - Draper~\cite{beauregard2002circuit}]\label{prp: addition by a constant Draper}
There is a circuit implementing an $n$-bit addition by a constant
as per definition~\ref{def:adderwithconstant} using $0$ ancillas.
\end{proposition}

\begin{proof}
\begin{figure}[!ht]
    \centering
    \includegraphics{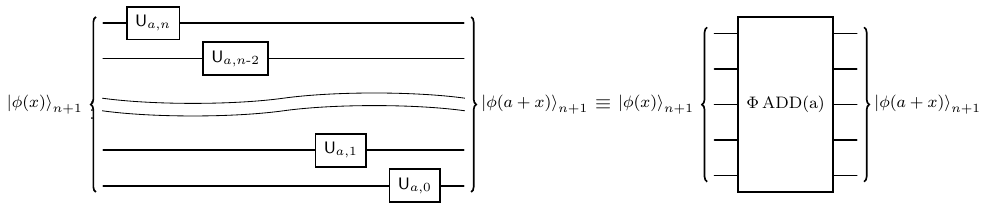}
    \caption{$\mathsf{QFT}$-based addition by a constant~\cite{beauregard2002circuit}
    }
    \label{fig:qft-constant-adder}
\end{figure}
The complete proof directly follows from corollary~\ref{cor:Draper-qft-adder}. 
figure~\ref{fig:qft-constant-adder} details the unitary gates that are required to perform the $\mathsf{QFT}$-based addition by a constant, where $\mathsf{U}_{a,i}$ is given by
\begin{align}
    \mathsf{U}_{a,i} &= \prod_{k=0}^i \mathsf{R}\left(a_i \theta_{i-k+1}\right)= \mathsf{R}\left(\sum_{k=0}^i a_k \theta_{i-k+1}\right) = \mathsf{R}\left(\frac{1}{2^{i+1}}\sum_{k=0}^i 2^ka_k\right)~.
\end{align}
We denote the $n$ qubit gate $\otimes_{i=0}^{n-1} \mathsf{U}_{a,i}$ as $\mathsf{\Phi_{ADD}}(a)$ (because our rotation gates are determined by the value of the bits of $a$).
\end{proof}

\subsection{Controlled addition by a constant} 

\begin{definition}[Controlled addition by a constant]\label{def:ctrl-const-adder}
    We define a \emph{controlled quantum adder} as any unitary gate implementing the $n$ strings addition, by a classically known constant $a\in \{0,1\}^n$ and controlled on a qubit $\ket{c}$, according to the following circuit
\begin{equation*}
        \includegraphics[width=0.4\textwidth]{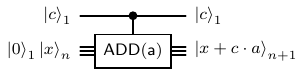}
\end{equation*}
\end{definition}

\begin{table}[!htbp]
\centering
\small
\begin{tabular}{cccccc}
\toprule
\textbf{Procedure} & \textbf{Tof Count} & \textbf{Ancillas} & \textbf{CNOT} & \textbf{Ref.} \\
\midrule
CDKPM& $2n$ & $n+1$ & $4n+1+2|a|$ & prop.~\ref{prp:controlled addition by a constant - any} \\
Gidney& $n$ & $2n$ & $6n-1+2|a|$ & prop.~\ref{prp:controlled addition by a constant - any} \\
\toprule
 & $\mathsf{QFT}_{n+1}$ & \textbf{Ancillas} & $\mathsf{C\text{-}\Phi_{ADD}}(a)$ & \textbf{Ref.} \\
\midrule
Draper & $2$ & $0$ & $1$ & prop.~\ref{prp:controlled addition by a constant - draper} \\
\midrule
\bottomrule
\end{tabular}
\caption{Controlled addition by a constant '$a$'}\label{table:controlled addition by a constant}
\end{table}
We now present a known construction that works with any adder to perform a controlled addition by a constant as per definition~\ref{def:ctrl-const-adder}
\begin{proposition}[Controlled adder by a constant~\cite{vedral1996quantum}]\label{prp:controlled addition by a constant - any}
Let $\mathsf{Q_{ADD}}$ be a quantum circuit that performs an $n$-bit quantum addition as per definition~\ref{def:abstractadder} using $s$ ancillas, and $r$ $\mathsf{Tof}$ gates. There is a circuit performing an $n$-bit controlled addition by a constant as per definition~\ref{def:ctrl-const-adder}, using $n+s$ ancillas, $r$ $\mathsf{Tof}$ gates and an additional $2|a|$ CNOTs.

\end{proposition}
\begin{proof}
    This circuit can be obtained first by performing the mapping
$$\ket{c}_1\ket{0}_{n+s+1}\ket{x}_n\xmapsto{\mathsf{Load}}\ket{c}_1\ket{ca}_{n}\ket{0}_{s+1}\ket{x}_n,$$ 
using $|a|$ CNOT gates. Then, we use the adder circuit, leaving the registers in the state
$$\ket{c}_1\ket{ca}_{n}\ket{0}_{s+1}\ket{x}_n \xmapsto{\mathsf{Q_{ADD}}} \ket{c}_1\ket{ca}_{n}\ket{0}_{s}\ket{x+ca}_{n+1} ~.$$ 
Finally, one has to uncompute the register storing $a$, applying the previous sequence of CNOT gates again.
    The corresponding circuit is depicted in figure~\ref{fig:ctrl-const-adder_circuit} keeping implied the $s$ ancillary qubits required by $\mathsf{Q_{ADD}}$.
\begin{figure}[!ht]
    \centering
        \includegraphics[width=0.7\textwidth]{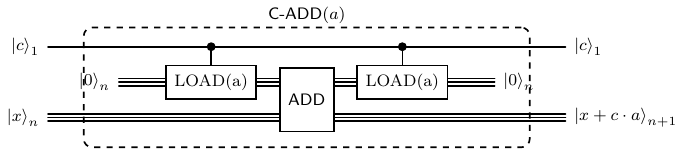}
    \caption{Controlled addition by a constant value $a$. The controlled unitary gate $\mathsf{Load(a)}$ is a self-adjoint implementation of the mapping which loads and unloads the value $c\cdot a$.}
    \label{fig:ctrl-const-adder_circuit}    
\end{figure}
\end{proof}

Similarly to the previous section, one can adapt Draper's implementation to suitably implement the gate of definition~\ref{def:ctrl-const-adder}.
\begin{proposition}[Controlled adder by a constant - Draper~\cite{beauregard2002circuit}]\label{prp:controlled addition by a constant - draper}
There is a circuit performing an $n$-bit controlled addition by a constant as per definition~\ref{def:ctrl-const-adder}
with zero ancilla qubits.
\end{proposition}

\begin{proof}
We reuse the construction of proposition~\ref{prp: addition by a constant Draper} and add controls to every rotation gate. The cost of this circuit depends on the binary representation of $a$. We define as $\mathsf{C\text{-}\Phi_{ADD}}(a)$ the gate performing the operation
\begin{equation*} \ket{c}_1\ket{\phi(x)}_{n+1}\xmapsto{\mathsf{C\text{-}\Phi_{ADD}}(a)}\ket{c}_1\ket{\phi(c\cdot a+x)}_{n+1}.
\end{equation*}

\end{proof}


\subsection{Subtraction}
\begin{definition}[Quantum subtraction]\label{def:abstractSubtractor}
We define a \emph{quantum subtractor} as any unitary gate implementing the $n$ strings subtraction (see definition~\ref{def:stringSubtraction}) according to the following circuit
\begin{equation*}
        \includegraphics[width=0.4\textwidth]{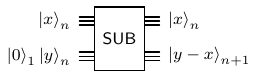}
\end{equation*}
\end{definition}
Due to the unitarity of the quantum adder gate and the uniqueness of the adjoint, it follows that $(\mathsf{ADD})^\dagger$ implements a quantum subtractor.
Alternatively, one can implement a quantum subtractor directly from the definition of bit-string subtractions.

\begin{theorem}\label{thm:otherquantumsubtractors}
    The following two circuits implement a quantum subtractor. 
    \begin{equation}
        \includegraphics[width=0.4\textwidth]{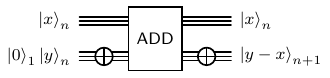}
        \label{eq:subtractor_1}
    \end{equation}
        
    \begin{equation}
         \includegraphics[width=0.6\textwidth]{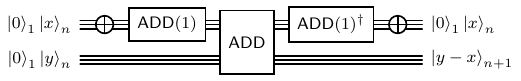}
        \label{eq:subtractor_2}
    \end{equation}
   
\end{theorem}

\begin{proof}
    By definition, the gate $\bigoplus$ flips all bits of a given register $\ket{r}$  namely producing the state $\ket{\overline{r}}$, i.e. the $1$'s complement of the string $r$ (see definition~\ref{def:string1complement}).
    Therefore, the circuit~\ref{eq:subtractor_1} implements bit string subtraction as given by definition~\ref{def:stringSubtraction}.
    On the other hand, the second circuit in theorem~\ref{thm:otherquantumsubtractors} implements subtraction using the $2$'s complement of the second addend as given by proposition~\ref{prop:subtraction-2complement} in the appendix.
\end{proof}
Interestingly, we recall that the gate performing addition by $1$, i.e. the ``increment by one'' operator, enjoys a more specialized implementation~\cite{Gidney2015LargeIncrementGates}.

\begin{remark}[Adjoint of Gidney Adder]~\label{rem:adjoint-gidney-adder} While a subtractor can be built using an adjoint of the adder, circuits involving a measurement (e.g. logical-$\mathsf{AND}$ adder) are generally not invertible. One way to invert the logical-$\mathsf{AND}$ adder is to swap the roles of the computation (logical-$\mathsf{AND}$) and uncomputation (MBU uncomputation of logical-$\mathsf{AND}$) of the logical-$\mathsf{AND}$'s in the circuit.
\end{remark}

\subsection{Comparison}\label{sec:comparison}
As expressed by proposition~\ref{prop:comparation-as-subtraction} in appendix~\ref{Sec:binaryarithmetics}, given a quantum subtractor, one can easily obtain a string comparator reading the sign (i.e. the most significant bit) of the difference.

\begin{definition}[Quantum comparison]\label{def:abstract_comparator}
    We define a \emph{quantum comparator } as any unitary gate implementing the $n$ strings comparison according to the following circuit
    \begin{equation*}
        \includegraphics[width=0.4\textwidth]{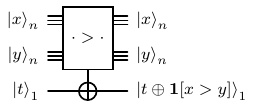}
    \end{equation*}
    where $\mathbf{1}[\mathsf{x>y}]$ coincides with the most significant bit of $y-x$, i.e. is $1$ when $x>y$ and $0$ otherwise.
\end{definition}

\begin{table}[htbp]
\centering
\small
\begin{tabular}{cccccc}
\toprule
\textbf{Procedure} & \textbf{Tof Count} & \textbf{Ancillas} & \textbf{CNOT} & \textbf{Ref.} \\
\midrule
CDKPM & $2n$ & $1$ & $4n+1$ & prop.~\ref{prp:comparator-cuccaro} \\
Gidney & $n$ & $n$ & $6n+1$ &  prop.~\ref{prp:comparator using half gidney subtractor} \\
\toprule
 & $\mathsf{QFT}_{n+1}$ & \textbf{Ancillas} & $-$ & \textbf{Ref.} \\
\midrule
Draper & $6$ & $1$ & $-$ & prop.~\ref{prp:quantum-comparator-beauregard-draper} \\
\midrule
\bottomrule
\end{tabular}
\caption{Comparators}\label{table:comparison with classical constant}
\end{table}

Due to the close link between the string subtraction and the string comparison algorithm, it is possible to implement a quantum comparator by adapting the implementations of quantum adder proposed in the previous section.

\begin{proposition}[Comparator - using an adder and a subtractor~\cite{vedral1996quantum}]~\label{prp:comp-using-two-adders}
Let $\mathsf{Q_{ADD}}$ be a quantum circuit that performs an $n$-bit quantum addition as per definition~\ref{def:abstractadder} using $s_{\mathsf{ADD}}$ ancillas, and $r_{\mathsf{ADD}}$ $\mathsf{Tof}$ gates. Additionally, let $\mathsf{Q_{SUB}}$ be a quantum circuit that performs an $n$-bit quantum subtraction as per definition~\ref{def:abstractSubtractor} using $s_{\mathsf{SUB}}$ ancillas, and $r_{\mathsf{SUB}}$ $\mathsf{Tof}$ gates.
Then, there is a circuit implementing an $n$-bit quantum comparator as per definition~\ref{def:abstract_comparator} using $s=1+\max(s_{\mathsf{ADD}}, s_{\mathsf{SUB}})$ ancilla qubits and $r=r_{\mathsf{ADD}} + r_{\mathsf{SUB}}$ $\mathsf{Tof}$ gates as well as an additional $\mathsf{CNOT}$
\end{proposition}
\begin{proof}
    This circuit can be obtained first by performing the mapping
$$\ket{0}_{s}\ket{x}_n\ket{y}_n\ket{t}_1\xmapsto{\mathsf{Q_{SUB}}} \ket{0}_{s-1}\ket{x}_n\ket{y-x}_{n+1}\ket{t}_1$$ 
Here, the value most significant bit of $(y-x)$ i.e. $(y-x)_{n}$ gives us the required comparator output $(y-x)_{n}=\mathbf{1}[x > y]$. Therefore, applying a $\mathsf{CNOT}$ on the target register $T$ with the most significant qubit of $\ket{y-x}$ as the control, we get
$$\ket{0}_{s-1}\ket{x}_n\ket{y-x}_{n+1}\ket{t}_1 \xmapsto{\mathsf{CNOT}} \ket{0}_{s-1}\ket{x}_n\ket{y-x}_{n+1}\ket{t \oplus \mathbf{1}[x > y]}_1~.$$
We can now apply $\mathsf{Q_{ADD}}$ to return back to the required state
$$\ket{0}_{s-1}\ket{x}_n\ket{y-x}_{n+1}\ket{t \oplus \mathbf{1}[x > y]}_1 \xmapsto{\mathsf{Q_{ADD}}} \ket{0}_{s}\ket{x}_n\ket{y}_{n}\ket{t \oplus \mathbf{1}[x > y]}_1~.$$
\end{proof}

The circuit described below is a component of the $\mathsf{QFT}$ modular addition circuit described by Beauregard~\cite{beauregard2002circuit}, which is based on the Draper $\mathsf{QFT}$ adder circuit (proposition~\ref{prp:draper-orig-qft-adder}). 
\begin{proposition}[Comparator - Draper/Beauregard~\cite{beauregard2002circuit}]\label{prp:quantum-comparator-beauregard-draper}
Let $x\in \{0,1\}^{n}, y\in \{0,1\}^{n}$ be two bit strings stored in register $X$ and $Y$ of size $n$ and $n+1$ respectively (the most significant bit of register $Y$ is set to 0). There is a circuit $\mathsf{Q_{COMP}}$ performing a quantum comparison as per definition~\ref{def:abstract_comparator}
using $1$ $\mathsf{CNOT}$, $2$ $\mathsf{QFT}_{n+1}$, $2$ $\mathsf{IQFT}_{n+1}$, $1$ $\mathsf{\Phi_{SUB}}$ and $1$ $\mathsf{\Phi_{ADD}}$ gates, and a single ancilla qubit.
\end{proposition}
\begin{proof}
We can directly apply proposition~\ref{prp:comp-using-two-adders} giving a comparator using an adder and a subtractor.
    Start preparing the register in the Fourier space.
\begin{align*}
    \ket{x}_n\ket{y}_{n+1}\ket{t}&\xmapsto{\mathsf{QFT}_{n+1}}\ket{x}_n\ket{\phi(y)}_{n+1}\ket{t}~.\\ \intertext{Then, perform a $\mathsf{\Phi_{SUB}}$ on the $Y$ register (subtracting by $x$), as the significant bit of this subtraction is 1 only when $x>y$.}
    \ket{x}_n\ket{\phi(y)}_{n+1}\ket{t}&\xmapsto{\mathsf{\Phi_{SUB}}}\ket{x}_n\ket{\phi(y-x)}_{n+1}\ket{t}~.\\ \intertext{We now return back to the computational basis using a $\mathsf{IQFT}_{n+1}$}
    \ket{x}_n\ket{\phi(y-x)}_{n+1}\ket{t}&\xmapsto{\mathsf{IQFT}_{n+1}}\ket{x}_n\ket{y-x}_{n+1}\ket{t} ~.\ \intertext{Controlling on and copying the most significant qubit of the difference i.e. $(y-x)_n$, we compute the required comparison}
    \ket{x}_n\ket{y-x}_{n+1}\ket{t}&\xmapsto{\mathsf{CNOT}}\ket{x}_n\ket{y-x}_{n+1}\ket{t\oplus \mathbf{1}[x > y]}_1 ~.\\ \intertext{The final set of operations described below, now return the value $y-x$ to its original state $y$}
    \ket{x}_n\ket{y-x}_{n+1}\ket{t\oplus \mathbf{1}[x > y]}_1&\xmapsto{\mathsf{QFT}}\ket{x}_n\ket{\phi(y-x)}_{n+1}\ket{t\oplus \mathbf{1}[x > y]}_1~ \xmapsto{}\\
    &\xmapsto{\mathsf{\Phi_{ADD}}} \ket{x}_n\ket{\phi(y)}_{n+1}\ket{t\oplus \mathbf{1}[x > y]}_1 ~\xmapsto{}\\
    &\xmapsto{\mathsf{IQFT}}\ket{0}_1\ket{x}_n\ket{y}_{n}\ket{t\oplus \mathbf{1}[x > y]}_1 ~.
\end{align*}
\end{proof}

Although the comparator described in proposition~\ref{prp:comp-using-two-adders} works for any adder and subtractor, based on the type of adder used, there are more specific constructions requiring fewer gates. 
\begin{proposition}[Comparator - CDKPM - using half a subtractor~\cite{cuccaro2004new}]\label{prp:comparator-cuccaro} 
There is a circuit implementing an $n$-bit quantum comparator as per definition~\ref{def:abstract_comparator} using $1$ ancilla qubit and $2n$ $\mathsf{Tof}$ gates.
\end{proposition}
\begin{proof}
The sought circuit is given by figure~\ref{fig:comparator}.
Recall that the adjoint of a unitary quantum adder performs a subtraction.
Notice that the first part of the circuit is one-half the adjoint of the CDKPM implementation of the adder.
Comparing the definition of the $\UMA$ gate, see figure~\ref{fig:UmaGate}, with the definition of the borrowing $b_k$, given by equation~\eqref{eq:borrow} of definition~\ref{def:stringSubtraction}, one can see that the first sequence of $\UMA^\dagger$ accumulate in the register $X$ the value of the borrow, in particular $X_{n}\leftarrow (y-x)_{n}$.
By construction, see also proposition~\ref{prop:signeddifference}, $(y-x)_{n}$ computes the overflow of $y-x$ i.e. is $1$ whenever $x>y$.
The subsequent CNOT copies the latter value in the extra register $T$, and the following sequence of $\UMA$ gate uncomputes the previous operation.

\end{proof}

The above comparator can also be executed using the temporary logical-$\mathsf{AND}$ gate of Gidney, which consumes $n$ more ancillas (to store the temporary logical and) in comparison to proposition~\ref{prp:comparator-cuccaro}

\begin{proposition}[Comparator - Gidney - using half a subtractor]\label{prp:comparator using half gidney subtractor}
There is a circuit implementing an $n$-bit quantum comparator as per definition~\ref{def:abstract_comparator}
using $n$ ancilla qubits and $n$ $\mathsf{Tof}$ gates.
\end{proposition}

\begin{figure}[h!]
    \centering
    \includegraphics[width=0.8\textwidth]{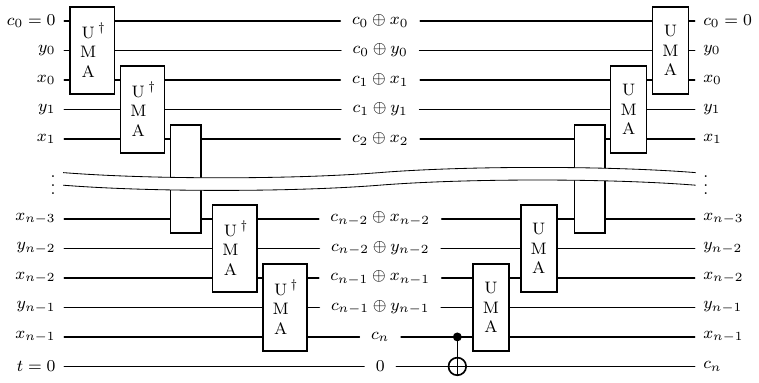}
    \caption{Comparator circuit based on the CDKPM adder.}
    \label{fig:comparator}
\end{figure}

One can also easily realise a controlled implementation of the comparator employing the previous addition circuits.
\begin{definition}[Controlled quantum comparison]\label{def:abstract_controlled_comparator}
    We define a \emph{controlled quantum comparator} as any unitary gate implementing the $n$ bits strings comparison according to the following circuit
    \begin{equation*}
        \includegraphics[width=0.4\textwidth]{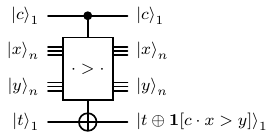}
    \end{equation*}
\end{definition}

\begin{proposition}[Controlled comparator - CDKPM - using half a subtractor~\cite{cuccaro2004new}]\label{prp:controlled comparator CDKMP}
There is a circuit implementing an $n$-bit controlled comparator as per definition~\ref{def:abstract_controlled_comparator} using $0$ ancilla qubits and $2n+1$ $\mathsf{Tof}$ gates.
\end{proposition}
\begin{proof}
    The comparator is controlled by a qubit $\ket{c}_1$. From proposition~\ref{prp:comparator-cuccaro}, we see that the output of the comparator is computed as the most significant bit $(y-x)_n$ of the subtraction of $y$ and $x$. The $\mathsf{CNOT}$ that copies the comparator value after the chain of $\mathsf{UMA}^{\dagger}$'s needs to now be controlled, thus leading to an extra $\mathsf{Tof}$ gate.
\end{proof}

The previous quantum controlled comparator can also be executed using the temporary logical-$\mathsf{AND}$ gate of Gidney, which consumes $n$ more ancillas (to store the temporary logical AND) in comparison to proposition~\ref{prp:controlled comparator CDKMP} but requires fewer $\mathsf{Tof}$ gates.

\begin{proposition}[Controlled comparator - Gidney - using half a subtractor]\label{prp:controlled comparator using half gidney subtractor}
There is a circuit implementing an $n$-bit controlled quantum comparator as per definition~\ref{def:abstract_controlled_comparator} using $n+1$ ancilla qubits and $n+1$ $\mathsf{Tof}$.
\end{proposition} 

When comparing registers, we note that they may not always be of equal size $n$. 
When registers $x \in \{0,1\}^{n}$ and $y \in \{0,1\}^{n'}$ with $n<n'$, one can circumvent this problem by padding the register $\ket{x}_n$ with an extra $n'-n$ qubits and perform the comparison on two $n'$ qubit registers. 
However, there is a slightly more efficient solution available for the CDKPM and Gidney adders.
\begin{remark} Assume that $n'=n+1$ without loss of generality. Let $y_{[0,1,\cdots,n-1]}$ represent the first $n$ bits of $y$, notice that $\mathbf{1}[x > y]\equiv \mathbf{1}[x > y_{[0,1,\cdots,n-1]}]\cdot (\mathsf{NOT}y_{n})$. Therefore, the $\mathsf{CNOT}$ computing the required comparison now becomes a $\mathsf{Tof}$ gate, with controls $\mathbf{1}[x > y_{[0,1,\cdots,n-1]}]$ and $(\mathsf{NOT}y_{n})$. Thus using a single extra $\mathsf{Tof}$ gate with no extra ancillas.
\end{remark}

We now explore constructions for comparing a quantum register with a classical value.
\begin{definition}[Quantum comparison with a classical constant]\label{def:abstract_constant_comparator}
    We define a \emph{quantum comparator with a classical constant} $a\in\{0,1\}^n$ as any unitary gate implementing the $n$ strings comparison with $a$ according to the following circuit
    \begin{equation*}
        \includegraphics[width=0.4\textwidth]{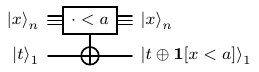}
    \end{equation*}
    where $\mathbf{1}[\mathsf{x<a}]$ coincides with the most significant bit of $x-a$ i.e is $1$ when $x<a$ and $0$ otherwise.
\end{definition}

\begin{proposition}[Comparator by a classical constant]\label{prp:comparator by a classical constant any}
Let $\mathsf{Q_{COMP}}$ be a circuit performing the $n$-bit comparator operation as per definition~\ref{def:abstract_comparator} using $r$ $\mathsf{Tof}$ gates and $s$ ancillas. Then, there is a circuit implementing an $n$-bit comparator by a classical constant as per definition~\ref{def:abstract_constant_comparator} using the $r$ $\mathsf{Tof}$ gates and $s+n$ ancillas, with an additional $2|a|$ NOT gates.
\end{proposition}
\begin{proof}
    Similar to the idea in proposition~\ref{prp: addition by a constant - any}, we first use an $n$ qubit register to load our constant $a$ into memory by applying $|a|$ NOT gates (one for each bit $a_i\neq 0$), giving the state
    $$\ket{0}_{s+n}\ket{x}_n\ket{t}_1\xmapsto{\mathsf{Load}}\ket{0}_{s}\ket{a}_n\ket{x}_n\ket{t}~.$$
    Now, we apply the comparator circuit with inputs $\ket{a}\ket{x}$ on target $\ket{t}$
    $$\ket{0}_{s}\ket{a}_n\ket{x}_n\ket{t}_1\xmapsto{\mathsf{Q_{COMP}}}\ket{0}_{s}\ket{a}_n\ket{x}_n\ket{t\oplus\mathbf{1}[a > x]}$$
    The ancilla is then uncomputed by applying the previous sequence of NOT gates a second time
    $$\ket{0}_{s}\ket{a}_n\ket{x}_n\ket{t\oplus\mathbf{1}[x > y]}_1\xmapsto{\mathsf{Unload}}\ket{0}_{s+n}\ket{x}_n\ket{t\oplus\mathbf{1}[x < a]}_1 ~.$$
\end{proof}

Similar to the construction in proposition~\ref{prp:comp-using-two-adders}, a comparator by a classical constant can be constructed by using an adder and a subtractor, as shown below.
\begin{theorem}[Comparator with a classical constant - using an adder by a constant and a subtractor by a constant]\label{thm:comp-by-a-classical-constant-with-adder-and-subtractor-by-a-contant}
Let $\mathsf{Q_{ADD}}(a)$ be a circuit performing an $n$-bit addition by a constant as per definition~\ref{def:adderwithconstant}, using $r_{\mathsf{ADD}}$ $\mathsf{Tof}$ gates and $s_{\mathsf{ADD}}$ ancillas. Also, let $\mathsf{Q_{SUB}}(a)$ be a circuit performing an $n$-bit subtraction by a constant as per definitions~\ref{def:adderwithconstant} and \ref{def:abstractSubtractor}, using $r_{\mathsf{SUB}}$ $\mathsf{Tof}$ gates and $s_{\mathsf{SUB}}$ ancillas. Then, there is a circuit implementing an $n$-bit comparator by a classical constant as per definition~\ref{def:abstract_constant_comparator}
using $r=r_{\mathsf{ADD}} + r_{\mathsf{SUB}}$ $\mathsf{Tof}$ gates, an additional $\mathsf{CNOT}$ gate and $s=1+\max(s_{\mathsf{ADD}},s_{\mathsf{SUB}})$ ancilla qubits
\end{theorem}
\begin{proof}
The idea is similar to proposition~\ref{prp:comp-using-two-adders} giving a comparator using an adder and a subtractor.
\end{proof}
Now, theorem~\ref{thm:comp-by-a-classical-constant-with-adder-and-subtractor-by-a-contant} can be applied specifically to the case of the Draper adder to give the theorem below.
\begin{proposition}[Comparator by a classical constant - Draper/Beauregard~\cite{beauregard2002circuit}]\label{prp:draper-comp-by-a-classical-constant}
Let $x\in \{0,1\}^{n}, a\in \{0,1\}^{n}$ be two bit strings. There is a circuit implementing a comparator by a classical constant:
$$\ket{x}_n\ket{t}_1\mapsto \ket{x}_n\ket{t \oplus \mathbf{1}[x < a]}_1.$$
using a circuit with $1$ $\mathsf{CNOT}$, $2$ $\mathsf{QFT}_{n+1}$, $2$ $\mathsf{IQFT}_{n+1}$, $1$ $\mathsf{\Phi_{SUB}}(a)$ and $1$ $\mathsf{\Phi_{ADD}}(a)$ gates, and a single ancilla qubit.
\end{proposition}
\begin{proof}
Using $\mathsf{Q_{ADD}}=\mathsf{IQFT}_{n+1}\mathsf{\Phi_{ADD}}(a)\mathsf{QFT}_{n+1}$ and $\mathsf{Q_{SUB}}=\mathsf{IQFT}_{n+1}\mathsf{\Phi_{SUB}}(a)\mathsf{QFT}_{n+1}$ in theorem~\ref{thm:comp-by-a-classical-constant-with-adder-and-subtractor-by-a-contant}, we require only a single ancilla.
\end{proof}

\begin{definition}[Controlled comparison by a classical constant]\label{def:abstract_constant_controlled_comparator}
    We define a \emph{controlled quantum comparator with a classical constant} $a\in \{0,1\}^n$ as any unitary gate implementing the $n$ strings comparison with $a$ according to the following circuit
    \begin{equation*}
        \includegraphics[width=0.4\textwidth]{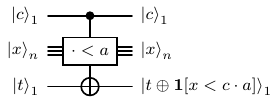}
    \end{equation*}
\end{definition}

\begin{theorem}[Controlled comparator by a classical constant - CDKPM - using half a subtractor]\label{thm:controlled comparator by a classical constant CDKMP}
There is a circuit implementing an $n$-bit controlled comparator by a classical constant as per definition~\ref{def:abstract_constant_controlled_comparator}
using $n+1$ ancilla qubits and $2n$ $\mathsf{Tof}$ gates 
\end{theorem}
\begin{proof}
    The previous statement follows by using a comparator with classical constant (proposition~\ref{prp:comparator by a classical constant any}) and using an ancilla register for writing (controlled on the register $\ket{c}$) the value of $a$ using $|a|$ CNOTs, which gets uncomputed after using the comparator.
\end{proof}

\begin{remark}\label{rem:invertingcomparison}
    Given any circuit implementing the comparison $(x>y)$ as defined in definition~\ref{def:abstract_comparator} one can get the opposite comparison $(x\leq y)$ by simply postcomposing the first circuit with an $\mathsf{X}$ gate acting on register $T$.
\end{remark}

\section{Quantum modular addition} \label{sec:modular_adders}
In this section, we detail and build on the methods of Vedral et. al~\cite{vedral1996quantum} for constructing modular adders as a suitable combination of the circuits for addition, subtraction and comparison introduced in section \ref{sec:adders}.
\begin{definition}[Modular addition]\label{def:modadder}
Let $x\in \{0,1\}^{n}, y\in \{0,1\}^{n}$ be two bit strings, and $p\in \{0,1\}^n$ classically known with $0\leq x,y<p$.
We define a \emph{quantum modular adder} as any unitary gate implementing the $n$ strings \emph{addition modulo $p$} according to the following circuit
\begin{equation*}
        \includegraphics[width=0.4\textwidth]{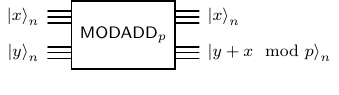}
\end{equation*}
\end{definition}

The following implementation is a slight modification of \cite{vedral1996quantum} and \cite{litinski2023compute}.
\begin{proposition}[Modular adder - Vedral's architecture~\cite{vedral1996quantum}]\label{prp:mod-add-first-vedral-architecture}
Let 
\begin{itemize}
    \item $\mathsf{Q_{ADD}}$ be a quantum circuit that performs a quantum addition using $s_{\mathsf{ADD}}$ ancillas, and $r_{\mathsf{ADD}}$ $\mathsf{Tof}$ gates,
    \item $\mathsf{Q_{COMP}}(p)$ be a quantum comparator with a classical constant $p$ using $s_{\mathsf{COMP}}$ ancillas, $r_{\mathsf{COMP}}$ $\mathsf{Tof}$ gates,
    \item $\mathsf{C\text{-}Q_{SUB}}(p)$ be a quantum circuit that performs a  (controlled) subtraction by a constant $p$ using $s_{\mathsf{C\text{-}SUB}}$ ancillas, and $r_{\mathsf{C\text{-}SUB}}$ $\mathsf{Tof}$ gates,
    \item $\mathsf{Q_{COMP}'}$: be a quantum comparator using ($s_{\mathsf{COMP}}'$ ancillas, $r_{\mathsf{COMP}}'$ $\mathsf{Tof}$ gates).
\end{itemize}
There is a circuit, $\mathsf{Q_{MODADD_\text{p}}}$, performing the $n$-bit modular addition operation as per definition~\ref{def:modadder} using $r=(r_{\mathsf{{ADD}}}+r_{\mathsf{{COMP}}}+r_{\mathsf{{C\text{-}SUB}}}+r_{\mathsf{{COMP}}}')$ $\mathsf{Tof}$ gates and $s=2+\max(s_{\mathsf{ADD}},s_{\mathsf{COMP}},s_\mathsf{C\text{-}SUB},s_{\mathsf{COMP}}')$ ancillas.
\end{proposition}
\begin{proof}
We start with a plain addition into our target
$$\ket{0}_{s+1}\ket{x}_n\ket{y}_n\xmapsto{\mathsf{Q_{ADD}}}\ket{0}_{s}\ket{x}_n\ket{x+y}_{n+1}~.$$
Using our comparator $\mathsf{Q_{COMP}}(p)$, we can perform a comparison of the previous sum with the given modulus $p$
$$\ket{0}_{s-1}\ket{x}_n\ket{x+y}_{n+1}\xmapsto{\mathsf{Q_{COMP}}(p)}\ket{0}_{s-1}\ket{x}_n\ket{x+y}_{n+1}\ket{\mathbf{1}[x+y <p]}~.$$
When the comparator's output is $0$, i.e. when $x+y\geq p$, we need to subtract the modulus $p$ to compute the correct sum $x+y \mod p$. Therefore, we perform a controlled subtraction of $p$ in the register currently storing $x+y$ (\textit{Note: We need to flip, using an $\mathsf{X}$ gate, the value of the comparison qubit before performing our controlled subtraction as pointed out in remark~\ref{rem:invertingcomparison}})
\begin{align*}
    \ket{0}_{s-1}\ket{x}_n\ket{x+y}_{n+1}\ket{1\oplus\mathbf{1}[x+y <p]}_1
= \ket{0}_{s-1}\ket{x}_n\ket{x+y}_{n+1}\ket{\mathbf{1}[x+y \geq p]}_1
\\
\xmapsto{\mathsf{C\text{-}Q_{SUB}(p)}}\ket{0}_{s-1}\ket{x}_n\ket{x+y \mod p}_{n+1}\ket{\mathbf{1}[x+y \geq p]}_1
~.
\end{align*}
Finally, we need to uncompute our previously computed output of the comparator. We notice that $\mathbf{1}[x+y \geq p] \equiv \mathbf{1}[x+y  \mod p < x]$ whenever $y<p$ strictly.
Namely, call $d=\mathbf{1}[x+y \geq p]$ then $x+y \mod p \equiv x+y-d\cdot p$ hence 
$$\mathbf{1}[x > x + y \mod p] = \mathbf{1}[d\cdot p >  y ] = d~.$$

Therefore, we can use $\mathsf{Q_{COMP}'}$ on inputs $x$ and $(x+y \mod p)$ to perform the final step 
\begin{align*}
    \ket{0}_{s-1}\ket{x}_n\ket{x+y\mod p}_{n+1}&\ket{\mathbf{1}[x+y \geq p]}_1\\
    \xmapsto{\mathsf{Q_{COMP}'}}&\ket{0}_{s-1}\ket{x}_n\ket{x+y\mod p}_{n+1}\ket{\mathbf{1}[x+y > p] \oplus \mathbf{1}[x+y\mod p < x]}_1\\
    =&\ket{0}_{s}\ket{x}_n\ket{x+y\mod p}_{n+1} ~.
\end{align*}
\end{proof}

\begin{figure}[h!]
        \centering
        \includegraphics[width=\textwidth]{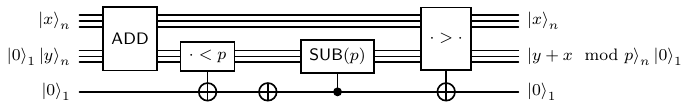}
        \caption{Implementation of a modular adder using previous quantum adder circuits.}\label{fig:modularadder}
\end{figure}

\begin{remark}
    In principle, one can implement a subroutine computing $x \mod p$ for any given number $0\leq x<2p$ as follows
    \begin{displaymath}
        \includegraphics[width=\textwidth]{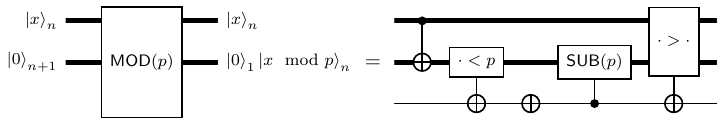}
    \end{displaymath}
    Therefore, a circuit computing $x+y \mod p$ can be alternatively given by
    \begin{displaymath}
        \includegraphics[width=.7\textwidth]{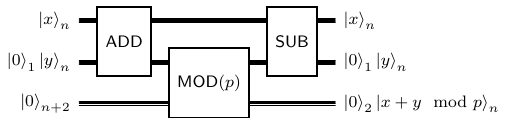}
    \end{displaymath}

\end{remark}

Proposition~\ref{prp:mod-add-first-vedral-architecture} gives an overview for constructing modular adders using any available combination of adders, (controlled) subtractor (by a constant), and comparators (by a constant). In the theorems below, we delve into specific adder circuits used for modular addition and calculate concrete costs more thoroughly.
\begin{proposition}[Modular adder - CDKPM~\cite{cuccaro2004new}]\label{prp:modular adder with CDKPM} 
There is a circuit, $\mathsf{Q_{MODADD_\text{p}}}$, performing the $n$-bit modular addition operation as per definition~\ref{def:modadder} using
  $n+3$ ancillas and $8n$ $\mathsf{Tof}$ gates.
\end{proposition}
    
\begin{proof}
    We employ proposition~\ref{prp:mod-add-first-vedral-architecture} with the following subroutines:
    \begin{itemize}
        \item $\mathsf{Q_{ADD}}$: CDKPM plain adder (proposition~\ref{prp: CDKPM plain adder});
        \item $\mathsf{Q_{COMP}}(p)$: CDKPM half subtractor based constant comparator proposition~\ref{prp:comparator-cuccaro};
        \item $\mathsf{C\text{-}Q_{SUB}}(p)$: CDKPM plain adder, more concretely, we use $\mathsf{Q_{ADD}^{\dagger}}$ in theorem~\ref{prp:controlled addition by a constant - any} for controlled constant addition/subtraction
        (proposition~\ref{prp:comparator-cuccaro});
        \item $\mathsf{Q_{COMP}'}$: CDKPM half subtractor based comparator proposition~\ref{prp:comparator-cuccaro}.
    \end{itemize}
From proposition~\ref{prp:mod-add-first-vedral-architecture}, we know that the total $\mathsf{Tof}$ cost of performing modular addition using the above selection of subroutines is $r=(r_{\mathsf{{ADD}}}+r_{\mathsf{{COMP}}}+r_{\mathsf{{C\text{-}SUB}}}+r_{\mathsf{{COMP}}}') = 8n$ $\mathsf{Tof}$ gates.
Similarly, the number of ancillas required is $s=2+\max(s_{\mathsf{ADD}},s_{\mathsf{COMP}},s_\mathsf{C\text{-}SUB},s_{\mathsf{COMP}}')=2+\max(1,n+1,n+1,1)=n+3$.

 \end{proof}

The version presented above has essentially the same structure of~\cite{vedral1996quantum} but enjoys some improvements. Namely, we use CDKPM's adder (in proposition~\ref{prp: CDKPM plain adder} instead of proposition~\ref{prp:vedraladder}), and we substitute the use of two adder circuits for the comparator $\mathsf{Q_{COMP}'}$ (proposition~\ref{prp:comp-using-two-adders}) with a single comparator circuit.  
These changes are already known in the literature (e.g.~\cite{litinski2023compute}), but it is not often mentioned the use of $n$ ancilla qubits for internal operations. If one has access to $2n$ ancillas, one can use Gidney's adder (proposition~\ref{prp:gidneyAdder} instead of the CDKPM adder (i.e. proposition~\ref{prp: CDKPM plain adder}). This allows for a decrease in the number of $\mathsf{Tof}$ gates and also the $\mathsf{Tof}$ depth.

\begin{proposition}[Modular adder - Gidney~\cite{gidney2018halving}]\label{prp: modular adder gidney}
There is a circuit, $\mathsf{Q_{MODADD_\text{p}}}$, performing the $n$-bit modular addition operation as per definition~\ref{def:modadder} using $2n+3$ ancillas and $4n$ $\mathsf{Tof}$ gates.
\end{proposition}
\begin{proof}
The proof follows the same steps as proposition~\ref{prp:modular adder with CDKPM}. We leverage proposition~\ref{prp:mod-add-first-vedral-architecture}, with the following subroutines:
    \begin{itemize}
        \item $\mathsf{Q_{ADD}}$: Gidney plain adder (proposition~\ref{prp:gidneyAdder});
        \item $\mathsf{Q_{COMP}}(p)$: Gidney half subtractor based constant comparator proposition~\ref{prp:comparator using half gidney subtractor};
        \item $\mathsf{C\text{-}Q_{SUB}}(p)$: Gidney plain adder, more concretely, we use $\mathsf{Q_{ADD}^{\dagger}}$ in proposition~\ref{prp:controlled addition by a constant - any} for controlled constant addition/subtraction
        (proposition~\ref{prp:comparator-cuccaro});
        \item $\mathsf{Q_{COMP}'}$: Gidney half subtractor based comparator proposition~\ref{prp:comparator using half gidney subtractor}.
    \end{itemize}
From proposition~\ref{prp:mod-add-first-vedral-architecture}, we know that the total $\mathsf{Tof}$ cost of performing modular addition using the above selection of subroutines is $r=(r_{\mathsf{{ADD}}}+r_{\mathsf{{COMP}}}+r_{\mathsf{{C\text{-}SUB}}}+r_{\mathsf{{COMP}}}') = 4n$ $\mathsf{Tof}$ gates.
Similarly, the number of ancillas required is $s=2+\max(s_{\mathsf{ADD}},n+s_{\mathsf{COMP}},n+s_\mathsf{C\text{-}SUB},s_{\mathsf{COMP}}')=2+\max(n+1,2n+1,2n+1,n+1)=2n+3$.
\end{proof}

The above modular adder can be further modified to reduce the number of ancillas. We propose to use the Gidney adder only for the first addition and final comparison, and use the CDKPM adder for the constant subtraction and constant comparison, thus giving us the following result.

\begin{theorem}[Modular adder - Gidney + CDKPM]\label{thm: modular adder gidney + CDKPM}
There is a circuit, $\mathsf{Q_{MODADD_\text{p}}}$, performing the $n$-bit modular addition operation as per definition~\ref{def:modadder} using $n+3$ ancillas and $6n$ $\mathsf{Tof}$ gates.
\end{theorem}
\begin{proof}
The proof follows the same steps as the proof of proposition~\ref{prp:modular adder with CDKPM}. In particular:
\begin{itemize}
        \item $\mathsf{Q_{ADD}}$: Gidney plain adder (proposition~\ref{prp:gidneyAdder});
        \item $\mathsf{Q_{COMP}}(p)$: CDKPM half subtractor based constant comparator proposition~\ref{prp:comparator-cuccaro};
        \item $\mathsf{C\text{-}Q_{SUB}}(p)$: CDKPM plain adder, more concretely, we use $\mathsf{Q_{ADD}^{\dagger}}$ in proposition~\ref{prp:controlled addition by a constant - any} for controlled constant addition/subtraction
        (proposition~\ref{prp:comparator-cuccaro});
        \item $\mathsf{Q_{COMP}'}$: Gidney half subtractor based comparator (proposition~\ref{prp:comparator using half gidney subtractor}).
    \end{itemize}
    From proposition~\ref{prp:mod-add-first-vedral-architecture}, we know that the total $\mathsf{Tof}$ cost of performing modular addition using the above selection of subroutines is $r=(r_{\mathsf{{ADD}}}+r_{\mathsf{{COMP}}}+r_{\mathsf{{C\text{-}SUB}}}+r_{\mathsf{{COMP}}}') = (n+2n+2n+n)=6n$ $\mathsf{Tof}$ gates.
Similarly, the number of ancillas required is $s=2+\max(s_{\mathsf{ADD}},n+s_{\mathsf{COMP}},n+s_\mathsf{C\text{-}SUB},s_{\mathsf{COMP}}')=2+\max(n+1,n+1,n+1,n+1)=n+3$ ancillas.
\end{proof}

\begin{proposition}[Modular adder - Draper/Beauregard~\cite{beauregard2002circuit}]\label{prp:draper-beauregard-modular-adder}
There is a circuit, $\mathsf{Q_{MODADD_\text{p}}}$, performing the $n$-bit modular addition operation as per definition~\ref{def:modadder} with a gate count of $3$ $\mathsf{QFT}$'s, $3$ $\mathsf{IQFT}$'s, $2$ $\mathsf{CNOT}$'s, $2$ $\mathsf{\Phi_{ADD}}$'s, $1$ $\mathsf{\Phi_{SUB}}$, $1$ $\mathsf{C\text{-}\Phi_{SUB}(p)}$, $1$ $\mathsf{\Phi_{ADD}(p)}$, $1$ $\mathsf{\Phi_{SUB}(p)}$ and
  $2$ ancillas.
\end{proposition}
\begin{proof}
We leverage proposition~\ref{prp:mod-add-first-vedral-architecture}, with the following subroutines:
    \begin{itemize}
        \item $\mathsf{Q_{ADD}}$: Draper plain adder corollary~\ref{cor:Draper-qft-adder};
        \item $\mathsf{Q_{COMP}}(p)$: Draper's comparator by a classical constant proposition~\ref{prp:draper-comp-by-a-classical-constant};
        \item $\mathsf{C\text{-}Q_{SUB}}(p)$: Draper adder, more concretely, we use $\mathsf{Q_{ADD}^{\dagger}}$ in proposition~\ref{prp:controlled addition by a constant - draper} and ~\ref{thm:controlled-draper-simplified} for controlled constant addition/subtraction;
        \item $\mathsf{Q_{COMP}'}$: Draper/Beauregard comparator proposition~\ref{prp:quantum-comparator-beauregard-draper}.
    \end{itemize}
The $\mathsf{IQFT}$ of $\mathsf{Q_{ADD}}$ cancels with the $\mathsf{QFT}$ of $\mathsf{Q_{COMP}}(p)$. We also find similar reductions for pairs the  $\mathsf{IQFT}$'s/$\mathsf{QFT}$'s of $\mathsf{Q_{COMP}}(p)/\mathsf{C\text{-}Q_{SUB}}(p)$ and $\mathsf{C\text{-}Q_{SUB}}(p)/\mathsf{Q_{COMP}'}$. The total gate count is therefore, $3$ $\mathsf{QFT}$'s, $3$ $\mathsf{IQFT}$'s, $2$ $\mathsf{CNOT}$'s, $2$ $\mathsf{\Phi_{ADD}}$'s, $1$ $\mathsf{\Phi_{SUB}}$, $1$ $\mathsf{C\text{-}\Phi_{SUB}(p)}$, $1$ $\mathsf{\Phi_{ADD}(p)}$ and $1$ $\mathsf{\Phi_{SUB}(p)}$. The number of ancillas required is $s=2$
\end{proof}

\subsection{Controlled modular addition}
\begin{definition}[Controlled modular addition]\label{def:ctrlmodadder}
We will denote controlled modular addition (with modulus $p$) as the following circuit
    \begin{equation*}
        \includegraphics[width=0.4\textwidth]{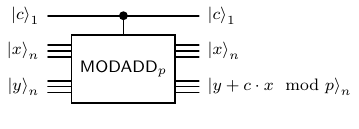}
    \end{equation*}
\end{definition}

A naive approach to creating a controlled modular adder is to add an extra control to each gate of the four subroutines composing the modular adder described in proposition~\ref{prp:mod-add-first-vedral-architecture}. However, there is a more efficient approach where only the first adder and final comparator need to be controlled. 
Such a trick has been mentioned in~\cite{litinski2023compute} without providing a proof.

\begin{proposition}[Controlled modular adder]~\label{prp:controlled-modular-addition}
Let 
\begin{itemize}
    \item $\mathsf{C\text{-}Q_{ADD}}$ be a quantum circuit that performs a quantum (controlled) addition using $s_{\mathsf{C\text{-}ADD}}$ ancillas, and $r_{\mathsf{C\text{-}ADD}}$ $\mathsf{Tof}$ gates;
    \item $\mathsf{Q_{COMP}}(p)$ be a quantum comparator by a constant $p$, using $s_{\mathsf{COMP}}$ ancillas, $r_{\mathsf{COMP}}$ $\mathsf{Tof}$ gates;
    \item $\mathsf{C\text{-}Q_{SUB}}(p)$ be a quantum circuit that performs a quantum subtraction by a constant $p$, using $s_{\mathsf{C\text{-}SUB}}$ ancillas, and $r_{\mathsf{C\text{-}SUB}}$ $\mathsf{Tof}$ gates;
    \item $\mathsf{C\text{-}Q_{COMP}'}$: be a controlled quantum comparator using ($s_{\mathsf{C\text{-}COMP}}'$ ancillas, $r_{\mathsf{C\text{-}COMP}}'$ $\mathsf{Tof}$ gates).
\end{itemize}
There is a circuit $\mathsf{C\text{-}Q_{MODADD_\text{p}}}$, performing the $n$-bit controlled modular addition operation as per definition~\ref{def:ctrlmodadder},
using $r=(r_{\mathsf{{C\text{-}ADD}}}+r_{\mathsf{{COMP}}}+r_{\mathsf{{C\text{-}SUB}}}+r_{\mathsf{{C\text{-}COMP}}}')$ $\mathsf{Tof}$ gates and $s=2+\max(s_{\mathsf{C\text{-}ADD}},n+s_{\mathsf{COMP}},n+s_\mathsf{C\text{-}SUB},s_{\mathsf{C\text{-}COMP}}')$ ancillas.
\end{proposition}
\begin{proof}
In order to perform a controlled modular adder, start by first performing out controlled addition of the addend $x$ into the target register
\begin{align*}
\ket{0}_s\ket{c}_1\ket{x}_n\ket{y}_n \xmapsto{\mathsf{C\text{-}Q_{ADD}}}&\ket{0}_{s-1}\ket{c}_1\ket{x}_n\ket{c\cdot x+y}_{n+1}\\ \intertext{We can now follow the step similar to the modular addition protocol, where we perform a constant comparison}
\ket{0}_{s-1}\ket{c}_1\ket{x}_n\ket{c\cdot x+y}_{n+1}\xmapsto{\mathsf{Q_{COMP}}(p)}&\ket{0}_{s-2}\ket{c}_1\ket{x}_n\ket{c\cdot x+y}_{n+1}\ket{\mathbf{1}[c\cdot x + y < p]}_1\\ \intertext{Once again, similar to the modular addition protocol, we perform a controlled subtraction of the modulus $p$, controlled on our comparator's output because if $(c\cdot x + y) \geq p$, then the modulus $p$ needs to be subtract to reduce the sum $\mod p$. \textit{Note: We need to first flip the comparison qubit, and then perform the controlled subtraction (see remark~\ref{rem:invertingcomparison})}}
\ket{0}_{s-2}\ket{c}_1\ket{x}_n\ket{c\cdot x+y}_{n+1}\ket{\mathbf{1}[c\cdot x + y \geq p]}_1&\\
\xmapsto{\mathsf{C\text{-}Q_{SUB}}(p)}&\ket{0}_{s-1}\ket{c}_1\ket{x}_n\ket{c\cdot x + y\mod p}_{n}\ket{\mathbf{1}[c\cdot x + y  \geq p]}_1\\ \intertext{We need to now uncompute the value initally computed by $\mathsf{Q_{COMP}}$. As in the modular addition case, we need to find a comparison operation that is equivalent to  $d=\mathbf{1}[c\cdot x + y \geq p]$ (the qubit we wish to uncompute/reset). 
Since $(c\cdot x + y  \mod p) = (c\cdot x - c\cdot d \cdot p)$, one finds out that $c\cdot \mathbf{1}[c\cdot x + y \geq p]\equiv  \mathbf{1}[(c\cdot x + y)\mod p < c\cdot x]$ and in particular $d= c\cdot d$ since $ (c\cdot x + y \geq p)$ is false whenever $c=0$ from the condition $y<p$.
The only difference here compared to proposition~\ref{prp:mod-add-first-vedral-architecture} is that the comparator bit is only set for states with $c\cdot x+y \geq p$, $\mathsf{AND}$ with the control bit set. The required operation is nothing but a controlled comparison, with control $c$, inputs $x$ and $(c\cdot x + y)\mod p$}
\ket{0}_{s-1}\ket{c}_1\ket{x}_n\ket{c\cdot x + y\mod p}_{n}&\ket{\mathbf{1}[c\cdot x + y \geq p]}_1\\
\xmapsto{\mathsf{C\text{-}Q_{COMP}'}}&\ket{0}_{s}\ket{c}_1\ket{x}_n\ket{c\cdot x + y\mod p}_{n}~.
\end{align*} 

\end{proof}
We will now explore some costs of controlled modular adders for a few specific combinations of subroutines.

\begin{proposition}[Controlled modular adder - CDKPM]\label{prp:cdkpm-controlled-modular-adder}
There is a circuit $\mathsf{C\text{-}Q_{MODADD_\text{p}}}$, performing the $n$-bit controlled modular addition operation as per definition~\ref{def:ctrlmodadder},
using $n+3$ ancillas and  $9n+1$ $\mathsf{Tof}$ gates.
\end{proposition}
\begin{proof}
We implement a controlled modular adder by using proposition~\ref{prp:controlled-modular-addition} in conjunction with the following subroutines:
    \begin{itemize}
    \item $\mathsf{C\text{-}Q_{ADD}}$ be a quantum circuit that performs a quantum (controlled) using a CDKPM adder with a single extra ancilla (theorem~\ref{thm:cuccaro-controlled-adder-1-extra-ancilla}) requiring $s_{\mathsf{C\text{-}ADD}}=2$ ancillas, and $r_{\mathsf{C\text{-}ADD}}=3n$ $\mathsf{Tof}$ gates;
    \item $\mathsf{Q_{COMP}}(p)$ be a quantum comparator by a constant $p$ using Controlled comparator - using half a CDKPM subtractor (theorem~\ref{thm:controlled comparator by a classical constant CDKMP}), requiring $s_{\mathsf{COMP}}=n+1$ ancillas, $r_{\mathsf{COMP}}=2n$ $\mathsf{Tof}$ gates;
    \item $\mathsf{C\text{-}Q_{SUB}}(p)$ be a quantum circuit that performs a quantum subtraction by a constant $p$ using the CDKPM adder in proposition~\ref{prp: addition by a constant - any}, requiring $s_{\mathsf{C\text{-}SUB}}=n+1$ ancillas, and $r_{\mathsf{C\text{-}SUB}}=2n$ $\mathsf{Tof}$ gates;
    \item $\mathsf{C\text{-}Q_{COMP}'}$ be the quantum comparator described in proposition~\ref{prp:controlled comparator CDKMP}: Controlled comparator - using half a CDKPM subtractor. This subroutine requires $s_{\mathsf{C\text{-}COMP}}'=1$ ancilla, $r_{\mathsf{C\text{-}COMP}}'=2n+1$ $\mathsf{Tof}$ gates.
\end{itemize}
There is a circuit $\mathsf{C\text{-}Q_{MODADD_\text{p}}}$, performing the controlled modular addition operation as per definition~\ref{def:ctrlmodadder},
using $r=(r_{\mathsf{{C\text{-}ADD}}}+r_{\mathsf{{COMP}}}+r_{\mathsf{{C\text{-}SUB}}}+r_{\mathsf{{C\text{-}COMP}}}') =(3n+2n+2n+(2n+1))=9n+1$ $\mathsf{Tof}$ gates and $s=2+\max(s_{\mathsf{C\text{-}ADD}},n+s_{\mathsf{COMP}},n+s_\mathsf{C\text{-}SUB},s_{\mathsf{C\text{-}COMP}}')=2+\max(2,n+1,n+1,1)=n+3$ ancillas.
\end{proof}

\begin{proposition}[Controlled modular adder - Gidney]\label{prp:gidney-controlled-modular-adder}
There is a circuit $\mathsf{C\text{-}Q_{MODADD_\text{p}}}$, performing the $n$-bit controlled modular addition operation as per definition~\ref{def:ctrlmodadder},
using $2n+3$ ancillas and $5n+1$ $\mathsf{Tof}$ gates.
\end{proposition}
\begin{proof}
We implement a controlled modular adder by using proposition~\ref{prp:controlled-modular-addition} in conjunction with the following subroutines:
    \begin{itemize}
    \item $\mathsf{C\text{-}Q_{ADD}}$ be a quantum circuit that performs a quantum (controlled) using a Gidney adder with a single extra ancilla (proposition~\ref{prp:gidney-controlled-adder-1-extra-ancilla}) requiring $s_{\mathsf{C\text{-}ADD}}=n+1$ ancillas, and $r_{\mathsf{C\text{-}ADD}}=2n$ $\mathsf{Tof}$ gates;
    \item $\mathsf{Q_{COMP}}(p)$ be a quantum comparator by a constant $p$ using Controlled comparator - using half a Gidney subtractor (proposition~\ref{prp:controlled comparator using half gidney subtractor}), requiring $s_{\mathsf{COMP}}=2n+1$ ancillas, $r_{\mathsf{COMP}}=n$ $\mathsf{Tof}$ gates;
    \item $\mathsf{C\text{-}Q_{SUB}}(p)$ be a quantum circuit that performs a controlled quantum subtraction by a constant $p$ using the Gidney adder in proposition~\ref{prp: addition by a constant - any}, requiring $s_{\mathsf{C\text{-}SUB}}=2n+1$ ancillas, and $r_{\mathsf{C\text{-}SUB}}=n$ $\mathsf{Tof}$ gates;
    \item $\mathsf{C\text{-}Q_{COMP}'}$ be the controlled quantum comparator described in proposition~\ref{prp:controlled comparator using half gidney subtractor}: controlled comparator - using half a Gidney subtractor. This subroutine requires $s_{\mathsf{C\text{-}COMP}}'=n+1$ ancillas, $r_{\mathsf{C\text{-}COMP}}'=n+1$ $\mathsf{Tof}$ gates.
\end{itemize}
According to proposition~\ref{prp:controlled-modular-addition}, there is a circuit performing the operation
$$\ket{c}_1\ket{x}_n\ket{y}_n\xmapsto{\mathsf{C\text{-}Q_{MODADD}}}\ket{c}_1\ket{x}_n\ket{c\cdot x+y\mod p}_n$$
using $r=(r_{\mathsf{{C\text{-}ADD}}}+r_{\mathsf{{COMP}}}+r_{\mathsf{{C\text{-}SUB}}}+r_{\mathsf{{C\text{-}COMP}}}') =(2n+n+n+(n+1))=5n+1$ $\mathsf{Tof}$ gates and $s=2+\max(s_{\mathsf{C\text{-}ADD}},s_{\mathsf{COMP}},s_\mathsf{C\text{-}SUB},s_{\mathsf{C\text{-}COMP}}')=2+\max(n+1,2n+1,2n+1,n+1)=2n+3$ ancillas.
\end{proof}

\subsection{Modular addition by a constant}
\begin{definition}[Modular addition by a constant]\label{def:abstract-modular-adder-by-a-const}
    Let $x \in \{0,1\}^n$. Assume to know classically $p,a \in \{0,1\}^n$ with $0\leq a,x<p$. We will denote modular addition by a constant $a$ (with modulus $p$) with the following circuit
\begin{equation*}
        \includegraphics[width=0.4\textwidth]{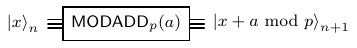}
\end{equation*}
\end{definition}
A simple way to perform a modular addition by a constant is by using a quantum modular adder as given by definition~\ref{def:modadder}.
\begin{proposition}[Modular adder by a constant]~\label{prp:constant-modular-adder}
Let $\mathsf{Q_{MODADD_{p}}}$ be a quantum circuit that performs an $n$-bit quantum modular addition as per definition~\ref{def:modadder}, using $s$ ancillas, and $r$ $\mathsf{Tof}$ gates. Then, there is a circuit $\mathsf{Q_{MODADD_{p}}}(a)$ implementing an $n$-bit quantum modular addition by a constant  
as per definition~\ref{def:abstract-modular-adder-by-a-const} using $n+s$ ancilla qubits and $r$ $\mathsf{Tof}$ gates, and an additional $2|a|$ $\mathsf{NOT}$ gates.
\end{proposition}
\begin{proof}
 Similar to the proposition~\ref{prp: addition by a constant - any}, We first take an unused $n$ qubit register, and load the value $a$ using $|a|$ $\mathsf{NOT}$ gates
\begin{align*}
    \ket{0}_{s+n}\ket{x}_n\xmapsto{\mathsf{Load}}& \ket{0}_{s}\ket{a}_n\ket{x}_n\\ \intertext{We now perform the required modular addition}
    \ket{0}_{s}\ket{a}_n\ket{x}_n \xmapsto{\mathsf{Q_{MODADD}}}& \ket{0}_{s}\ket{a}_n\ket{x+a\mod p}_n\\ \intertext{Finally, we unload the constant $a$}
    \ket{0}_{s}\ket{a}_n\ket{x+a\mod p}_n\xmapsto{\mathsf{Unload}}&\ket{0}_{s+n} \ket{x+a\mod p}_n~.
\end{align*}
\end{proof}
The above construction works for any modular adder construction that performs the operation $\ket{x}_n\ket{y}_n\xmapsto{\mathsf{Q_{MODADD}}_p} \ket{x}_n\ket{y+x \mod p}_n$. 
However, when examining the VBE architecture, we can modify certain subroutines to enable a different approach to performing constant modular addition.

\begin{theorem}[Modular adder by a constant - in VBE architecture]~\label{thm:constant-modular-adder-vbe-architecture}
Let
    \begin{itemize}
    \item $\mathsf{Q_{ADD}(a)}$ be a quantum circuit that performs a constant quantum addition using $s_{\mathsf{ADD}}$ ancillas, and $r_{\mathsf{ADD}}$ $\mathsf{Tof}$ gates;
    \item $\mathsf{Q_{COMP}}(p)$ be a quantum comparator using $s_{\mathsf{COMP}}$ ancillas, $r_{\mathsf{COMP}}$ $\mathsf{Tof}$ gates;
    \item $\mathsf{C\text{-}Q_{SUB}}(p)$ be a quantum circuit that performs a  (controlled) subtraction by a constant $p$ using $s_{\mathsf{C\text{-}SUB}}$ ancillas, and $r_{\mathsf{C\text{-}SUB}}$ $\mathsf{Tof}$ gates;
    \item $\mathsf{Q_{COMP}'}(a)$: be a constant quantum comparator using ($s_{\mathsf{COMP}}'$ ancillas, $r_{\mathsf{COMP}}'$ $\mathsf{Tof}$ gates).
\end{itemize}
there is a circuit $\mathsf{Q_{MODADD_{p}}}(a)$ implementing an $n$-bit quantum modular addition by a constant  
as per definition~\ref{def:abstract-modular-adder-by-a-const} using $r=(r_{\mathsf{{ADD}}}+r_{\mathsf{{COMP}}}+r_{\mathsf{{C\text{-}SUB}}}+r_{\mathsf{{COMP}}}')$ $\mathsf{Tof}$ gates and $s=2+\max(s_{\mathsf{ADD}},s_{\mathsf{COMP}},s_\mathsf{C\text{-}SUB},s_{\mathsf{COMP}}')$ ancillas.
\end{theorem}
\begin{proof}
The proof is very similar to the proof of the VBE modular adder architecture (proposition~\ref{prp:mod-add-first-vedral-architecture}).
We start with an addition by a constant $a$ into our target
$$\ket{0}_{s}\ket{x}_n\xmapsto{\mathsf{Q_{ADD}}(a)}\ket{0}_{s-1}\ket{x+a}_{n+1}~.$$
Using our comparator $\mathsf{Q_{COMP}}(p)$, we can perform a comparison of the previous sum with our modulus $p$
$$\ket{0}_{s}\ket{x+a}_{n+1}\xmapsto{\mathsf{Q_{COMP}}(p)}\ket{0}_{s-2}\ket{x+a}_{n+1}\ket{\mathbf{1}[x+a <p]}_1~.$$
When the comparator's output is $0$, i.e. when $x+a\geq p$, we need to subtract the modulus $p$ to compute the correct sum $x+a \mod p$. Therefore, we perform a controlled subtraction of $p$ in the register currently storing $x+a$ (note: we need to flip, using an $\mathsf{X}$ gate, the value of the comparison qubit before performing our controlled subtraction):
$$\ket{0}_{s-2}\ket{x+a}_{n+1}\ket{\mathbf{1}[x+a \geq p]}_1\xmapsto{\mathsf{C\text{-}Q_{SUB}(p)}}\ket{0}_{s-1}\ket{x}_n\ket{x+a \mod p}_{n}\ket{\mathbf{1}[x+a \geq p]}_1~.$$
We finally need to uncompute the comparator qubit. We notice that $\mathbf{1}[x+a \geq p] \equiv \mathbf{1}[x+a  \mod p < a]$, therefore, we can use $\mathsf{Q_{COMP}'}(a)$ on inputs $a$ and $(x+a \mod p)$ to perform the final step 
\begin{align*}
    \ket{0}_{s-1}\ket{x}_n\ket{x+a\mod p}_{n}&\ket{\mathbf{1}[x+a \geq p]}_1\\
    \xmapsto{\mathsf{Q_{COMP}'}(a)}&\ket{0}_{s-1}\ket{x+a\mod p}_{n}\ket{\mathbf{1}[x+a < p] \oplus \mathbf{1}[x+a\mod p \geq a]}_1\\
    =&\ket{0}_{s}\ket{x+a\mod p}_{n}~.
\end{align*}
\end{proof}

Given that we are adding by a constant, various ways exist to create more cost-effective modular addition circuits by modifying the plain addition implemented in the modular adder. Takahashi~\cite{takahashi2009quantum} in 2009 was the first to observe that the initial two operations in the VBE modular architecture could be merged into a single operation when performing a constant modular addition.

\begin{proposition}[Modular adder by a constant - in Takahashi architecture~\cite{Takahashi2009}]~\label{prp:takahashi-cheap-const-mod-adder}
Let
    \begin{itemize}
    \item $\mathsf{Q_{SUB}(a)}$ be a quantum circuit that performs a constant quantum subtraction using $s_{\mathsf{SUB}}$ ancillas, and $r_{\mathsf{SUB}}$ $\mathsf{Tof}$ gates;
    \item $\mathsf{C\text{-}Q_{ADD}}(p)$ be a quantum circuit that performs a  (controlled) addition by a constant $p$ using $s_{\mathsf{C\text{-}ADD}}$ ancillas, and $r_{\mathsf{C\text{-}ADD}}$ $\mathsf{Tof}$ gates;
    \item $\mathsf{Q_{COMP}}(a)$ be a classical comparator using $s_{\mathsf{COMP}}$ ancillas, $r_{\mathsf{COMP}}$ $\mathsf{Tof}$ gates.
\end{itemize}
there is a circuit $\mathsf{Q_{MODADD_{p}}}(a)$ implementing an $n$-bit quantum modular addition by a constant  
as per definition~\ref{def:abstract-modular-adder-by-a-const} using $r=(r_{\mathsf{{SUB}}}+r_{\mathsf{{C\text{-}ADD}}}+r_{\mathsf{{COMP}}})$ $\mathsf{Tof}$ gates and $s=1+\max(s_{\mathsf{SUB}},s_\mathsf{C\text{-}ADD},s_{\mathsf{COMP}})$ ancillas.
\end{proposition}
\begin{proof}
 To perform the unitary operation $\ket{0}_s\ket{x}_n \xmapsto{\mathsf{Q_{MODADD}}_p(a)} \ket{x+ a \pmod p}_n$, we start subtracting $p-a$ (which is different from the first step of proposition~\ref{prp:mod-add-first-vedral-architecture}, where we add $a$)
\begin{align*}
    \ket{0}_{s-1}\ket{x}_n \xmapsto{\mathsf{Q_{SUB}(p-a)}}& \ket{0}_{s-1}\ket{x - (p-a)}_{n+1} = \ket{0}_{s-1}\ket{x+a-p}_{n+1}\\ \intertext{Now, we need to add back $p$ to those states whenever $x+a<p$. Notice that $(x+a-p)_n$, i.e. the most significant bit of $x+a-p$ is $1$ only when $x < (p-a)$. Controlled on $(x+a-p)_n$, we add back $p$}
    \ket{0}_{s-1}\ket{x+a-p}_{n+1}\xmapsto{\mathsf{C\text{-}Q_{ADD}(p)}}&\ket{0}_{s-1}\ket{x+a\mod p}_{n}\ket{\mathbf{1}[x+a<p]}_1~.\\ \intertext{Now, as we have seen many times before, we have to find an operation that uncomputes $d=\mathbf{1}[x+a<p]$. 
    We notice that $\mathbf{1}[x+a<p]\equiv\mathbf{1}[x+a\mod p \geq a]$ since $x+a \mod p = x+a +(d-1)\cdot p$. Therefore, we can use a constant comparator $(\cdot < a)$ followed by a NOT to perform the required reset/uncomputation }
    \ket{0}_{s-1}\ket{x+a\mod p}_{n}\ket{\mathbf{1}[x+a<p]}&\xmapsto{\mathsf{Q_{COMP}(a)}}_1\\
    &\ket{0}_{s-1}\ket{x+a\mod p}_{n}\ket{\mathbf{1}[x+a<p] \oplus \mathbf{1}[x+a\mod p<a]}_1\\
    =&\ket{0}_{s-1}\ket{x+a\mod p}_{n}\ket{1} \\
    &\xmapsto{\mathsf{Q_{NOT}}}_1\ket{0}_{s}\ket{x+a\mod p}_{n}~.
\end{align*}   
\end{proof}

\subsection{Controlled modular addition by a constant}
\begin{definition}[Controlled modular addition by a constant]\label{def:abstract-controlled-modular-adder-by-a-const}
    Given $x,a,p \in \{0,1\}^n$ (with $a,p$ classically known), and a control $c\in \{0,1\}$ with $0\leq a,x<p $.
    We will denote controlled modular addition by a constant $a$ (with modulus $p$) with the following circuit
\begin{equation*}
        \includegraphics[width=0.4\textwidth]{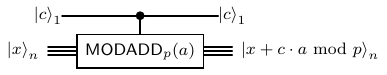}
\end{equation*}
\end{definition}
A simple way to perform a controlled modular addition by a constant is by using a quantum modular adder as per definition~\ref{def:modadder}.
\begin{theorem}[Controlled modular adder by a constant]~\label{thm:controlled-constant-modular-adder}
Let $\mathsf{Q_{MODADD_{p}}}$ be a quantum circuit that performs an $n$-bit quantum modular addition as per definition~\ref{def:modadder} using $s$ ancillas, and $r$ $\mathsf{Tof}$ gates. Then, there is a circuit $\mathsf{C\text{-}Q_{MODADD_{p}}}(a)$ implementing an $n$-bit controlled quantum modular addition by a constant  
as per definition~\ref{def:abstract-controlled-modular-adder-by-a-const} using $n+s$ ancilla qubits and $r$ $\mathsf{Tof}$ gates, and an additional $2|a|$ $\mathsf{CNOT}$ gates.
\end{theorem}
\begin{proof}
 The technique is similar to what we have used in many previous cases when an operation by a constant is involved.
 We first take an unused $n$ qubit register, and load the value $a$ (controlled on $c$) using $|a|$ $\mathsf{CNOT}$ gates
\begin{align*}
    \ket{0}_{s+n}\ket{x}_n\xmapsto{\mathsf{Load}}& \ket{0}_{s}\ket{c\cdot a}_n\ket{x}_n~.\\ \intertext{We now perform the required modular addition}
    \ket{0}_{s}\ket{c\cdot a}_n\ket{x}_n \xmapsto{\mathsf{Q_{MODADD}}_p(a)}& \ket{0}_{s}\ket{c\cdot a}_n\ket{x+c\cdot a\mod p}_n~.\\ \intertext{We have to finally unload the constant $a$}
    \ket{0}_{s}\ket{c\cdot a}_n\ket{x+c\cdot a\mod p}_n\xmapsto{\mathsf{Unload}}&\ket{0}_{s+n} \ket{x+c\cdot a\mod p}_n~.
\end{align*}
\end{proof}

The above construction works for any modular adder construction that performs the operation $\ket{x}_n\ket{y}_n\xmapsto{\mathsf{Q_{MODADD}}_p} \ket{x}_n\ket{y+x}_n$. However, if we look at the VBE architecture specifically, we can modify some of the subroutines to provide a different method for performing a controlled constant modular addition.
\begin{proposition}[Controlled modular adder by a constant - in VBE architecture~\cite{vedral1996quantum}]~\label{prp:controlled-constant-modular-adder-vbe-architecture}
Let
    \begin{itemize}
    \item $\mathsf{C\text{-}Q_{ADD}(a)}$ be a quantum circuit that performs a controlled constant addition using $s_{\mathsf{C\text{-}ADD}}$ ancillas, and $r_{\mathsf{C\text{-}ADD}}$ $\mathsf{Tof}$ gates;
    \item $\mathsf{Q_{COMP}}(p)$ be a quantum comparator using $s_{\mathsf{COMP}}$ ancillas, $r_{\mathsf{COMP}}$ $\mathsf{Tof}$ gates;
    \item $\mathsf{C\text{-}Q_{SUB}}(p)$ be a quantum circuit that performs a  (controlled) subtraction by a constant $p$ using $s_{\mathsf{C\text{-}SUB}}$ ancillas, and $r_{\mathsf{C\text{-}SUB}}$ $\mathsf{Tof}$ gates;
    \item $\mathsf{C\text{-}Q_{COMP}'}(a)$: be a controlled constant quantum comparator using ($s_{\mathsf{C\text{-}COMP}}'$ ancillas, $r_{\mathsf{C\text{-}COMP}}'$ $\mathsf{Tof}$ gates).
\end{itemize}
There is a circuit $\mathsf{C\text{-}Q_{MODADD_{p}}}(a)$ implementing an $n$-bit controlled quantum modular addition by a constant  
as per definition~\ref{def:abstract-controlled-modular-adder-by-a-const} using $r=(r_{\mathsf{{C\text{-}ADD}}}+r_{\mathsf{{COMP}}}+r_{\mathsf{{C\text{-}SUB}}}+r_{\mathsf{{C\text{-}COMP}}}')$ $\mathsf{Tof}$ gates and $s=2+\max(s_{\mathsf{C\text{-}ADD}},s_{\mathsf{COMP}},s_\mathsf{C\text{-}SUB},s_{\mathsf{C\text{-}COMP}}')$ ancillas.
\end{proposition}
\begin{proof}
The proof is very similar to the proof of the VBE modular adder architecture (proposition~\ref{prp:mod-add-first-vedral-architecture}).
We start with a controlled addition by a constant $a$ into our target
$$\ket{0}_{s}\ket{c}_1\ket{x}_n\xmapsto{\mathsf{C\text{-}Q_{ADD}}(a)}\ket{0}_{s-1}\ket{x+c\cdot a}_{n+1}~.$$
Using our comparator $\mathsf{Q_{COMP}}(p)$, we can perform a comparison of the previous sum with our modulus $p$
$$\ket{0}_{s-1}\ket{c}_1\ket{x+c\cdot a}_{n+1}\xmapsto{\mathsf{Q_{COMP}}(p)}\ket{0}_{s-2}\ket{c}_1\ket{x+c\cdot a}_{n+1}\ket{\mathbf{1}[x+c\cdot a <p]}_1~.$$
When the comparator's output is 0, i.e. when $x+c\cdot a\geq p$, we need to subtract the modulus $p$ to compute the correct sum $x+c\cdot a \mod p$. Therefore, we perform a controlled subtraction of $p$ in the register currently storing $x+c\cdot a$ (\textit{Note: We need to flip, using an $\mathsf{X}$ gate, the value of the comparison qubit before performing our controlled subtraction })
$$\ket{0}_{s-2}\ket{c}_1\ket{x+c\cdot a}_{n+1}\ket{\mathbf{1}[x+c\cdot a \geq p]}_1\xmapsto{\mathsf{C\text{-}Q_{SUB}(p)}}\ket{0}_{s-1}\ket{c}_1\ket{x}_n\ket{x+c\cdot a \mod p}_{n}\ket{\mathbf{1}[x+c\cdot a \geq p]}_1$$
We finally need to uncompute our previously computed output of the comparator. We notice that $x\cdot \mathbf{1}[x+c\cdot a \geq p] \equiv \cdot \mathbf{1}[x+c\cdot a  \mod p < c\cdot a]$ (similar to the proof of proposition~\ref{prp:controlled-modular-addition}), therefore, we can use $\mathsf{Q_{C\text{-}COMP}'}(a)$ on inputs $a$ and $(x+c\cdot a \mod p)$ to perform the final step 
\begin{align*}
    \ket{0}_{s-1}\ket{c}_1\ket{x}_n\ket{x+c\cdot a\mod p}_{n}\ket{\mathbf{1}[x+c\cdot a \geq p]}_1\\
    \xmapsto{\mathsf{Q_{C\text{-}COMP}'}(a)}\ket{0}_{s-1}\ket{c}_1\ket{x+c\cdot a\mod p}_{n}&\ket{\mathbf{1}[x+c\cdot a \geq p] \oplus c\cdot\mathbf{1}[x+c\cdot a\mod p < c\cdot a]}_1\\
    =\ket{0}_{s}\ket{c}_1\ket{x+c\cdot a\mod p}_{n}&~.
\end{align*}
\end{proof}
A $\mathsf{QFT}$-based controlled constant modular adder was proposed by Beauregard~\cite{beauregard2002circuit}, and works by combining the VBE architecture with Draper's $\mathsf{QFT}$ adder~\cite{draper2000addition} (proposition~\ref{prp:draper-orig-qft-adder}). We propose a slightly modified proof below.

\begin{proposition}[Controlled modular adder by a constant - Beauregard~\cite{beauregard2002circuit}]\label{prp:bouregard controlled modular addition by a constant}
There is a circuit $\mathsf{C\text{-}Q_{MODADD_{p}}}(a)$ implementing an $n$-bit controlled quantum modular addition by a constant as per definition~\ref{def:abstract-controlled-modular-adder-by-a-const} with a gate count of $3$ $\mathsf{QFT}$'s, $3$ $\mathsf{IQFT}$'s, $2$ $\mathsf{CNOT}$'s, $1$ $\mathsf{C\text{-}\Phi_{ADD}}(a)$, $1$ $\mathsf{\Phi_{ADD}}(a)$, $1$ $\mathsf{\Phi_{SUB}}(a)$, $1$ $\mathsf{C\text{-}\Phi_{SUB}(p)}$, $1$ $\mathsf{\Phi_{ADD}(p)}$, $1$ $\mathsf{\Phi_{SUB}(p)}$ and $2$ ancillas.
\end{proposition}
\begin{proof}

We use the following subroutines. Let
\begin{itemize}
    \item $\mathsf{C\text{-}Q_{ADD}(a)}$ be Draper's controlled adder by a constant (proposition~\ref{prp:controlled addition by a constant - draper}), using $0$ ancillas;
    \item $\mathsf{Q_{COMP}}(p)$ be a quantum comparator based on proposition~\ref{prp:draper-comp-by-a-classical-constant}, using $1$ ancilla;
    \item $\mathsf{C\text{-}Q_{SUB}}(p)$ be a quantum circuit that performs a  controlled subtraction by a constant based on proposition~\ref{prp:controlled addition by a constant - draper} using $0$ ancillas;
    \item $\mathsf{C\text{-}Q_{COMP}'}(a)$: be a controlled quantum comparator by a constant as derived in theorem~\ref{thm:controlled comparator by a classical constant CDKMP} using $1$ ancilla.
\end{itemize}
Note that adjacent $\mathsf{QFT}$s and $\mathsf{IQFT}$s can cancel out. The total gate count is therefore, $3$ $\mathsf{QFT}$'s, $3$ $\mathsf{IQFT}$'s, $2$ $\mathsf{CNOT}$'s, $1$ $\mathsf{C\text{-}\Phi_{ADD}}(a)$, $1$ $\mathsf{\Phi_{ADD}}(a)$, $1$ $\mathsf{\Phi_{SUB}}(a)$, $1$ $\mathsf{C\text{-}\Phi_{SUB}(p)}$, $1$ $\mathsf{\Phi_{ADD}(p)}$ and $1$ $\mathsf{\Phi_{SUB}(p)}$. The number of ancillas required is $s=2$
\end{proof}

In proposition~\ref{prp:bouregard controlled modular addition by a constant}, the modular addition circuit is based on a single controlled $\mathsf{QFT}$ adder. For completeness, we also share Beauregard's~\cite{beauregard2002circuit} original circuit in figure~\ref{fig:qft-modular-adder} which uses two controls, and also uses a slightly different set of gates..
\begin{figure}[!ht]
    \centering
    \includegraphics[width=1.0\textwidth]{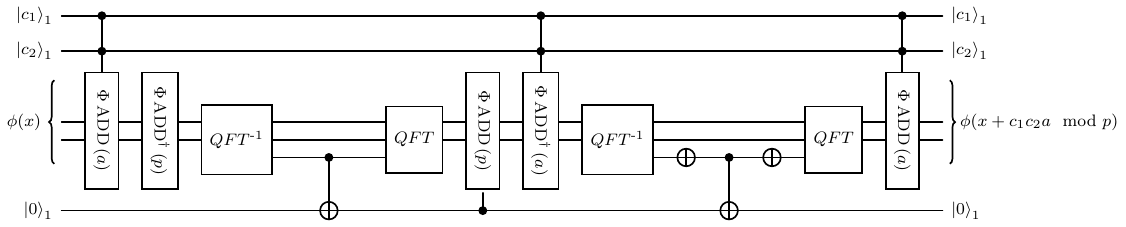}
    \caption{Quantum Fourier Transform modular adder by a constant. Here $\ket{\phi(x)}$ is made of $n+1$ qubits instead of $n$ qubits, in order to account for the initial overflow from adding $a$. }
    \label{fig:qft-modular-adder}
\end{figure}

\section{Measurement-based uncomputation speed-ups}\label{sec:MBU}

We are now prepared to present the proof of the MBU lemma and its applications to modular arithmetic.
The MBU algorithm for uncomputing a single qubit is the circuit that can be found in figure~\ref{fig:MBU}. Recall that MBU has already been introduced in~\cite{gidney2018halving} (for the special case of $\mathsf{Tof}$ gates) but here we formalize this method for the single-qubit uncomputation of any unitary. 
Attempts at generalization of MBU to perform multi-qubit computation has been already treated in the literature, see e.g. ~\cite{gidney2019windowed}. 

\begin{figure}[h!]
    \centering
        \includegraphics[width=0.7\textwidth]{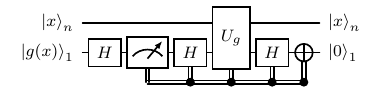}
\caption{A circuit for measurement-based uncomputing the unitary $U_g$.}\label{fig:MBU}
\end{figure}

\begin{lemma}[Measurement-based uncomputation lemma]
    Let $g:\,\{0,1\}^n\to\{0,1\}$ be a function, and let $U_g$ be a self-adjoint $(n+1)$-qubit unitary that implements $g$ as follows:
$$    U_g\ket{x}_A\ket{b}_G=\ket{x}_A\ket{b\oplus g(x)}_G.$$
There is an algorithm performing the mapping 
$$\sum_x\alpha_x\ket{x}_A\ket{g(x)}_G \mapsto \sum_x\alpha_x\ket{x}_A\ket{0}_G $$
        using with $1/2$ probability $U_g$, $2$ $\mathsf{H}$ gates, and $1$ $\mathsf{NOT}$ gate.
    It also involves one single-qubit computational basis measurement  and a single $\mathsf{H}$ gate.
\end{lemma}

\begin{proof}
The algorithm is described in figure~\ref{fig:MBU}. We first apply the Hadamard gate on the register $G$ and measure it in the computational basis. Since $g(x) \in \{0,1\}$, the resulting outcomes is $0$ or $1$ with equal probabilities. If we obtain $0$, we successfully uncomputed the register, leaving the registers in the state $\sum_x\alpha_x\ket{x}_A\ket{0}_G$. On the other hand, if we get the outcome $1$, the state at this point is $\sum_x\alpha_x(-1)^{g(x)}\ket{x}_A\ket{1}_G.$ We then apply another Hadamard to $G$, followed by $U_g$. This performs a phase kickback, leaving the quantum computer in the state
$$\sum_x\alpha_x(-1)^{g(x)}\ket{x}_A
\frac{\ket{0\oplus g(x)}_G-\ket{1\oplus g(x)}_G}{\sqrt{2}}
=
\sum_x\alpha_x(-1)^{g(x)\oplus g(x)}\ket{x}_A\ket{-}_G.$$
Applying on $G$ another Hadamard and a NOT gate results in the final state 
$\sum_x\alpha_x\ket{x}_A\ket{0}_G$.
\end{proof}

\subsection{Modular addition}\label{sec:MBUmodularaddition}

\begin{figure}[!ht]
    \centering
    \includegraphics[width=\textwidth]{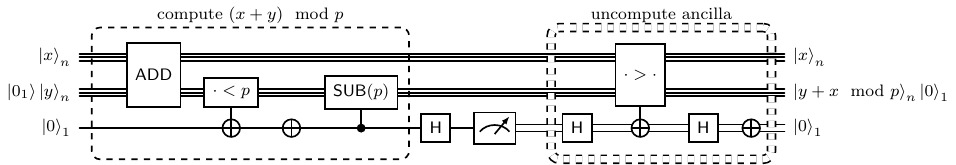}
    \caption{Implementation of a MBU modular addition. The double dotted lines represents the uncomputation of the ancilla, controlled on the measurement value.}
\end{figure}

\begin{theorem}[MBU modular adder - VBE architecture]\label{thm:mbu-modular-adder-vedral}
Let 
\begin{itemize}
    \item $\mathsf{Q_{ADD}}$ be a quantum circuit that performs a quantum addition using $s_{\mathsf{ADD}}$ ancillas, and $r_{\mathsf{ADD}}$ $\mathsf{Tof}$ gates;
    \item $\mathsf{Q_{COMP}}(p)$ be a quantum comparator using $s_{\mathsf{COMP}}$ ancillas, $r_{\mathsf{COMP}}$ $\mathsf{Tof}$ gates;
    \item $\mathsf{C\text{-}Q_{SUB}}(p)$ be a quantum circuit that performs a  (controlled) subtraction by a constant $p$ using $s_{\mathsf{C\text{-}SUB}}$ ancillas, and $r_{\mathsf{C\text{-}SUB}}$ $\mathsf{Tof}$ gates;
    \item $\mathsf{Q_{COMP}'}$: be a quantum comparator using ($s_{\mathsf{COMP}}'$ ancillas, $r_{\mathsf{COMP}}'$ $\mathsf{Tof}$ gates).
\end{itemize}
There is a circuit, $\mathsf{Q_{MODADD_\text{p}}}$, performing the $n$-bit modular addition operation as per definition~\ref{def:modadder} using $r=\left(r_{\mathsf{{ADD}}}+r_{\mathsf{{COMP}}}+r_{\mathsf{{C\text{-}SUB}}}+\left(r_{\mathsf{{COMP}}}'/2\right)\right)$ $\mathsf{Tof}$ gates in expectation and \[s=2+\max(s_{\mathsf{ADD}},s_{\mathsf{COMP}},s_\mathsf{C\text{-}SUB},s_{\mathsf{COMP}}')\] ancillas.
\end{theorem}
\begin{proof}
As seen in the modular adder - VBE architecture (proposition~\ref{prp:mod-add-first-vedral-architecture}), the first three subroutines perform the following mapping
\begin{align*}
\ket{0}_{s}\ket{x}_n\ket{y}_{n}\xmapsto{\mathsf{Q_{ADD}},\mathsf{Q_{COMP}}(p),\mathsf{C\text{-}Q_{SUB}(p)}}\ket{0}_{s-1}\ket{x}_n\ket{x+y \mod p}_{n}\ket{\mathbf{1}[x+y \geq p]}_1. 
\end{align*}
The comparator qubit storing $\mathbf{1}[x+y \geq p]$ is finally supposed to be uncomputed using $\mathsf{Q_{COMP}'}$. We can apply the MBU lemma on the comparator qubit and the uncomputation subroutine $\mathsf{Q_{COMP}'}$, thus halving the cost of $\mathsf{Q_{COMP}'}$ in expectation. Therefore, we finally get a cost of $r=\left(r_{\mathsf{{ADD}}}+r_{\mathsf{{COMP}}}+r_{\mathsf{{C\text{-}SUB}}}+\left(r_{\mathsf{{COMP}}}'/2\right)\right)$ $\mathsf{Tof}$ gates in expectation, and the same ancilla count at proposition~\ref{prp:mod-add-first-vedral-architecture}.
\end{proof}

\begin{theorem}[MBU modular adder - CDKPM]\label{thm:mbumodular-adder}
There is a circuit, $\mathsf{Q_{MODADD_\text{p}}}$, performing the $n$-bit modular addition operation as per definition~\ref{def:modadder} using
  $n+3$ ancillas and  $7n$ $\mathsf{Tof}$ gates in expectation.
\end{theorem}
\begin{proof}
For the CDKPM modular adder described in proposition~\ref{prp:modular adder with CDKPM}, the following subroutines are used:
\begin{itemize}
        \item $\mathsf{Q_{ADD}}$: CDKPM plain adder proposition~\ref{prp: CDKPM plain adder};
        \item $\mathsf{Q_{COMP}}(p)$: CDKPM half subtractor based constant comparator proposition~\ref{prp:comparator-cuccaro};
        \item $\mathsf{C\text{-}Q_{SUB}}(p)$: CDKPM plain adder, more concretely, we use $\mathsf{Q_{ADD}^{\dagger}}$ in proposition~\ref{prp:controlled addition by a constant - any}.
        \item $\mathsf{Q_{COMP}'}$: CDKPM half subtractor based comparator proposition~\ref{prp:comparator-cuccaro}.
    \end{itemize}
Now, applying theorem~\ref{thm:mbu-modular-adder-vedral}, we find that our expected $\mathsf{Tof}$ gate cost is $$r=\left(r_{\mathsf{{ADD}}}+r_{\mathsf{{COMP}}}+r_{\mathsf{{C\text{-}SUB}}}+\left(r_{\mathsf{{COMP}}}'/2\right)\right)=(2n+2n+2n+(2n/2))=7n$$ and the ancilla count is $s=n+3$.

\end{proof}

\begin{theorem}[MBU modular adder - Gidney]\label{thm: MBU modular adder gidney} 
There is a circuit, $\mathsf{Q_{MODADD_\text{p}}}$, performing the $n$-bit modular addition operation as per definition~\ref{def:modadder} using $2n+3$ ancillas,  $3.5n$ $\mathsf{Tof}$ gates in expectation.
\end{theorem}
\begin{proof}
Assume to access the following subroutines are used:
\begin{itemize}
        \item $\mathsf{Q_{ADD}}$: Gidney plain adder proposition~\ref{prp:gidneyAdder};
        \item $\mathsf{Q_{COMP}}(p)$: Gidney half subtractor based constant comparator proposition~\ref{prp:comparator using half gidney subtractor};
        \item $\mathsf{C\text{-}Q_{SUB}}(p)$: Gidney plain adder, more concretely, we use $\mathsf{Q_{ADD}^{\dagger}}$ in proposition~\ref{prp:controlled addition by a constant - any}.
        \item $\mathsf{Q_{COMP}'}$: Gidney half subtractor based comparator proposition~\ref{prp:comparator using half gidney subtractor}.
    \end{itemize}
Now, applying theorem~\ref{thm:mbu-modular-adder-vedral}, we find that our expected $\mathsf{Tof}$ gate cost is $$r=\left(r_{\mathsf{{ADD}}}+r_{\mathsf{{COMP}}}+r_{\mathsf{{C\text{-}SUB}}}+\left(r_{\mathsf{{COMP}}}'/2\right)\right)=(n+n+n+(n/2))=3.5n$$ and the ancilla count is $s=2n+3$    .

\end{proof}

\begin{theorem}[MBU modular adder - Gidney + CDKPM]\label{thm: modular adder gidney and cdkpm}
There is a circuit, $\mathsf{Q_{MODADD_\text{p}}}$, performing the $n$-bit modular addition operation as per definition~\ref{def:modadder} using $n+3$ ancillas,  $5.5n$ $\mathsf{Tof}$ gates in expectation.
\end{theorem}
\begin{proof}
For the Gidney+CDKPM modular adder described in theorem~\ref{thm: modular adder gidney and cdkpm}, the following subroutines are used:
\begin{itemize}
        \item $\mathsf{Q_{ADD}}$: Gidney plain adder proposition~\ref{prp:gidneyAdder};
        \item $\mathsf{Q_{COMP}}(p)$: CDKPM half subtractor based constant comparator proposition~\ref{prp:comparator-cuccaro};
        \item $\mathsf{C\text{-}Q_{SUB}}(p)$: CDKPM plain adder, more concretely, we use $\mathsf{Q_{ADD}^{\dagger}}$ in proposition~\ref{prp:controlled addition by a constant - any}.
        \item $\mathsf{Q_{COMP}'}$: Gidney half subtractor based comparator proposition~\ref{prp:comparator using half gidney subtractor}.
    \end{itemize}
Now, applying theorem~\ref{thm:mbu-modular-adder-vedral}, we find that our expected $\mathsf{Tof}$ gate cost is $$r=\left(r_{\mathsf{{ADD}}}+r_{\mathsf{{COMP}}}+r_{\mathsf{{C\text{-}SUB}}}+\left(r_{\mathsf{{COMP}}}'/2\right)\right)=(n+2n+2n+(n/2))=5.5n$$ and the ancilla count is $s=n+3$ .
\end{proof}
\begin{theorem}[MBU modular adder - Draper/Beauregard]\label{thm:QFTmodular-adder}
Let $x\in \{0,1\}^{n}, y\in \{0,1\}^{n}$ be two bit strings, and $p\in \{0,1\}^n$ classically known with $0\leq x,y < p$.
There is a circuit, $\mathsf{Q_{MODADD_\text{p}}}$, performing the modular addition operation as per definition~\ref{def:modadder} with an expected gate count of $2.5$ $\mathsf{QFT}$'s, $2.5$ $\mathsf{IQFT}$'s, $1.5$ $\mathsf{CNOT}$'s, $1.5$ $\mathsf{\Phi_{ADD}}$'s, $0.5$ $\mathsf{\Phi_{SUB}}$'s, $1$ $\mathsf{C\text{-}\Phi_{SUB}(p)}$, $1$ $\mathsf{\Phi_{ADD}(p)}$, $1$ $\mathsf{\Phi_{SUB}(p)}$ and
  $2$ ancillas.
\end{theorem}
\begin{proof}
We leverage proposition~\ref{prp:mod-add-first-vedral-architecture}, with the following subroutines:
    \begin{itemize}
        \item $\mathsf{Q_{ADD}}$: Draper plain adder corollary~\ref{cor:Draper-qft-adder};
        \item $\mathsf{Q_{COMP}}(p)$: Draper's comparator by a classical constant proposition~\ref{prp:draper-comp-by-a-classical-constant};
        \item $\mathsf{C\text{-}Q_{SUB}}(p)$: Draper adder, more concretely, we use $\mathsf{Q_{ADD}^{\dagger}}$ in proposition~\ref{prp:controlled addition by a constant - draper} and ~\ref{thm:controlled-draper-simplified} for controlled constant addition/subtraction;
        \item $\mathsf{Q_{COMP}'}$: Draper/Beauregard comparator proposition~\ref{prp:quantum-comparator-beauregard-draper}.
    \end{itemize}
The $\mathsf{IQFT}$ of $\mathsf{Q_{ADD}}$ cancels with the $\mathsf{QFT}$ of $\mathsf{Q_{COMP}}(p)$. We also find similar reductions for pairs the  $\mathsf{IQFT}$'s/$\mathsf{QFT}$'s of $\mathsf{Q_{COMP}}(p)/\mathsf{C\text{-}Q_{SUB}}(p)$ and $\mathsf{C\text{-}Q_{SUB}}(p)/\mathsf{Q_{COMP}'}$. The expected gate count is therefore, $2.5$ $\mathsf{QFT}$'s, $2.5$ $\mathsf{IQFT}$'s, $1.5$ $\mathsf{CNOT}$'s, $1.5$ $\mathsf{\Phi_{ADD}}$'s, $0.5$ $\mathsf{\Phi_{SUB}}$'s, $1$ $\mathsf{C\text{-}\Phi_{SUB}(p)}$, $1$ $\mathsf{\Phi_{ADD}(p)}$ and $1$ $\mathsf{\Phi_{SUB}(p)}$. The number of ancillas required is $s=2$.
\end{proof}

\subsection{Controlled modular addition}\label{sec:MBUcontrolledmodular}

\begin{theorem}[MBU controlled modular adder - VBE architecture]\label{thm:ctrMBUmodular-adder}
Let
\begin{itemize}
    \item $\mathsf{C\text{-}Q_{ADD}}$ be a quantum circuit that performs a quantum (controlled) addition using $s_{\mathsf{C\text{-}ADD}}$ ancillas, and $r_{\mathsf{C\text{-}ADD}}$ $\mathsf{Tof}$ gates;
    \item $\mathsf{Q_{COMP}}(p)$ be a quantum comparator by a constant $p$, using $s_{\mathsf{COMP}}$ ancillas, $r_{\mathsf{COMP}}$ $\mathsf{Tof}$ gates;
    \item $\mathsf{C\text{-}Q_{SUB}}(p)$ be a quantum circuit that performs a quantum subtraction by a constant $p$, using $s_{\mathsf{C\text{-}SUB}}$ ancillas, and $r_{\mathsf{C\text{-}SUB}}$ $\mathsf{Tof}$ gates;
    \item $\mathsf{C\text{-}Q_{COMP}'}$: be a controlled quantum comparator using ($s_{\mathsf{C\text{-}COMP}}'$ ancillas and $r_{\mathsf{C\text{-}COMP}}'$ $\mathsf{Tof}$ gates).
\end{itemize}
There is a circuit $\mathsf{C\text{-}Q_{MODADD_{p}}}$ implementing an $n$-bit controlled quantum modular addition as per definition~\ref{def:ctrlmodadder} using $r=\left(r_{\mathsf{{C\text{-}ADD}}}+r_{\mathsf{{COMP}}}+r_{\mathsf{{C\text{-}SUB}}}+\left(r_{\mathsf{{C\text{-}COMP}}}'\right/2)\right)$ $\mathsf{Tof}$ gates in expectation and $s=2+\max(s_{\mathsf{C\text{-}ADD}},s_{\mathsf{COMP}},s_\mathsf{C\text{-}SUB},s_{\mathsf{C\text{-}COMP}}')$ ancillas.
\end{theorem}
\begin{proof}
    As seen in the controlled modular adder - VBE architecture (proposition~\ref{prp:controlled-modular-addition}), the first three subroutines perform the following mapping
\begin{align*}
\ket{0}_{s}\ket{c}_1\ket{x}_n\ket{y}_{n}\xmapsto{\mathsf{C\text{-}Q_{ADD}},\mathsf{Q_{COMP}}(p),\mathsf{C\text{-}Q_{SUB}(p)}}\ket{0}_{s-1}\ket{c}_1\ket{x}_n\ket{c\cdot x+y \mod p}_{n}\ket{\mathbf{1}[c\cdot x+y \geq p]}_1~.
\end{align*}
The comparator qubit storing $\mathbf{1}[c\cdot x+y \geq p]$ is finally supposed to be uncomputed using $\mathsf{C\text{-}Q_{COMP}'}$. We can apply the MBU lemma here on the comparator qubit and the uncomputation subroutine $\mathsf{C\text{-}Q_{COMP}'}$, thus halving the cost of $\mathsf{C\text{-}Q_{COMP}'}$ in expectation.
Therefore, we finally get a cost of $r=\left(r_{\mathsf{{C\text{-}ADD}}}+r_{\mathsf{{COMP}}}+r_{\mathsf{{C\text{-}SUB}}}+\left(r_{\mathsf{{C\text{-}COMP}}}'/2\right)\right)$ $\mathsf{Tof}$ gates in expectation, and the same ancilla count as proposition~\ref{prp:controlled-modular-addition}.
\end{proof}

\begin{theorem}[MBU controlled modular adder - CDKPM]\label{thm:controlled-mbumodular-adder}
There is a circuit $\mathsf{C\text{-}Q_{MODADD_{p}}}$ implementing an $n$-bit controlled quantum modular addition as per definition~\ref{def:ctrlmodadder} using $n+3$ ancillas and $8n+0.5$ $\mathsf{Tof}$ gates in expectation.
\end{theorem}
\begin{proof}
For the controlled CDKPM modular adder described in proposition~\ref{prp:cdkpm-controlled-modular-adder}, the following subroutines are used:
    \begin{itemize}
    \item $\mathsf{C\text{-}Q_{ADD}}$ be a quantum circuit that performs a quantum (controlled) using a CDKPM adder with a single extra ancilla (theorem~\ref{thm:cuccaro-controlled-adder-1-extra-ancilla}) requiring $s_{\mathsf{C\text{-}ADD}}=2$ ancillas, and $r_{\mathsf{C\text{-}ADD}}=3n$ $\mathsf{Tof}$ gates;
    \item $\mathsf{Q_{COMP}}(p)$ be a quantum comparator by a constant $p$ using Controlled comparator - using half a CDKPM subtractor (theorem~\ref{thm:controlled comparator by a classical constant CDKMP}), requiring $s_{\mathsf{COMP}}=n+1$ ancillas, $r_{\mathsf{COMP}}=2n$ $\mathsf{Tof}$ gates;
    \item $\mathsf{C\text{-}Q_{SUB}}(p)$ be a quantum circuit that performs a quantum subtraction by a constant $p$ using the CDKPM adder in proposition~\ref{prp: addition by a constant - any}, requiring $s_{\mathsf{C\text{-}SUB}}=n+1$ ancillas, and $r_{\mathsf{C\text{-}SUB}}=2n$ $\mathsf{Tof}$ gates;
    \item $\mathsf{C\text{-}Q_{COMP}'}$ be the controlled quantum comparator described in proposition~\ref{prp:controlled comparator CDKMP}: Controlled comparator - using half a CDKPM subtractor. This subroutine requires $s_{\mathsf{C\text{-}COMP}}'=1$ ancilla, $r_{\mathsf{C\text{-}COMP}}'=2n+1$ $\mathsf{Tof}$ gates.
\end{itemize}

Now, applying theorem~\ref{thm:mbu-modular-adder-vedral}, we find that our expected $\mathsf{Tof}$ gate cost is $$r=\left(r_{\mathsf{{C\text{-}ADD}}}+r_{\mathsf{{COMP}}}+r_{\mathsf{{C\text{-}SUB}}}+\left(r_{\mathsf{{C\text{-}COMP}}}'/2\right)\right)=(3n+2n+2n+(2n+1)/2)=8n+0.5$$ and the ancilla count is $s=n+3$.
\end{proof}

\begin{theorem}[MBU controlled modular adder - Gidney]\label{thm:mbu-gidney-controlled-modular-adder}
There is a circuit $\mathsf{C\text{-}Q_{MODADD_{p}}}$ implementing an $n$-bit controlled quantum modular addition as per definition~\ref{def:ctrlmodadder} using $2n+3$ ancillas,  $4.5n+0.5$ $\mathsf{Tof}$ gates in expectation.
\end{theorem}
\begin{proof}
For the Gidney controlled modular adder described in proposition~\ref{prp:gidney-controlled-modular-adder}, the following subroutines are used:
    \begin{itemize}
    \item $\mathsf{C\text{-}Q_{ADD}}$ be a quantum circuit that performs a quantum (controlled) using a Gidney adder with a single extra ancilla (proposition~\ref{prp:gidney-controlled-adder-1-extra-ancilla}) requiring $s_{\mathsf{C\text{-}ADD}}=n+1$ ancillas, and $r_{\mathsf{C\text{-}ADD}}=2n$ $\mathsf{Tof}$ gates;
    \item $\mathsf{Q_{COMP}}(p)$ be a quantum comparator by a constant $p$ using Controlled comparator - using half a Gidney subtractor (proposition~\ref{prp:controlled comparator using half gidney subtractor}), requiring $s_{\mathsf{COMP}}=2n+1$ ancillas, $r_{\mathsf{COMP}}=n$ $\mathsf{Tof}$ gates;
    \item $\mathsf{C\text{-}Q_{SUB}}(p)$ be a quantum circuit that performs a controlled quantum subtraction by a constant $p$ using the Gidney adder in proposition~\ref{prp: addition by a constant - any}, requiring $s_{\mathsf{C\text{-}SUB}}=2n+1$ ancillas, and $r_{\mathsf{C\text{-}SUB}}=n$ $\mathsf{Tof}$ gates;
    \item $\mathsf{C\text{-}Q_{COMP}'}$ be the controlled quantum comparator described in proposition~\ref{prp:controlled comparator using half gidney subtractor}: controlled comparator - using half a Gidney subtractor. This subroutine requires $s_{\mathsf{C\text{-}COMP}}'=n+1$ ancillas, $r_{\mathsf{C\text{-}COMP}}'=n+1$ $\mathsf{Tof}$ gates.
\end{itemize}
Now, applying theorem~\ref{thm:mbu-modular-adder-vedral}, we find that our expected $\mathsf{Tof}$ gate cost is $$r=\left(r_{\mathsf{{C\text{-}ADD}}}+r_{\mathsf{{COMP}}}+r_{\mathsf{{C\text{-}SUB}}}+\left(r_{\mathsf{{C\text{-}COMP}}}'/2\right)\right)=(2n+n+n+(n+1)/2)=4.5n+0.5$$ and the ancilla count is $s=2n+3$.
\end{proof}

\subsection{Modular addition by a constant}\label{sec:MBUmodularbyaconstant}

\begin{theorem}[MBU modular addition by a constant - VBE architecture]\label{thm:mbu-modular-addition-by-a-constant}
Let
    \begin{itemize}
    \item $\mathsf{Q_{ADD}(a)}$ be a quantum circuit that performs a constant quantum addition using $s_{\mathsf{ADD}}$ ancillas, and $r_{\mathsf{ADD}}$ $\mathsf{Tof}$ gates;
    \item $\mathsf{Q_{COMP}}(p)$ be a quantum comparator using $s_{\mathsf{COMP}}$ ancillas, $r_{\mathsf{COMP}}$ $\mathsf{Tof}$ gates;
    \item $\mathsf{C\text{-}Q_{SUB}}(p)$ be a quantum circuit that performs a  (controlled) subtraction by a constant $p$ using $s_{\mathsf{C\text{-}SUB}}$ ancillas, and $r_{\mathsf{C\text{-}SUB}}$ $\mathsf{Tof}$ gates;
    \item $\mathsf{Q_{COMP}'}(a)$: be a constant quantum comparator using ($s_{\mathsf{COMP}}'$ ancillas, $r_{\mathsf{COMP}}'$ $\mathsf{Tof}$ gates).
\end{itemize}
there is a circuit $\mathsf{Q_{MODADD_{p}}}(a)$ implementing an $n$-bit quantum modular addition by a constant as per definition~\ref{def:abstract-modular-adder-by-a-const} using $r=\left(r_{\mathsf{{ADD}}}+r_{\mathsf{{COMP}}}+r_{\mathsf{{C\text{-}SUB}}}+\left(r_{\mathsf{{COMP}}}'/2\right)\right)$ $\mathsf{Tof}$ gates and $s=2+\max(s_{\mathsf{ADD}},s_{\mathsf{COMP}},s_\mathsf{C\text{-}SUB},s_{\mathsf{COMP}}')$ ancillas.
\end{theorem}
\begin{proof}
    As seen in the constant modular adder in VBE architecture (theorem~\ref{thm:constant-modular-adder-vbe-architecture}), the first three subroutines perform the following mapping
\begin{align*}
\ket{0}_{s}\ket{x}_n\xmapsto{\mathsf{Q_{ADD}}(a),\mathsf{Q_{COMP}}(p),\mathsf{C\text{-}Q_{SUB}}(p)}\ket{0}_{s-1}\ket{x}_n\ket{x+a \mod p}_{n}\ket{\mathbf{1}[x+a \geq p]}_1~.
\end{align*}
The comparator qubit storing $\ket{\mathbf{1}[x+a \geq p]}_1$ is finally supposed to be uncomputed using $\mathsf{Q_{COMP}'}(a)$. We can apply the MBU lemma here on the comparator qubit and the uncomputation subroutine $\mathsf{Q_{COMP}'}(a)$, thus halving the cost of $\mathsf{Q_{COMP}'}(a)$ in expectation.
Therefore, we finally get a cost of $r=\left(r_{\mathsf{{ADD}}}+r_{\mathsf{{COMP}}}+r_{\mathsf{{C\text{-}SUB}}}+\left(r_{\mathsf{{COMP}}}'/2\right)\right)$ $\mathsf{Tof}$ gates in expectation, and the same ancilla count at theorem~\ref{thm:constant-modular-adder-vbe-architecture}.
\end{proof}

\begin{theorem}[MBU modular adder by a constant - in Takahashi Architecture]~\label{thm:mbu-takahashi-cheap-const-mod-adder}
Let $x \in \{0,1\}^n$. Assume to know classically $p,a \in \{0,1\}^n$ with $0\leq a,x<p$. Also, assume we have access to the following subroutines:
    \begin{itemize}
    \item $\mathsf{Q_{SUB}(a)}$ be a quantum circuit that performs a constant quantum subtraction using $s_{\mathsf{SUB}}$ ancillas, and $r_{\mathsf{SUB}}$ $\mathsf{Tof}$ gates;
    \item $\mathsf{C\text{-}Q_{ADD}}(p)$ be a quantum circuit that performs a  (controlled) addition by a constant $p$ using $s_{\mathsf{C\text{-}ADD}}$ ancillas, and $r_{\mathsf{C\text{-}ADD}}$ $\mathsf{Tof}$ gates;
    \item $\mathsf{Q_{COMP}}(a)$ be a classical comparator using $s_{\mathsf{COMP}}$ ancillas, $r_{\mathsf{COMP}}$ $\mathsf{Tof}$ gates.
\end{itemize}
Then, there is a circuit $\mathsf{Q_{MODADD_{p}}}(a)$ implementing quantum modular addition by a constant  as per definition~\ref{def:abstract-modular-adder-by-a-const} using $r=\left(r_{\mathsf{{SUB}}}+r_{\mathsf{{C\text{-}ADD}}}+\left(r_{\mathsf{{COMP}}}/2\right)\right)$ $\mathsf{Tof}$ gates in expectation and $s=1+\max(s_{\mathsf{SUB}},s_\mathsf{C\text{-}ADD},s_{\mathsf{COMP}})$ ancillas.
\end{theorem}
\begin{proof}
 As seen in the Takashi constant modular adder (proposition~\ref{prp:takahashi-cheap-const-mod-adder}), the first two subroutines perform the following mapping
\begin{align*}
    \ket{0}_{s}\ket{x}_n\xmapsto{\mathsf{Q_{SUB}}(a),\mathsf{C\text{-}Q_{ADD}}(p)}&\ket{0}_{s-1}\ket{x}_n\ket{x+a \mod p}_{n}\ket{\mathbf{1}[x+a < p]}_1~.
\end{align*}
The comparator qubit storing $\ket{\mathbf{1}[x+a < p]}_1$ is finally supposed to be uncomputed using $\mathsf{Q_{COMP}}(a)$ and a NOT gate. We can apply the MBU lemma here on the comparator qubit and the uncomputation subroutine $\mathsf{Q_{COMP}}(a)$, thus halving the cost of $\mathsf{Q_{COMP}}(a)$ in expectation.
Therefore, we finally get a cost of $r=\left(r_{\mathsf{{SUB}}}+r_{\mathsf{{C\text{-}ADD}}}+\left(r_{\mathsf{{COMP}}}/2\right)\right)$ $\mathsf{Tof}$ gates in expectation, and the same ancilla count at proposition~\ref{prp:takahashi-cheap-const-mod-adder}.
\end{proof}

\subsection{Controlled modular addition by a constant}\label{sec:MBUcontrolledmodularbyaconstant}

\begin{theorem}[MBU controlled modular adder by a constant - VBE architecture]\label{thm:mbu-controlled-constant-modular-adder-vbe-architecture}
Let
    \begin{itemize}
    \item $\mathsf{C\text{-}Q_{ADD}(a)}$ be a quantum circuit that performs a controlled addition by a constant using $s_{\mathsf{C\text{-}ADD}}$ ancillas, and $r_{\mathsf{C\text{-}ADD}}$ $\mathsf{Tof}$ gates;
    \item $\mathsf{Q_{COMP}}(p)$ be a quantum comparator using $s_{\mathsf{COMP}}$ ancillas, $r_{\mathsf{COMP}}$ $\mathsf{Tof}$ gates;
    \item $\mathsf{C\text{-}Q_{SUB}}(p)$ be a quantum circuit that performs a  (controlled) subtraction by a constant $p$ using $s_{\mathsf{C\text{-}SUB}}$ ancillas, and $r_{\mathsf{C\text{-}SUB}}$ $\mathsf{Tof}$ gates;
    \item $\mathsf{C\text{-}Q_{COMP}'}(a)$: be a controlled constant quantum comparator using ($s_{\mathsf{C\text{-}COMP}}'$ ancillas, $r_{\mathsf{C\text{-}COMP}}'$ $\mathsf{Tof}$ gates).
\end{itemize}
There is a circuit $\mathsf{C\text{-}Q_{MODADD_{p}}}(a)$ implementing an $n$-bit controlled quantum modular addition by a constant as per definition~\ref{def:abstract-controlled-modular-adder-by-a-const} using $r=\left(r_{\mathsf{{C\text{-}ADD}}}+r_{\mathsf{{COMP}}}+r_{\mathsf{{C\text{-}SUB}}}+\left(r_{\mathsf{{C\text{-}COMP}}}'/2\right)\right)$ $\mathsf{Tof}$ gates and $s=2+\max(s_{\mathsf{C\text{-}ADD}},s_{\mathsf{COMP}},n+s_\mathsf{C\text{-}SUB},s_{\mathsf{C\text{-}COMP}}')$ ancillas.   
\end{theorem}
\begin{proof}
As seen in the controlled modular adder by a constant in VBE architecture (proposition~\ref{prp:controlled-constant-modular-adder-vbe-architecture}), the first three subroutines perform the following mapping
\begin{align*}
\ket{0}_{s}\ket{c}_1\ket{x}_n\xmapsto{\mathsf{C\text{-}Q_{ADD}}(a),\mathsf{Q_{COMP}}(p),\mathsf{C\text{-}Q_{SUB}}(p)}\ket{0}_{s-1}\ket{c}_1\ket{x}_n\ket{x+c\cdot a \mod p}_{n}\ket{\mathbf{1}[x+c\cdot a \geq p]}_1~.
\end{align*}
The comparator qubit storing $\ket{\mathbf{1}[x+c\cdot a \geq p]}_1$ is finally supposed to be uncomputed using $\mathsf{C\text{-}Q_{COMP}'}(a)$. We can apply the MBU lemma here on the comparator qubit and the uncomputation subroutine $\mathsf{C\text{-}Q_{COMP}'}(a)$, thus halving the cost of $\mathsf{C\text{-}Q_{COMP}'}(a)$ in expectation.
Therefore, we finally get a cost of $r=\left(r_{\mathsf{{C\text{-}ADD}}}+r_{\mathsf{{COMP}}}+r_{\mathsf{{C\text{-}SUB}}}+\left(r_{\mathsf{{C\text{-}COMP}}}'/2\right)\right)$ $\mathsf{Tof}$ gates in expectation, and the same ancilla count at proposition~\ref{prp:controlled-constant-modular-adder-vbe-architecture}.
\end{proof}

\subsection{Two-sided comparison}
The following result (stating directly the complexity with and without using MBU) checks if a register $\ket{x}$ has a value in the range specified by two other quantum registers $\ket{y}$ and $\ket{z}$. Compared to a non-MBU implementation, we save $25\%$ for the $\mathsf{Tof}$ gate cost.

\begin{theorem}[Two-sided comparator]~\label{thm:two-sided-comp}
Let $x,y,z \in \{0,1\}^n$, and $t \in \{0,1\}$ Let us assume we have access to $\mathsf{Q_{COMP}}$, with
$$\ket{x}_n\ket{y}_n\ket{t}_1 \xmapsto{\mathsf{Q_{COMP}}}\ket{x}_n\ket{y}_n\ket{t \oplus \mathbf{1}[x<y]}_1$$ using $s_{\mathsf{COMP}}$ ancillas and $r_{\mathsf{COMP}}$ $\mathsf{Tof}$ gates.
Let us assume we also have access to $\mathsf{C\text{-}Q'_{COMP}}$, with
$$\ket{c}_1\ket{x}_n\ket{y}_n\ket{t}_1 \xmapsto{\mathsf{C\text{-}Q'_{COMP}}}\ket{c}_1\ket{x}_n\ket{y}_n\ket{t \oplus c\cdot \mathbf{1}[x < y]}_1$$ using $s'_{\mathsf{C\text{-}COMP}}$ ancillas and $r'_{\mathsf{C\text{-}COMP}}$ $\mathsf{Tof}$ gates. Then, there is a circuit that performs the operation $\mathsf{Q_{IN\_RANGE}}$
$$\ket{x}_n\ket{y}_n\ket{z}_n\ket{t}_1\xmapsto{\mathsf{Q_{IN\_RANGE}}}\ket{x}_n\ket{y}_n\ket{z}_n\ket{t \oplus \mathbf{1}[x \in (y,z)]}_1$$ using $s=1+s_{\mathsf{COMP}}$ ancillas and $r=1.5r_\mathsf{COMP} + r'_\mathsf{C\text{-}COMP}$ $\mathsf{Tof}$ gates.
\end{theorem}
\begin{proof}
    We begin by first checking if $y < x$
    \begin{align*}
       \ket{0}_s\ket{x}_n\ket{y}_n\ket{z}_n\ket{t}_1\xmapsto{\mathsf{Q_{COMP}}}& \ket{0}_{s-1}\ket{x}_n\ket{y}_n\ket{z}_n\ket{\mathbf{1}[y < x]}\ket{t}_1\\ \intertext{Next we compute whether $x<z$ controlled on $\mathsf{1}[y<x]$, and store the output in our target. We know that $\mathbf{1}[x \in (y,z)]\equiv(\mathbf{1}[y< x]\cdot \mathbf{1}[x< z])$, and since we have both the comparator values stored in memory}
       \ket{0}_{s-1}\ket{x}_n\ket{y}_n\ket{z}_n\ket{\mathbf{1}[y < x]}\ket{t}_1\xmapsto{\mathsf{Q'_{C\text{-}COMP}}}\ket{0}_{s-1}\ket{x}_n&\ket{y}_n\ket{z}_n\ket{\mathbf{1}[y < x]}_1\ket{t\oplus \mathbf{1}[x \in (y,z)]}_1~.\\ \intertext{We finally uncompute the intermediate value $\mathsf{1}[y<x]$ using $\mathsf{Q_{COMP}}$}
       \ket{0}_{s-1}\ket{x}_n\ket{y}_n\ket{z}_n\ket{\mathbf{1}[y < x]}_1\ket{t\oplus \mathbf{1}[x \in (y,z)]}_1\xmapsto{\mathsf{Q_{COMP}}}&\ket{0}_{s}\ket{x}_n\ket{y}_n\ket{z}_n\ket{t\oplus (\mathbf{1}[y < x]\cdot \mathbf{1}[x < z])}_1 ~.
    \end{align*}
    Therefore, we get a cost of $r=2r_\mathsf{COMP} + r'_\mathsf{C\text{-}COMP}$. This cost can be decreased by applying the MBU lemma on the $\mathsf{Q_{COMP}}$ uncomputing $\mathbf{1}[y<x]$, thus leading to a cost of $r=1.5r_\mathsf{COMP} + r'_\mathsf{C\text{-}COMP}$.
\end{proof}

\subsection*{Acknowledgement}
This work started when AS was working at Inveriant Pte. Ltd. AS is supported by Innovate UK under grant \emph{10004359}.
AMM has received funding from the European Union’s \emph{Horizon 2020 research and innovation programme} under the grant agreement \emph{N°101034324} and has been partially supported by the \emph{Italian Group for Algebraic and Geometric Structures and their Application} (GNSAGA–INdAM).
This work is supported by the National Research Foundation, Singapore, and A*STAR under its CQT Bridging Grant. We also acknowledge funding from the Quantum Engineering Programme (QEP 2.0) under grants \emph{NRF2021-QEP2-02-P05} and \emph{NRF2021-QEP2-02-P01}. We thank Anupam Chattopadhyay and Zeyu Fan for useful discussions. AMM and AL thank Niels Benedikter for the hospitality at Università degli Studi di Milano.

\printbibliography

\appendix
\section{Background on binary arithmetics}\label{Sec:binaryarithmetics}
In this appendix, we briefly survey some standard results on bit string arithmetics. In particular, we recall how bit strings can be interpreted as unsigned or signed integers.
We refer to the section \ref{subsec:Conventions} for our notation on bit strings arithmetics.

    \begin{proposition}\label{prop:subtraction-2complement}
        The bit string subtraction can be expressed in terms of $2$'s complement as
        $$ \mathbf{x}-\mathbf{y} = \mathbf{x} + \overline{\overline{\mathbf{y}}}~.$$
    \end{proposition}
    \begin{proof}
        According to definition~\ref{def:stringAddition} and~\ref{def:string2complement}, the string on the right-hand side is defined bitwise as
        $ r_k=x_k\oplus y_k  \oplus 1 \oplus z_k $,
        where $z_k$ is defined recursively as follows
        \begin{align}
            z_0 & = 1 \notag \\
            z_{i+1} & = \maj(x_i, y_i\oplus 1, z_i) ~. \notag
        \end{align}
        The thesis follows by showing inductively that $z_k=b_k\oplus 1$, where $b_k$ is the borrowing defined in definition~\ref{def:stringSubtraction}.
        Namely, assume $z_k = b_k\oplus 1$, then 
        \begin{align*}
            z_{k+1} &=~
            x_k y_k \oplus x_k y_k \oplus x_k z_k \oplus y_k z_k \oplus x_k \oplus z_k 
            \\
            &=~
            x_k y_k \oplus x_k y_k \oplus x_k (b_k\oplus 1) \oplus y_k (b_k\oplus 1) \oplus x_k \oplus b_k \oplus 1 
            \\
            &=~
            x_k y_k \oplus x_k y_k \oplus x_k b_k \oplus y_k b_k \oplus y_k \oplus b_k \oplus 1 
                        \\
            &=~
            b_{k+1}\oplus 1 ~.
        \end{align*}
    \end{proof}

    \begin{remark}[Binary encoding of a positive integer]\label{rem:unsignedencoding}
        Any bit string can be understood as the binary representation of a given positive integer number via the following bijective function
        \begin{align*}
            \{0,1\}^n &\leftrightarrow \{0,1,\dots,2^n-1\}\\
            \mathbf{x} &\mapsto  \sum_{k=0}^{n-1} x_k 2^k
        \end{align*}
        If not stated differently, we shall often equate the decimal and binary representation of a given positive integer number.
        Note that bit string addition is precisely built in such a way that the string $\mathbf{x}+\mathbf{y}$ corresponds to the positive natural number $x+y$; the extra bit is added to take into account the overflow of the carry when $(x+y)\geq 2^n$.
        Moreover, one can easily check that $\mathbf{x}+\overline{\mathbf{x}} = \mathbf(2^n-1)$.
    \end{remark}
    \begin{proposition}\label{prop:comparation-as-subtraction}
        Given $x,y \in  \{0,1,\dots,2^n-1\}$ two positive integer, one has that $(\mathbf{x}-\mathbf{y})_{n+1} = (x < y)$.
    \end{proposition}
    \begin{proof}
        Assume that $x_j= y_j$ for any $k < j < n$ and the $k$-th bit is the most significant one such that $x_k\neq y_k$.
        According to definition~\ref{def:stringSubtraction}, one has $d_j= b_j = b_{k+1}$ for any $k < j \leq n$.
        Therefore if $x>y$ one must have $x_k=1$ and $y_k=0$, thence $d_{n+1}=b_{k+1}=0$. On the contrary, one has $d_{n+1}=b_{k+1}=1$.
    \end{proof}
    
    \begin{remark}[Binary encoding of signed integer]\label{rem:2scomplement}
        Any signed integer number can be encoded as a bit string via the following bijective function
        \begin{align*}
            \{0,1\}^{n+1} &\leftrightarrow \{-2^n,\dots,0,1,\dots,2^n-1\}\\[-0.5px]
            \mathbf{x} &\mapsto  -x_{n}2^n+\sum_{k=0}^{n-1} x_k 2^k \\[-0.75px]
            \mathbf{y|_2} &\mapsfrom y
        \end{align*}
        where $\mathbf{y|_{n+1}}$ denotes the string encoding the signed integer $y$ in the \emph{2's complement}, i.e.
        \begin{equation}
            \mathbf{y|_{n+1}}=
            \begin{cases}
                \mathbf{y} & y\geq 0 \\
                \overline{\overline{\mathbf{|y|}}} & y <0~,
            \end{cases}
        \end{equation}
            and $\mathbf{y}$ and $\mathbf{|y|}$ are the string encoding the positive numbers $y$ and $|y|$ according to remark~\ref{rem:unsignedencoding}.
        According to this convention, the most significant bit of the string encodes $\text{sign}(y)$.
    \end{remark}
    \begin{proposition}\label{prop:signeddifference}
        For any two $n$ bits strings $\mathbf{x},\mathbf{y}$ with corresponding unsigned integers $x$ and $y$, one has
        \begin{equation}
            \mathbf{x} - \mathbf{y} = \mathbf{(x-y)|_{n+1}}~,            
        \end{equation}
        where the left-hand side is the $n$ strings subtraction and the right-hand side is the string corresponding to the signed integer $(x-y)$ in the $2$'s complement convention.
    \end{proposition}
    \begin{proof}
        One has to check two separate cases. First case, when $x\geq y$, one has that $(x-y)$ is a positive number and $ \mathbf{x} - \mathbf{y}$ correspond to the string $\mathbf{(x-y)}$. In the other case, when $x<y$, one gets from the definition of $n$ strings subtraction that
        \begin{align*}
            \mathbf{x} - \mathbf{y} =&~ \mathbf{x} + \overline{\overline{\mathbf{y}}} 
            \\
            =&~ \overline{\mathbf{y}} + \mathbf{x} + \mathbf{1}
            \\
            =&~
            \overline{\mathbf{(y-x)}} + 1
            \\
            =&~
            \overline{\overline{\mathbf{(y-x)}}}
            \\
            =&~
            \overline{\overline{\mathbf{|x-y|}}}~.
        \end{align*}
    \end{proof}

    The upshot is that given any classical circuit implementing $n$ strings addition, one can implement signed integer arithmetic, encoding them in $2$'s complement.
    \begin{proposition}
        For any two $n$ bits strings $\mathbf{x|_{n}}, \mathbf{y|_{n}}$, corresponding to the signed integers $x$ and $y$ in  the $2$'s complement convention, one has
        $
            \mathbf{x|_{n}} + \mathbf{y|_{n}} = \mathbf{(x+y)|_{n+1}}
        $,
         where the the left-hand side is the $n$ strings addition and the right-hand side is the string corresponding to the signed integer $(x+y)$ in the $2$'s complement convention.
    \end{proposition}
    \begin{proof}
        One has to consider four separate cases. In the first case $x,y\geq 0$, one gets the usual unsigned addition. When $x\geq0$ and $y<0$ one has
        \begin{align}
            \text{lhs}=
            \mathbf{x} + \overline{\overline{\mathbf{|y|}}} = \mathbf{x} - \mathbf{|y|} = \mathbf{(x-|y|)|_{n+1}}
        \end{align}
        employing proposition~\ref{prop:signeddifference} in the last equality. Similarly, when $x<0$ and $y\geq 0$, one has
        \begin{align}
            \text{lhs}=
            \overline{\overline{\mathbf{|x|}}} + \mathbf{y} =
            \mathbf{|y|}-\mathbf{x} =
            \mathbf{(y-|x|)|_{n+1}}
            ~.
        \end{align}
         At last, when $x,y<0$, one notice that
        \begin{align}
            \text{lhs}=
            \overline{\overline{\mathbf{|x|}}} + \overline{\overline{\mathbf{|y|}}} =
            {\overline{\mathbf{|x|}}} + \overline{\overline{\mathbf{|y|}}} +1 =
            \overline{\overline{\mathbf{|x|}-\overline{\overline{|y|}}}} =
            \overline{\overline{\mathbf{|x|+|y|}}}
            ~.
        \end{align}
    \end{proof}

\end{document}